\documentclass[a4paper,10pt]{article}
\usepackage{graphicx}
\usepackage{amsmath}
\usepackage{citesort}
\usepackage[sort&compress,numbers]{natbib}
\usepackage{amsfonts}
\usepackage{amssymb}
\usepackage{amsmath}
\usepackage{latexsym}
\usepackage{color}
\usepackage{ntheorem}
\usepackage{braket}
\usepackage[colorlinks=true]{hyperref}
\usepackage{authblk}
\def\volform{\mbox{$\eta$}}

\newcounter{mnotecount}[section]

\renewcommand{\themnotecount}{\thesection.\arabic{mnotecount}}
\newcommand{\mnote}[1]%{}
{\protect{\stepcounter{mnotecount}}$^{\mbox{\footnotesize
$%\!\!\!\!\!\!\,
\bullet$\themnotecount}}$ \marginpar{%\color{red}%
\raggedright\tiny\em
$\!\!\!\!\!\!\,\bullet$\themnotecount: #1} }

\newtheorem{theorem}{\sc  Theorem\rm}[section]

\newtheorem{corollary}[theorem]{\sc  Corollary\rm}

\newtheorem{definition}[theorem]{\sc  Definition\rm}

\newtheorem{lemma}[theorem]{\sc Lemma\rm}
\newtheorem{Lemma}[theorem]{\sc Lemma\rm}

\newtheorem{proposition}[theorem]{\sc Proposition\rm}

\newtheorem{remark}[theorem]{\sc Remark\rm}

\newcommand{\ol}[1]{\overline{#1}{}}

\newcommand{\jlcax}[1]{}
\newcommand{\eean}{\nonumber\end{eqnarray}}

\newcommand{\F}{{\bf F}}
\newcommand{\kk}[1]{}%{\mnote{{\bf If we consider the KK case:} #1}}

\newcommand{\beq}{\begin{equation}}

\newcommand{\ep}{\epsilon}

\newcommand{\FS}       %{F_1} %
                  {F}
\newcommand{\HS} %{F_2}
       {H_{\mbox{\scriptsize volume}}}

{\ptc{this should be removed in the oberwolfach version}}%

\newcommand{\eeal}[1]{\label{#1}\end{eqnarray}}
\newcommand{\bed}{\begin{deqarr}}
\newcommand{\eed}{\end{deqarr}}
\newcommand{\bedl}[1]{\begin{deqarr}\label{#1}}
\newcommand{\eedl}[2]{\arrlabel{#1}\label{#2}\end{deqarr}}

\newcommand{\bel}[1]{\begin{equation}\label{#1}}
\newcommand{\bea}{\begin{eqnarray}}
\newcommand{\bean}{\begin{eqnarray}\nonumber}
\newcommand{\beal}[1]{\begin{eqnarray}\label{#1}}
\newcommand{\eea}{\end{eqnarray}}

\newcommand{\qed}{\hfill $\Box$ \medskip}

\newcommand{\be}{\begin{equation}}
\newcommand{\eeq}{\end{equation}}
\newcommand{\ee}{\end{equation}}
\newcommand{\beqa}{\begin{eqnarray}}
\newcommand{\eeqa}{\end{eqnarray}}
\newcommand{\beqan}{\begin{eqnarray*}}
\newcommand{\eeqan}{\end{eqnarray*}}
\newcommand{\ba}{\begin{array}}
\newcommand{\ea}{\end{array}}

\newcommand{\mcM}{{\mycal M}}

\newcommand{\scri}{{\mycal I}}%

\newcommand{\warn}[1]%{}%{}
{\protect{\stepcounter{mnotecount}}$^{\mbox{\footnotesize
$%\!\!\!\!\!\!\,
\bullet$\themnotecount}}$ \marginpar{%\color{red}%
\raggedright\tiny\em
$\!\!\!\!\!\!\,\bullet$\themnotecount: {\bf Warning:} #1} }

\newcommand{\eq}[1]{(\ref{#1})}

\newcommand{\ptc}[1]{\mnote{{\bf ptc:}#1}}

\newcommand{\mcL}{{\mycal L}}

\newcommand{\beqar}{\begin{deqarr}}
\newcommand{\eeqar}{\end{deqarr}}

\newcommand{\beaa}{\begin{eqnarray*}}
\newcommand{\eeaa}{\end{eqnarray*}}

\def\M{\mathcal M}
\def\J{{\cal Z}}
\def\la{\left \langle}
\def\ran{\right \rangle}

\DeclareFontFamily{OT1}{rsfs}{}
\DeclareFontShape{OT1}{rsfs}{m}{n}{ <-7> rsfs5 <7-10> rsfs7 <10-> rsfs10}{}
\DeclareMathAlphabet{\mycal}{OT1}{rsfs}{m}{n}

{\catcode `\@=11 \global\let\AddToReset=\@addtoreset}
\AddToReset{equation}{section}

{\catcode `\@=11 \global\let\AddToReset=\@addtoreset}
\AddToReset{figure}{section}

{\catcode `\@=11 \global\let\AddToReset=\@addtoreset}
\AddToReset{table}{section}

\newcommand{\Dconst}{C_{\mathrm{el}}}
\newcommand{\Cconst}{C_{\mathrm{mag}}}
\newcommand{\CconstL}{\sqrt{\frac{\Lambda}{3}}C_{\mathrm{mag}}}

\begin{document}

\title{Characterization of (asymptotically) Kerr-de Sitter-like spacetimes at null infinity%
%\tim{changed}\jose{OK with me}
\thanks{Preprint UWThPh-2016-5. }
}
\author[1]{Marc Mars}
\author[2]{Tim-Torben Paetz}
\author[3]{Jos\'e M. M. Senovilla}
\author[4]{Walter Simon}
\affil[1]{Instituto de F\'isica Fundamental y Matem\'aticas, Universidad de Salamanca, Plaza de la Merced s/n, 37008 Salamanca, Spain}
\affil[2,4]{Gravitational Physics, University of Vienna, Boltzmanngasse 5, 1090 Vienna, Austria}
\affil[3]{F\'isica Te\'orica, Universidad del Pa\'is Vasco, Apartado 644, 48080 Bilbao, Spain}
\affil[4]{Institute of Physics, Jagiellonian University, Lojasiewicza 11, 30-348
Krak\'ow, Poland}

\maketitle

\vspace{-0.2em}

\begin{abstract}
We investigate  solutions $(\mcM, g)$ to  Einstein's vacuum field equations with positive cosmological constant $\Lambda$ which admit
a smooth past null infinity $\scri^-$ \`a la Penrose and a Killing vector field whose associated  Mars-Simon tensor (MST) vanishes.
The main purpose of this work is to provide a characterization of these spacetimes in terms of their Cauchy data
on~$\scri^-$.

Along the way, we also study spacetimes for which the MST does not vanish. In that case there is an ambiguity in its definition 
which is captured by  a scalar function $Q$. We analyze properties of the MST for different choices of $Q$.
In doing so, we are led to a definition of ``asymptotically Kerr-de Sitter-like spacetimes'', which we also characterize in terms of their 
asymptotic data on $\scri^-$.
\end{abstract}

\noindent
\hspace{2.1em} PACS numbers: 02.30.Jr, 04.20.Ex, 04.20.Ha,

\tableofcontents

\section{Introduction}

This is the first in a series of at least two papers \cite{mpss} in which we (resp.\ some of us) 
analyze the asymptotic structure, and a certain initial value problem, for vacuum solutions of 
Einstein's equations
\begin{equation}
R_{\mu\nu}\,=\,\Lambda g_{\mu\nu}
\label{LambdaVac}
\end{equation}
on a 4-dimensional spacetime $(\mcM, g)$, 
where $g$ is smooth and  $\Lambda$ is a (``cosmological'') constant.
We focus on the case $\Lambda > 0$ but compare occasionally with $\Lambda = 0$. 
Spacetime indices are greek, 
while coordinates in $1+3$ splits are denoted by $\{x^{\alpha}\} = \{t, x^i \}$ 
(rather than  $\{x^0, x^i \}$), with corresponding tensorial indices.
 Our conventions for the signature, 
the curvature tensor $R_{\mu\nu\sigma}{}^{\kappa}$, the Weyl tensor  $C_{\mu\nu\sigma}{}^{\kappa}$, 
the Ricci tensor $R_{\mu\sigma}$ and the scalar curvature $R$ follow e.g.\ \cite{Wald}. 
The Levi-Civita connection of $g$ is denoted by $\nabla$.

The setting of our work is an asymptotic structure 
 \emph{\`a la Penrose} \cite{p2, F_lambda}. 
By that we mean that an appropriate conformal rescaling of  $(\mcM, g)$
\begin{equation}
 g \mapsto \widetilde g = \Theta^2 g\;, \quad  \mcM  \overset{\phi}{\hookrightarrow} \widetilde{\mcM\enspace}\hspace{-0.5em} \;, 
\quad \Theta|_{\phi(\mcM)}>0\;,
\end{equation}
leads to an \emph{unphysical spacetime} $(\widetilde{\mcM\enspace}\hspace{-0.5em} ,\widetilde g)$ 
which admits a representation of null infinity
%\jose{displayed and little rewriting later to avoid repetitions}
$$
\scri=\{ \Theta=0\,, \; \mathrm{d}\Theta \ne 0\}\cap \partial\phi(\mcM)$$ 
through which the unphysical metric $\widetilde g$ 
and the conformal factor $\Theta$ can be smoothly extended.
$\scri$ is a smooth hypersurface which consists of two (not necessarily
connected) subsets: 
Future and past null infinity, distinguished by the absence of endpoints of past
or future causal curves contained in $(\mcM,g)$, respectively. 
%and similarly for past null infinity. These are usually denoted by $\scri^+$ and $\scri^-$.
In this paper we will normally denote by $\scri^-$ and $\scri^+$ chosen connected
components of past and future null infinity, respectively. 
Clearly, all initial value results in this paper starting from
$\scri^-$ have obvious  
 ``final value counterparts'' obtained via replacing $\scri^-$ by
$\scri^+$, ``future'' by ``past'', etc. 
We will implicitly identify $\mcM$ with its image $\phi(\mcM)\subset \widetilde{\mcM\enspace}$, 
so that we can write $\widetilde g = \Theta^2 g$. 
Indices of physical and unphysical fields will be raised and lowered with $g$
and $\widetilde g$, respectively.

 In this setting, Friedrich \cite{F_lambda, F2} has shown that, in terms of
 suitable variables, the field equations  (\ref{LambdaVac})  become a regular, symmetric hyperbolic system on 
$(\widetilde{\mcM\enspace}\hspace{-0.5em} ,\widetilde g)$. 
We recall these ``metric conformal field equations'' (MCFE) in Sect
\ref{Mars-Simon_conf}. An important member of the MCFE is the rescaled Weyl tensor 
\begin{equation}
\label{dten}
\widetilde d_{\alpha\beta\gamma}{}^{\delta} := \Theta^{-1}C_{\alpha\beta\gamma}{}^{\delta}.
\end{equation}

and key properties of $C_{\alpha\beta\gamma\delta}$ and
$\widetilde d_{\alpha\beta\gamma\delta}$ are the following: 
%\tim{rewordings}

\begin{enumerate}
\item[I.] $C_{\alpha\beta\gamma}{}^{\delta}$ vanishes on $\scri$, whence
$\widetilde d_{\alpha\beta\gamma}{}^{\delta}$ extends regularly to $\scri$. 
\item[II.] $C_{\alpha\beta\gamma}{}^{\delta}$ 
satisfies a regular, linear, homogeneous symmetric hyperbolic system   on $(\mcM, g)$. 
\item[III.] $\widetilde d_{\alpha\beta\gamma}{}^{\delta}$ 
satisfies a regular, linear, homogeneous symmetric hyperbolic system   on  
 $(\widetilde{\mcM\enspace}\hspace{-0.5em} ,\widetilde g)$.
\end{enumerate}

These properties, together with stability of solutions of  symmetric hyperbolic systems, 
are the key ingredients in uniqueness and stability results of asymptotically
simple spacetimes $(\widetilde{\mcM}, \widetilde g)$ as defined in Def. 9.1. of \cite{F4}; the latter
definition includes the requirements that  $(\widetilde{\mcM}, \widetilde g)$
has a compact Cauchy hypersurface and every maximally extended null geodesic has a past endpoint on
$\scri^-$ and a future endpoint on $\scri^+$.  
We give here first a  uniqueness result for de Sitter
spacetime and then a sketchy version of the stability result 
(Thm 9.8 of \cite{F4}) which applies in particular to de Sitter.

\begin{theorem}
~
\begin{description}
\item[Uniqueness of de Sitter.]
Let smooth data for the MCFE be given on a $\scri^-$ which is topologically   $\mathbb{S}^3$
 and such that  $\widetilde d_{\alpha\beta\gamma}{}^{\delta}|_{\scri^-}$ 
vanishes identically. Then the evolving spacetime $(\widetilde{\mcM}, \widetilde g)$ 
is isometric to de Sitter.  
 \item[Stability of aymptotically simple solutions.]
Given an asymptotically simple spacetime  $(\widetilde{\mcM}, \widetilde g)$, 
then any data for the MCFE on $\scri^-$ which are close to the data for $(\widetilde{\mcM}, \widetilde g)$  
(in terms of suitable Sobolev norms) evolve to an asymptotically simple spacetime.  
\end{description}
\label{deSitterthm}
\end{theorem}

A motivation for the present work is to generalize these uniqueness and stability results to more general solutions 
of ($\ref{LambdaVac}$). 
Using again properties I.-III. above, it is straightforward  to generalize the above results on
asymptotically  simple solutions  to corresponding ``semiglobal'' results for any concrete family of solutions, where 
 ``semiglobal'' means the domain of dependence of $\scri^-$.
 On the other hand, and needless to say, 
any fully \emph{global} results for solutions which are not asymptotically simple but 
contain horizons and singularities involve ``cosmic censorship'' issues and 
will be very complicated. 
The main targets of the present work are Kottler (Schwarzschild-de Sitter) and Kerr-de Sitter
(KdS) spacetimes for which the topology of each connected component
of $\scri$ is  $\mathbb{R} \times \mathbb{S}^2$.
Our main achievement is a \emph{semiglobal} uniqueness 
result, namely Thm. \ref{first_main_thm2}, for a class of solutions which
includes Kerr- de Sitter. What makes our result highly non-trivial is
its particular formulation which we \emph{expect to be} useful for the fully 
global problem, for reasons given below.

In this uniqueness result, and from now onwards, we assume that $(\mcM, g)$ admits a non-trivial Killing vector field (KVF) $X$,
\begin{equation}
\label{Kill}
(\mcL_X g)_{\mu\nu} \,\equiv\, 2\nabla_{(\mu}X_{\nu)} \,=\, 0
\;.
\end{equation}
Since $X^{\mu}$ is a KVF, $F_{\mu\nu}:= \nabla_{\mu}X_{\nu}$ is a two-form:  $F_{(\mu\nu)}\,=\,0$.

The main purpose of this assumption is to achieve a simplification 
and to permit the use of a special technique.
However, as an aside we note that the existence of the isometry \emph{might}
 change the character of the stability problem substantially.
To see this on a heuristic basis, consider data for the MCFE on $\scri^-$  which are at the same time Killing initial data,  
and which are \emph{close to} Kerr-de Sitter in a suitable sense.
Now consider the time-evolution of such data, and assume that the spacetime
can be extended beyond its (``cosmological'') Cauchy horizon (as it is the
case for Schwarzschild-de Sitter and Kerr-de Sitter).
In this extension, the isometry should become timelike, and now another
conjecture, 
namely uniqueness of stationary black-holes, 
should lead to Kerr-de Sitter in the region between the event and the cosmological horizon. Extending backwards 
to the domain of dependence of $\scri^-$ suggests that the ``near Kerr-de Sitter'' data
will actually be Kerr-de Sitter in the above setting. Accordingly, the existence
of the isometry, together with reasonable global assumptions, can 
turn a  stability into a uniqueness problem. This ``effect'' is of course familiar
from uniqueness results for stationary, asymptotically flat solutions.
    
While obtaining global results sketched above is far beyond our present scope, 
it motivates our local analysis, in particular the use of the
so-called Mars-Simon tensor (MST)  \cite{mars,mars2,simon} 
in Thm  \ref{first_main_thm2}. This tensor is defined as follows.

\begin{equation}
\mathcal{S}_{\mu\nu\sigma\rho}\,:=\, \mathcal{C}_{\mu\nu\sigma\rho}  + Q\, \mathcal{U}_{\mu\nu\sigma\rho} 
\;,
\label{dfn_mars-simon}
\end{equation}
in terms of the quantities
%\jose{Here, and everywhere else, I have removed the boldface for the volume forms}
%
\begin{eqnarray}
 \mathcal{C}_{\mu\nu\sigma\rho} &:=& C_{\mu\nu\sigma\rho} +i C^{\star}_{\mu\nu\sigma\rho}
\;,
\label{Weyldual} \\
\mathcal{U}_{\mu\nu\sigma\rho}  &:=& -  \mathcal{F}_{\mu\nu}\mathcal{F}_{\sigma\rho} + \frac{1}{3}\mathcal{F}^2\mathcal{I}_{\mu\nu\sigma\rho} 
\;,
\\
\mathcal{I}_{\mu\nu\sigma\rho} &:=& \frac{1}{4} (g_{\mu\sigma}g_{\nu\rho} -g_{\mu\rho}g_{\nu\sigma} + i\volform_{\mu\nu\sigma\rho} )
\;,
\\
\mathcal{F}_{\mu\nu} 
&:=& F_{\mu\nu} +i F^{\star}_{\mu\nu}
\;,
\\
\mathcal{F}^2 &:=& \mathcal{F}_{\mu\nu} \mathcal{F}^{\mu\nu} \;. \label{Fsq}
\end{eqnarray}
%
%\marc{Some text added}
In these expressions $\volform_{\mu\nu\sigma\rho}$ is the volume form of
$g$, $\star$ the corresponding Hodge dual and $Q$ is a function. $\mathcal{F}_{\mu\nu}$
and ${\cal C}_{\alpha\beta\gamma\delta}$ are self-dual, i.e.\ they 
satisfy $\mathcal{F}^{\star}{}_{\mu\nu} = - i  \mathcal{F}_{\mu\nu}$ and 
 ${\cal C}_{~\alpha\beta\gamma\delta}^{\star} = -i  {\cal C}_{\alpha\beta\gamma\delta}$.
The symmetric double two-form $\mathcal{I}_{\mu\nu\sigma\rho}$ plays a natural role as a metric in the space
of self-dual two-forms, in the sense that ${\mathcal{I}}_{\mu\nu\sigma\rho} {\mathcal W}^{\sigma\rho} = {\mathcal W}_{\mu\nu}$ for any
self-dual two-form ${\mathcal W}_{\mu\nu}$.
% Note that as compared to e.g.\ \cite{IK}
% we have modified the definition of the tensor $\mathcal{U}_{\mu\nu\sigma\rho}$, in particular  in the sense that the function $Q$, which we shall specify later on, 
%is not part of its definition anymore.
In connection with this definition and its applications, there now arise naturally two 
\emph{a priori} independent problems:
\begin{enumerate}
\item Classify the solutions of (\ref{LambdaVac}) for which there \emph{exists} a $Q$ such
that the MST \eq{dfn_mars-simon} vanishes. 
\item
Prescribe the function $Q$ such that properties I. -III. above (or a subset thereof) hold for the MST.
  \end{enumerate} 

Problem 1 has been settled in \cite{mars,mars2} for case $\Lambda = 0$, while the extension to $\Lambda \neq 0$ was accomplished in
\cite{mars_senovilla}. The classes of solutions characterized in this way include Kerr and Kerr-de Sitter, respectively,
and these solutions can in fact be singled out by supplementing the condition $\mathcal{S}_{\mu\nu\sigma\rho} = 0$ 
with suitable  ``covariant'' conditions.   

As to problem 2 for $\Lambda = 0$, one sets  
\begin{equation}
\label{Q}
Q = 6 \sigma^{-1}
\end{equation}
in terms of the ``Ernst potential'' $\sigma$, defined up to an additive complex constant 
(called ``$\sigma$-constant'' henceforth) by
\begin{equation}
\partial_{\beta} \sigma = 2 X^{\alpha} {\cal F}_{\alpha\beta}
\;.
\end{equation}

%\tim{rewordings}
The corresponding MST then in fact satisfies a linear, homogeneous, symmetric hyperbolic system,
irrespective of how the $\sigma$-constant has been chosen
\cite{ik}.
In the asymptotically flat setting 
the  MST vanishes at infinity (which holds again for any  choice of the $\sigma$-constant
provided that the ADM mass is non-zero);  
it vanishes for all Kerr solutions in particular.
The $\sigma$-constant  is fixed uniquely in a 
natural way by requiring 
the Ernst potential to vanish at infinity. 
We remark that this symmetric hyperbolic system, or rather the wave equation which can be derived from it, has been used in uniqueness proofs for
stationary, asymptotically flat black holes \cite{aik1,ik}.

%Before we turn to Problem 2) for $\Lambda \neq 0$, 
%we recall the role of the MST 
%in a recent uniqueness theorem for stationary, asymptotically flat black
%holes
%As well-known, there is a subtlety in implementing Hawking's idea \cite{SH}  of extending, in a generic setting, the generating isometry of a 
%bifurcate Killing horizon off this horizon in order to obtain axial symmetry of the domain of outer communication. 
%In the last decade, substantial progress on this issue has been made in the case  $\Lambda =
%0$.  In particular \cite{aik1} has shown that stationary, asymptotically flat 
%black holes with bifurcate Killing horizons which are ``close to Kerr'' in a suitable sense are in fact Kerr. 
%Here the MST is the key tool in extending, in particular, the horizon isometry beyond a neighborhood of the horizon. 
%This is accomplished by considering the linear homogeneous wave equations wave equation which it satisfies on the domain of outer
%communication and applying Carleman estimates. This shows that the MST vanishes if it was assumed to be sufficiently small. 
%In the final step which leads to Kerr one can then use the classification result  \cite{mars2}.   

In analogy with \eq{dten} we now define

\begin{equation}
\label{T}
\widetilde {\cal T}_{\alpha\beta\gamma}{}^{\delta} := \Theta^{-1}{\cal S}_{\alpha\beta\gamma}{}^{\delta}.
\end{equation}

For $\Lambda > 0$,  key properties of these tensors can be summarized as follows
(I.-III. is shown in the present work while IV. is a reformulation of a 
result of \cite{mars_senovilla}; II.\ and IV.\ in fact hold for any sign of the cosmological constant)
\begin{enumerate}
\item[I.] There exists a function $Q_0$ such that the corresponding MST
${\cal S}_{\alpha\beta\gamma\delta}^{(0)}$ vanishes on $\scri$, whence
$\widetilde {\cal T}_{\alpha\beta\gamma\delta}^{(0)}$ extends regularly to $\scri$. 
\item[II.] There exists a function $Q^{(ev)}$ such that the
corresponding MST ${\cal S}_{\alpha\beta\gamma\delta}^{(ev)}$ 
satisfies a linear, homogeneous symmetric hyperbolic system on $(\mcM, g)$.
\item[III.] ${\cal T}_{\alpha\beta\gamma\delta}^{(ev)}$ 
satisfies a linear, homogeneous  symmetric hyperbolic system
on $(\widetilde{\mcM\enspace}\hspace{-0.5em} ,\widetilde g)$ which is of ``Fuchsian type'' at $\scri$.
\item[IV.] When ${\cal S}_{\alpha\beta\gamma\delta}$ is required to vanish 
identically for some $Q$, then $Q = Q_0 = Q^{(ev)}$. 
\end{enumerate}

%Concerning property (ii), the Weyl tensor $C_{\alpha\beta\gamma\delta}$ satisfies a symmetric hyperbolic system of evolution equations, 
%which immediately gives rise to a symmetric hyperbolic system of equations
%for its rescaled counterpart $\widetilde d_{\alpha\beta\gamma\delta}$ in the conformally rescaled spacetime.
%For the rescaled MST $\widetilde {\cal T}_{\alpha\beta\gamma\delta}$ a corresponding system of equations would be crucial to make sure that properties provided
%on $\scri^-$ propagate into the spacetime.

Conditions I. - III. stated above for the MST should be compared with the 
corresponding conditions stated earlier for the Weyl tensor.
Unfortunately, or maybe for a deeper reason, there appears to be
%\marc{slight change (don't think we can guarantee there exists no choice)}
 no universal definition of $Q$ anymore which satisfies I. - III. simultaneously. 

We proceed with explaining these findings in some detail, and with describing
their arrangement in the following sections.  
The function  $Q_0$  is introduced in \eq{definition_Q0}, and property I. is shown in Proposition \ref{prop_reg_S}.
%%\jose{I don't understand this sentence, until the end of paragraph}
%\marc{I agree, I find it confusing \\ -- \\
%tim: rewordings}
%In connecton with our main achievement, namely to characterize spacetimes with vanishing MST,  
%we remark that  $Q=Q_0$ \emph{necessarily} needs to hold in such a spacetime.
Next, Theorem \ref{thm_nec_cond} gives necessary and sufficient conditions
on the data in order for $\widetilde {\cal T}_{\alpha\beta\gamma\delta}^{(0)}$ 
to vanish on $\scri$. These conditions agree with conditions (i) and (ii) in
Theorem \ref{first_main_thm2} quoted below.

On the other hand, in \eq{ev_dfn_Q}-\eq{defJ} we define a \emph{class of functions}  
$Q^{(ev)}$ for which we show in Section~\ref{sect_deriv_ev_MST}
%-\ref{subsec:FOSH}
  that the corresponding MST ${\cal S}_{\alpha\beta\gamma\delta}^{(ev)}$ 
satisfies a linear, homogeneous symmetric hyperbolic system, which
 gives property II (and from which  one readily derives a system of wave equations).
For the  rescaled tensor $\widetilde {\cal  T}_{\alpha\beta\gamma\delta}^{(ev)}$ we then obtain equations
of the same form on $\widetilde{\mcM\enspace}$ (cf.\ Lemmas~ \ref{wave_eqn_MST}
and \ref{lemma_evolution}).  
The appropriate definition of $Q^{\mathrm{(ev)}}$  involves a ``$\sigma$-(integration)-constant''
(called ``$a$'' in \eqref{sigma_i_prelim}), 
and in analogy with the case $\Lambda = 0$ mentioned before there is again a natural way (namely
\eqref{afix}) 
of fixing the constant from the asymptotic conditions. However, in contrast to the case $\Lambda = 0$,
the resulting ${\cal S}_{\alpha\beta\gamma\delta}^{(ev)}$  does not vanish automatically on $\scri^-$,
whence ${\cal T}_{\alpha\beta\gamma\delta}^{(ev)}$  is not necessarily regular there. 
In Definition ~\ref{KdS_like} we call solutions for which ${\cal T}_{\alpha\beta\gamma\delta}^{(ev)}$  (with the optimal
$\sigma$-constant) can be regularized on $\scri^-$ (and agrees with 
$ \widetilde{\mathcal{T}}^{(0)}_{\mu\nu\sigma}{}^{\rho}$ on $\scri^-$) 
``asymptotically Kerr-de Sitter like''. 
 This class can be characterized in terms of the data as follows 
(this is a shortened version of Thm. ~\ref{prop_Qs}):      

\begin{theorem}
\label{short}
%\tim{rewordings and $|Y|^2>0$-condition removed}
Consider a $\Lambda>0$-vacuum spacetime which admits a smooth $\scri^-$ and a KVF $X$.
Denote by $Y$ the CKVF induced, in the conformally rescaled spacetime, by $X$ on $\scri^-$.

The condition
\begin{equation*}
\widetilde{\mathcal{T}}^{(\mathrm{ev})}_{\mu\nu\sigma}{}^{\rho}|_{\scri^-} \,=\, \widetilde{\mathcal{T}}^{(0)}_{\mu\nu\sigma}{}^{\rho}|_{\scri^-}
\end{equation*}
holds if and only if
$Y^j$ is a common eigenvector of $\widehat C_{ij}$ and $D_{ij}$,
%%
%\begin{equation*}
%\widehat \volform_{ij }{}^{k}  \widehat C_{kl} Y^{j}Y^l=0 ~~\mbox{and}~~  \widehat \volform_{ij }{}^{k}  D_{kl} Y^{j}Y^l=0
%\end{equation*}
%%
where $\widehat C_{ij}$ is the Cotton-York tensor (\ref{cotton-york}) and $D_{ij} =
d_{titj}|_{\scri^-}$.
\end{theorem}

This now suggests considering a Cauchy problem for the MCFE on $\scri^-$, 
starting from asymptotically Kerr-de Sitter-like data.
However, in contrast to the evolution equation for the rescaled Weyl tensor
$\widetilde d_{\alpha\beta\gamma\delta}$,  the coefficients in the evolution equation for 
$\widetilde{\cal T}_{\alpha\beta\gamma\delta}^{(ev)}$ are not  
regular at $\scri$ and even not necessarily so off some neighborhood of $\scri$.  (Non-regularities 
may occur already in the evolution equations for  ${\cal S}_{\alpha\beta\gamma\delta}^{(ev)}$). 
Regarding $\scri$, we are now dealing with a linear homogeneous \emph{Fuchsian} symmetric hyperbolic system
(Lemma  \ref{lemma_evolution}). Adapting  results available in the
literature, we prove in Lemma \ref{UniquenessLemma} a local uniqueness theorem for regular solutions of a class of Fuchsian
systems which includes the present one. 
We then apply this result  to ``trivial'' data satisfying
 $\widetilde{\cal T}_{\alpha\beta\gamma\delta}^{(ev)}|_{\scri^-} = 0$, 
which we call ``Kerr-de Sitter-like''.  
Our preliminary uniqueness result, Lemma \ref{UniquenessT},
 now yields local-in-time uniqueness of these solutions,
and implies that  ${\cal S}_{\alpha\beta\gamma\delta}^{(\mathrm{ev})}$ 
vanishes  near $\scri^-$. However, this conclusion does not immediately
extend to the whole domain of dependence of $\scri^-$ since the evolution equation is
manifestly regular only in some neighborhood of (and excluding) $\scri^-$.  
Nevertheless, the required result does follow from  the classification results  of 
\cite{mars_senovilla}, so ${\cal S}_{\alpha\beta\gamma\delta}^{(\mathrm{ev})}
\equiv 0$ indeed holds on the domain of dependence of $\scri^-$.

Altogether this yields the following classification result for
Kerr-de Sitter like spacetimes in terms of data on $\scri^-$, which may be
considered as counterpart of the first part of \ref{deSitterthm}  above:

\begin{theorem}
\label{first_main_thm2}
Let $(\Sigma,h)$ be a Riemannian 3-manifold which admits a CKVF~ $Y$ with $|Y|^2>0$,
 complemented by
a TT tensor $D_{ij}$ to asymptotic Cauchy data.
Then there exists a maximal globally hyperbolic
%\tim{is that clear, because the Cauchy problem is solved in the unphysical spacetime and in fact $\scri$ is not a Cauchy surface in the physical spacetime so sth like a 
%maximal Cauchy development makes no sense to me, cf. an mnote below}
%\walter{right! - any idea to fix this ?}
 $\Lambda>0$-vacuum spacetime $(\mcM,g)$
 which admits a KVF
$X^i$ with  $X^i|_{\scri^-} = Y^i$ and such that the associated MST vanishes, and 
$\Sigma$ represents past null infinity $\scri^-$ with $g_{ij}|_{\scri^-}=h_{ij}$ and $d_{titj}|_{\scri^-}= D_{ij}$  if and only if
\begin{enumerate}
\item[(i)] $ \widehat C_{ij} = \CconstL|Y|^{-5}(Y_iY_j -\frac{1}{3}|Y|^2 h_{ij})$ for some constant $\Cconst$, and
\item[(ii)]    $D_{ij} =\Dconst  |Y|^{-5} (Y_iY_j -\frac{1}{3}|Y|^2 h_{ij})$ for some constant $\Dconst$.
\end{enumerate}
\end{theorem}

Spacetimes with vanishing MST have very different properties
depending on the values taken by the free constants in the
family. In particular, the maximal domain
of dependence of $\scri$ may or may not be extendible across a Killing horizon,
and these different behaviours occur even within the class
of spacetimes with vanishing $\Cconst$. In this latter case, $\scri$ is locally conformally
flat and the data consist simply
in a choice of a conformal Killing vector $Y$ in (a domain of) $\mathbb{S}^3$
and a choice of a constant $\Dconst$. An interesting (and probably difficult)
question is whether it is possible to identify directly at $\scri$ the
behaviour of its domain of dependence in the large. In particular, it
would be interesting to see if the properties of $Y$ at its zeroes can
be related to the existence of a Killing horizon across which the
domain of dependence of $\scri$ can be extended.

We remark that we do not obtain a counterpart to the stability result (part 2 of
Thm. \ref{deSitterthm}, since (to our knowledge) there is no general result guaranteeing existence and stability of solutions to Fuchsian systems.
Recall also the remark after  \eqref{Kill} in connection with the significance of the stability problem
in the presence of isometries.

%This result also guarantees that in the setting where the MST vanishes for some
%choice of $Q$, one  in fact has $Q=Q_0=Q_{\mathrm{ev}}$. 
%While $Q_0$ is convenient to derive necessary conditions on $\scri^-$ to obtain a $\Lambda>0$-vacuum spacetime with vanishing MST, 
%$Q_{\mathrm{ev}}$ is essential to derive sufficient ones. 

%Theorem~\ref{prop_Qs} characterizes Cauchy data at null infinity for which both $Q$'s have the same leading order term at $\scri^-$. 

%One subtlety concerns the regularity of the evolution equation for ${\cal T}_{\alpha\beta\gamma\delta}^{(\mathrm{ev})}$.
%First of all, 

%Altogether this establishes our main result, Theorem~\ref{first_main_thm}, which characterizes $\Lambda>0$-vacuum spacetimes with smooth $\scri^-$ and vanishing MST, 
%so-called ``Kerr-de Sitter like''-spacetimes, which are formally defined in
%\tim{do we want to use this terminology?... other wise remove brackets in the title. Jose: It's fine with me}
% in terms of asymptotic Cauchy data at $\scri^-$.

%The fact that  $Q_{\mathrm{ev}}$ and $Q_0$ generally do not coincide means that
%the resulting MST ${\cal S}_{\alpha\beta\gamma\delta}^{(\mathrm{ev})}$ 
%does not {\it necessarily} vanish at $\scri$. Nevertheless, the requirement that it does vanish 
%leaves us with a large class of so-called ``asymptotically Kerr-de Sitter-like''
%\marc{removed ``past''}
%spacetimes  defined formally in Definition~\ref{asympt_KdS_like}  and characterised by their data on $\scri^-$ in Theorem~\ref{prop_Qs}.

In the final Section~\ref{sec_CKVFs} we analyze the relations between the vanishing of the rescaled MST
(\ref{T}) (or the corresponding condition on the data) and the existence of
other conformal Killing vector fields on $\scri$, and we discuss the 
extension of the latter to Killing vector fields on ${\mcM}$. This result,
given in Proposition \ref{prop_2CKVF},  will be relevant
%\tim{actually this is not so important...}\jose{changed, ok?}
 for the classification of spacetimes with vanishing ${\cal S}_{\alpha\beta\gamma\delta}^{(\mathrm{ev})}$   and
conformally flat $\scri$ presented in the
subsequent paper mentioned above already \cite{mpss}.

\section{The Mars-Simon tensor (MST) at null infinity}
\label{section_mars_simon}

\subsection{The conformally rescaled spacetime}
%\jose{Title changed}
\label{Mars-Simon_conf}

In this section we collect key equations which  are gauge-independent, 
and which hold irrespectively of the sign (or vanishing) of $\Lambda$.

In the asymptotic setting described in the introduction 
the pair $(\widetilde g, \Theta)$ satisfies the \emph{metric conformal field equations (MCFE)}
on $\widetilde{\mcM\enspace}$ \cite{F3} (we use
tildes for all geometric objects associated to $\widetilde g$),
\begin{eqnarray}
 && \widetilde \nabla_{\rho} \widetilde d_{\mu\nu\sigma}{}^{\rho} =0\;,
 \label{conf1}
\\
 && \widetilde\nabla_{\mu} \widetilde L_{\nu\sigma} - \widetilde\nabla_{\nu}\widetilde L_{\mu\sigma} =   \widetilde\nabla_{\rho}\Theta \, \widetilde d_{\nu\mu\sigma}{}^{\rho}\;,
 \label{conf2}
\\
 && \widetilde\nabla_{\mu}\widetilde\nabla_{\nu}\Theta = -\Theta \widetilde L_{\mu\nu} +\widetilde s \widetilde g_{\mu\nu}\;,
 \label{conf3}
\\
 && \widetilde\nabla_{\mu}\widetilde  s = -\widetilde L_{\mu\nu}\widetilde \nabla^{\nu}\Theta\;,
 \label{conf4}
\\
 && 2\Theta \widetilde s - \widetilde \nabla_{\mu}\Theta\widetilde \nabla^{\mu}\Theta = \Lambda /3
 \label{conf5}
\;,
\\
 && \widetilde R_{\mu\nu\sigma}{}^{\kappa}[\widetilde g]
  = \Theta  \widetilde d_{\mu\nu\sigma}{}^{\kappa} + 2\big(\widetilde g_{\sigma[\mu}\widetilde  L_{\nu]}{}^{\kappa}  - \delta_{[\mu}{}^{\kappa}\widetilde L_{\nu]\sigma} \big)
 \label{conf6}
\;.
\end{eqnarray}
%
%\marc{brackets dropped from $\widetilde R_{\mu\nu\sigma}{}^{\kappa}$ \\ -- \\
%tim: They are useful to emphasize that the Riemann tensor is regarded as a function of $\widetilde g$ in contrast to e.g.\ the Schouten
%tensor which, however, does not really play a role in this paper...}
%\marc{Bracket back and a phrase added}
where the Riemann tensor
$\widetilde R_{\mu\nu\sigma}{}^{\kappa}[\widetilde g]$ is to be regarded
as a differential operator on $\widetilde g$,
%\marc{check whether the global numerical factor agrees with the conventions used}
while 
$\widetilde L_{\mu\nu}:= \frac{1}{2}\widetilde R_{\mu\nu} - \frac{1}{12}\widetilde R \widetilde g_{\mu\nu}$, $\widetilde d_{\mu\nu\sigma}{}^{\rho} := \Theta^{-1}  
\widetilde C_{\mu\nu\sigma}{}^{\rho}$ are, respectively, the Schouten and rescaled Weyl tensor
of $\widetilde g$, and 
\begin{equation}
\widetilde s \,:=\,\frac{1}{4}\Box_{\widetilde g} \Theta + \frac{1}{24}  \widetilde R\Theta
\;.
\end{equation}
%
%which are treated as independent of $\tilde g$ and $\Theta$.

Let us now express the MST in terms of unphysical fields on $(\widetilde{\mcM\enspace}\hspace{-0.5em} ,\widetilde g)$.
We first of all note that the push-forward $\widetilde X^{\mu}$ of the KVF $X^{\mu}$, which we identify with $X^{\mu}$,  satisfies the \emph{unphysical Killing equations}
\cite{ttpKIDs}
\begin{equation}
\widetilde F_{(\mu\nu)}\,=\, 0 \quad \text{and} \quad \widetilde F\,=\, 4\widetilde X^{\mu}\widetilde\nabla_{\mu}\log\Theta
\label{unphys_Killing}
\;,
\end{equation}
where
\begin{equation}
 \widetilde F_{\mu\nu} \,:=\, (\widetilde\nabla_{\mu}\widetilde X_{\nu})_{\mathrm{tf}}\;, \quad  \widetilde F:=\widetilde\nabla_{\mu}\widetilde X^{\mu}
\;.
\end{equation}
and the symbol $(.)_{\mathrm{tf}}$ denotes the trace-free part of the corresponding $(0,2)$-tensor. $\widetilde F_{\mu\nu}$ is hence a two-form
and we can define $\widetilde{\mathcal{C}}_{\mu\nu\sigma\rho}$, $\widetilde{\mathcal{U}}_{\mu\nu\sigma\rho}$,
$\widetilde{\mathcal{I}}_{\mu\nu\sigma\rho}$, $\widetilde{\mathcal{F}}_{\mu\nu}$ and
$\widetilde{ \mathcal{F}}^2$ using definitions
 analogous to  (\ref{Weyldual})-(\ref{Fsq}) with all geometric objects referred to $(\widetilde{\mcM\enspace}\hspace{-0.5em} ,\widetilde g)$.
The following relations are found via a simple computation.
%\marc{Explicit definition of $\widetilde{\mathcal U}$ dropped, two $\partial$ changed to $\widetilde\nabla$}
%
\begin{eqnarray}
\mathcal{C}_{\mu\nu\sigma}{}^{\rho} &=&\widetilde{ \mathcal{C}}_{\mu\nu\sigma}{}^{\rho}\;,
\label{formula_C}
\\
\mathcal{I}_{\mu\nu\sigma}{}^{\rho} &=& \Theta^{-2}\widetilde{\mathcal{I}}_{\mu\nu\sigma}{}^{\rho}
\;,
\\
 F_{\mu\nu} &=&  \nabla_{\mu}(\Theta^{-2}\widetilde  X_{\nu})  \nonumber \\
& = & \Theta^{-2}(\widetilde F_{\mu\nu} + \frac{1}{4}\widetilde g_{\mu\nu}\widetilde F)
+  \Theta^{-3}(2\widetilde  X_{[\mu}\widetilde \nabla_{\nu]}\Theta  -  \widetilde g_{\mu\nu}\widetilde  X^{\sigma} \widetilde\nabla_{\sigma}\Theta ) 
\nonumber
%\\
%&=&  \Theta^{-2}(\breve {\widetilde F}_{\mu\nu}
%+ 2\widetilde  X_{[\mu}\widetilde\nabla_{\nu]}\log\Theta   ) 
\\
&=&  \Theta^{-2}(\widetilde F_{\mu\nu}
+\Theta^{-1}\widetilde H_{\mu\nu}  ) 
\;,
\\
%\widetilde H_{\mu\nu} &:=&   2\widetilde  X_{[\mu}\widetilde\nabla_{\nu]}\Theta\;,
%\\
\mathcal{F}_{\mu\nu} &=&  \Theta^{-2}(\widetilde{ \mathcal{F}}_{\mu\nu}
+\Theta^{-1}\widetilde{\mathcal{H}}_{\mu\nu} )
\;,
%\\
% \widetilde {\mathcal{H}}_{\mu\nu} &:=&  \widetilde H_{\mu\nu}  +i \widetilde H^{\star}_{\mu\nu} 
%\;,
\\
\mathcal{F}^2 &=& 
%\mathcal{F}_{\alpha\beta}\mathcal{F}^{\alpha\beta}   \,=\, 
\widetilde{ \mathcal{F}}^2
+ 2\Theta^{-1} {\widetilde{ \mathcal{F}}}_{\alpha\beta}\widetilde{\mathcal{H}}^{\alpha\beta} +\Theta^{-2} \widetilde{\mathcal{H}}^2\;,
\\
\mathcal{U}_{\mu\nu\sigma}{}^{\rho}  &=& \Theta^{-2}\widetilde{\mathcal{U}}_{\mu\nu\sigma}{}^{\rho} -\Theta^{-3}\big( 
\widetilde{ \mathcal{F}}_{\mu\nu}\widetilde{\mathcal{H}}_{\sigma}{}^{\rho} 
+ 
\widetilde{\mathcal{H}}_{\mu\nu} \widetilde{ \mathcal{F}}_{\sigma}{}^{\rho}
-\frac{2}{3} {\widetilde{ \mathcal{F}}}_{\alpha\beta}\widetilde{\mathcal{H}}^{\alpha\beta}\widetilde{\mathcal{I}}_{\mu\nu\sigma}{}^{\rho}  \big)
\nonumber
\\
&&
  -\Theta^{-4}\big( \widetilde{\mathcal{H}}_{\mu\nu}\widetilde{\mathcal{H}}_{\sigma}{}^{\rho} 
-\frac{1}{3}\widetilde{\mathcal{H}}^2 \widetilde{\mathcal{I}}_{\mu\nu\sigma}{}^{\rho} \big)
\;,
\label{formula_Q}
%\\
%\widetilde{\mathcal{U}}_{\mu\nu\sigma}{}^{\rho}  &:=& - \widetilde {\mathcal{F}}_{\mu\nu}\widetilde{\mathcal{F}}_{\sigma}{}^{\rho} + \frac{1}{3}\widetilde{\mat%hcal{F}}^2\widetilde{\mathcal{I}}_{\mu\nu\sigma}{}^{\rho}
%\;.
\end{eqnarray}
where we have set
\begin{eqnarray}
\widetilde H_{\mu\nu} &:=&   2\widetilde  X_{[\mu}\widetilde\nabla_{\nu]}\Theta\;,
\\
 \widetilde {\mathcal{H}}_{\mu\nu} &:=&  \widetilde H_{\mu\nu}  +i \widetilde H^{\star}_{\mu\nu} 
\;.
\end{eqnarray}
%
%Note that $\widetilde {\mathcal{F}}_{\mu\nu}$ and $\widetilde F_{\mu\nu}$  are anti-symmetric, i.e.\ 2-forms, due to the fact that  $\widetilde X^{\mu}$ needs to be a conformal Killing vector field (CKVF).
We want to investigate how the MST behaves when approaching the conformal boundary $\scri$.
Note that the conformal Killing equation implies that $\widetilde X^{\mu}$ admits a smooth extension across $\scri$
\cite{RG}, in particular the tensor 
$\widetilde{\mathcal{U}}_{\mu\nu\sigma}{}^{\rho}$ is a regular object there.

We observe that the following relations are fulfilled.
 The first two are general identities for
self-dual two-forms and the third
one is a consequence of $\widetilde H$ being a simple two-form,
\begin{eqnarray}
 {\widetilde{ \mathcal{F}}}_{\alpha\beta}\widetilde{\mathcal{H}}^{\alpha\beta}
% &=& 
%\breve{\widetilde F}_{\alpha\beta} \widetilde H^{\alpha\beta}  +i\widetilde \volform^{\alpha\beta\sigma\rho} \breve{\widetilde F}_{\alpha\beta} \widetilde H_{\sigma\rho}   - \frac{1}{4}\widetilde \volform_{\alpha\beta}{}^{\mu\nu} \widetilde \volform^{\alpha\beta\sigma\rho}\breve{\widetilde F}_{\mu\nu}\widetilde H_{\sigma\rho}
%\\
&=& 2{\widetilde F}_{\alpha\beta} \widetilde H^{\alpha\beta}  + 2 i {\widetilde F}^{\alpha\beta} \widetilde H^{\star}_{\alpha\beta}
\,=\, 2 {\widetilde{ \mathcal{F}}}_{\alpha\beta}\widetilde{H}^{\alpha\beta}\;,
\label{relation1}
\\
\widetilde{\mathcal{F}}^2
&=&
2\widetilde F^2  + 2i {\widetilde F}^{\star}_{\alpha\beta}{\widetilde F}^{\alpha\beta}
%\\
%\nabla_{\mu}\sigma &=& 2 \Theta^{-2}\widetilde X^{\alpha} (\breve {\widetilde{ \mathcal{F}}}_{\mu\alpha} +\Theta^{-1}\widetilde{\mathcal{H}}_{\mu\alpha} )
%\;, 
%\quad |\breve{\widetilde F}|^2\,=\, \breve{\widetilde F}_{\mu\nu}\breve{\widetilde F}^{\mu\nu}
\;,
\\
 \widetilde{\mathcal{H}}^2 &=&
2\widetilde H^2 
\;,
\\
  \widetilde X^{\mu} \widetilde{\mathcal{H}}_{\mu\nu} &=&
\widetilde X^2\widetilde \nabla_{\nu}\Theta-   \frac{1}{4} \Theta \widetilde  F \widetilde X_{\nu}
\label{formula_XH}
\;.
\end{eqnarray}
Here and henceforth we write $\widetilde T^2:= \widetilde T_{\alpha_1 \cdots \alpha_p} \widetilde
T^{\alpha_1 \cdots \alpha_p}$ for any  $(0,p)$-\emph{spacetime}-tensor, while we write
$| T|^2:=  T_{i_1 \cdots i_p} 
T^{i_1 \cdots i_p}$ for any  $(0,p)$-\emph{space}-tensor on $\scri$.
Since we  never write down explicitly the second component of a vector, the reader will not get confused by this notation.
%\jose{I don't like this notation, $|\,\,|^{2}$ not being positive definite. Couldn't we just use $\tilde T^2$. Cf. later
%\\ -- \\
%tim: I understand what you mean... for e.g. $X^2$ it might be confusing because it could be a component as well?} 

Moreover, the MCFE  and the unphysical Killing equations  imply that
\begin{eqnarray}
 \widetilde{\mathcal{H}}^2 &=&
 -\frac{4}{3}\Lambda \widetilde X^2  +  8\Theta \widetilde  s \widetilde X^2  - \frac{1}{4}\Theta^2 \widetilde F^2
\;,
\label{relation_H2}
\\
 {\widetilde{ \mathcal{F}}}_{\alpha\beta}\widetilde{\mathcal{H}}^{\alpha\beta}
%&=& 2\breve{\widetilde F}_{\alpha\beta} \widetilde H^{\alpha\beta}  +i\widetilde \volform^{\alpha\beta\mu\nu} \breve{\widetilde F}_{\alpha\beta} \widetilde H_{\mu\nu}
%\\
&=& \Theta  \widetilde X^{\alpha}\widetilde \nabla_{\alpha}\widetilde F 
+ 4\Theta \widetilde X^{\alpha} \widetilde X^{\beta} \widetilde L_{\alpha\beta}
- 4s\widetilde X^2
  +2i  \widetilde F^{\mu\nu} \widetilde H^{\star}_{\mu\nu}
\;.
\label{relation_FH}
\end{eqnarray}

\subsection{Cauchy data at $\scri^-$}
%\jose{Title changed}
%\subsection{Some useful relations on a spacelike $\scri$}
\label{constraints}

%\tim{add sth about the asympt Cauchy problem and the KID equation
%, in particular add a theorem}
Let us henceforth assume a positive cosmological constant 
\begin{equation}
\Lambda \,>\, 0\;.
\end{equation}
We consider a connected component $\scri^{-}$ of past null infinity.
As in \cite{ttp2}, to which we refer the reader for further details, we use adapted coordinates $(x^0=t, x^i)$ with $\scri^-=\{t=0\}$
%\marc{$\scri^{-}$ reserved for a connected component. Is this all right?  \\ -- \\ tim: $\scri$ / $\scri^-$... should be checked for consistency\\ --\\
%Jose: I have added a sentence in the Intro which makes things OK, I believe}
 and impose a \emph{wave map gauge condition}  
with
\begin{equation}
 \widetilde R=0\;, \quad \widetilde s|_{\scri^-}=0\;, \quad \widetilde g_{tt}|_{\scri^-}=-1\;, \quad \widetilde  g_{ti}|_{\scri^-}=0\;, 
\quad  \widetilde W^{\sigma}=0\;, \quad \check g_{\mu\nu} = \widetilde g_{\mu\nu}|_{\scri^-}
\;.
\label{gauge_conditions_compact}
\end{equation}
The gauge freedom to prescribe $\widetilde R$ and $\widetilde s|_{\scri^-}$ reflects the freedom to choose the conformal factor $\Theta$,
which is treated as an unknown in the MCFE.
It is well-known that the freedom to choose coordinates near a spacelike hypersurface with induced metric $h_{ij}$ can be employed to prescribe $\widetilde g_{tt}|_{\scri^-}$ 
and  $\widetilde  g_{ti}|_{\scri^-}$, as long as $(\widetilde g_{tt}-h^{ij}\widetilde g_{ti}
\widetilde g_{tj})|_{\scri^-}<0$ is satisfied.

The remaining freedom to choose coordinates is captured by the wave map gauge condition, a generalization of the classical harmonic gauge condition,
and requires the vanishing of the so-called wave gauge vector
\begin{equation}
H^{\sigma}\,:=\, g^{\alpha\beta}(\Gamma^{\sigma}_{\alpha\beta}-\check\Gamma^{\sigma}_{\alpha\beta}) - W^{\sigma } =0
\;,
\end{equation}
where $\check g_{\mu\nu}$ denotes some target metric, the $\check \Gamma^{\sigma}_{\alpha\beta}$'s are the associated connection coefficients, 
and the $W^{\sigma}$'s are the gauge source functions, which can be arbitrarily prescribed \cite{F}.
The target metric is introduced for the wave gauge vector to become a tensor. 
Here, as in \cite{ttp2}, we have chosen $\check g_{\mu\nu}$ to be independent of $t$ and
to agree with $\widetilde g_{\mu\nu}$ on $\scri^-$. 
The gauge has been chosen in such a way that $\partial_t \widetilde g_{\mu\nu}$ vanishes on $\scri^-$, in order to make
the computations as simple as possible.
Given arbitrary coordinates the wave map gauge can be realized by solving wave equations.

Viewing the MCFE as an evolution problem with initial data on $\scri^{-}$, 
the free data are a (connected) Riemannian 3-manifold 
%\marc{$\mcN$ changed to $\Sigma$}
$(\Sigma, h_{ij})$, which represents $\scri^-$ in the emerging spacetime,%
\footnote{It is actually merely the conformal class of the Riemannian 3-manifold which matters geometrically.
This will be relevant in paper II \cite{mpss}.
}
and a TT tensor $D_{ij}$ (i.e.: trace-free and divergence-free) which satisfies the relation 
\begin{equation}
D_{ij}=\widetilde d_{titj}|_{\scri^-}
\end{equation}
 once the asymptotic Cauchy problem has been solved
% \marc{phrase added to identify $\iota$ \\ --\\  tim: not explicitly mentioned later on...add some comment? \\ -- \\ marc: added}

\begin{theorem}[\cite{F_lambda}]
Let $(\Sigma, h_{ij})$ be a Riemanian 3-manifold, $D_{ij}$ a symmetric $(0,2)$-tensor and $\Lambda>0$.
Then, if and only if $D_{ij}$ is a TT tensor,
the tuple $(\Sigma, h_{ij}, D_{ij})$ defines an (up to isometries)
 unique maximal globally hyperbolic   development (in the unphysical spacetime)
of the $\Lambda$-vacuum field equations where $\Sigma$ can be embedded,
with embedding $\iota$, such that
 $\iota(\Sigma)$
represents $\scri^-$ 
%(i.e.\ $\ol \Theta=0$ and $\ol{\mathrm{d}\Theta}\ne 0 $) 
with $\iota^* \widetilde g_{ij}|_{\Sigma}=h_{ij}$ and $\iota^* \widetilde d_{titj}|_{\Sigma}=D_{ij}$.
 %The solution depends continuously on the initial data.
\end{theorem}

For simplicity, we will often identify $\Sigma$ with its image under
$\iota$ and drop all reference to the embedding.

It is a property of the spacelike Cauchy problem that all transverse derivatives can be computed algebraically from the initial data
(here $h_{ij}$ and $D_{ij}$).
In the gauge \eq{gauge_conditions_compact} the MCFE  \eq{conf1}-\eq{conf6}  enforce  the following relations on $\scri^-$, cf.\ \cite{F_lambda, ttp2},
\begin{eqnarray}
 &\ol {\widetilde g}_{tt} =-1\;, \quad \ol {\widetilde g}_{ti}=0\;, \quad \ol {\widetilde g}_{ij} = h_{ij}\;, \quad \ol{\partial_t {\widetilde g}_{\mu\nu}}=0\;,
\label{constr2}
&
\\
& \ol \Theta =0\;, \quad \ol{\partial_t \Theta} = \sqrt{\frac{\Lambda}{3}}\;, \quad \ol{\partial_t\partial_t\Theta}=0\;,
\label{constr3}
&
\\
&
\ol{\partial_t\partial_t\partial_t\Theta} = - \frac{1}{2}\sqrt{\frac{\Lambda}{3}} \widehat R
\;,\quad 
\ol{ \partial_t\partial_t\partial_t\partial_t\Theta} %=  -\ol{\widetilde\nabla_t\widetilde\nabla_t s}
=  0
\;,
&
\\
&\ol {\widetilde s}=0\;, \quad \ol{\partial_t \widetilde s} =  \frac{1}{4}\sqrt{\frac{\Lambda}{3}} \,\widehat   R\;,
\label{constr4}
&
\\
&\ol {\widetilde L}_{ij} = \widehat  L_{ij} \;, \quad \ol {\widetilde L}_{ti}=0\;, \quad \ol {\widetilde L}_{tt} = \frac{1}{4}\widehat  R \;,
\label{constr5}
&
\\
& \ol{\partial_t \widetilde L_{ij}} =  -\sqrt{\frac{\Lambda}{3}} \, D_{ij}\;, \quad
  \ol{\partial_t \widetilde L_{ti}} =\frac{1}{4}\partial_i\widehat  R  \;, \quad  \ol{\partial_t \widetilde L_{tt}} =0\;,
\label{constr6}
&
\\
&  \ol {\widetilde d}_{titj} = D_{ij}\;,\quad \ol {\widetilde d}_{tijk}  =  \sqrt{\frac{3}{\Lambda}} \widehat C_{ijk}\;,
\label{constr7}
&
\\
& \ol{\partial_{t} {\widetilde d}_{titj}}= \sqrt{\frac{3}{\Lambda}}\widehat  B_{ij}\;, \quad
\ol{\partial_{t} {\widetilde d}_{tijk}} = 2\widehat  \nabla_{[j}D_{k]i}\;.
\label{constr8}
 &
\\
& \ol {\widetilde \Gamma}^k_{ij} = \widehat \Gamma^k_{ij}\;, \quad  \ol {\widetilde \Gamma}^t_{ij} =    \ol {\widetilde \Gamma}^t_{ti} =\ol {\widetilde \Gamma}^t_{tt} =  \ol {\widetilde \Gamma}^k_{tt} =  \ol{\widetilde \Gamma}^k_{ti} = 0\;,
\label{christoffel}
&
\\
&\ol {\widetilde R}_{tijk} =0\;, \quad  \ol {\widetilde R}_{titj} =-\widehat L_{ij} + \frac{1}{4} h_{ij}\widehat R
\;,
&
\\
&
\ol{\partial_t \widetilde R_{tijk}} = \widehat C_{ijk}-\frac{1}{2} h_{i[j}\widehat \nabla_{k]}\widehat R
\;, \quad 
\ol{\partial_t\widetilde R_{titj}} = 2\sqrt{\frac{\Lambda}{3}} D_{ij}
\;.
\label{constr_last}
&
\end{eqnarray}
An overbar will be used to denote the restriction of spacetime objects to $\scri^-$,
if not explicitly stated otherwise  (in the latter   cases it will denote ``complex conjugation'').
We use the symbol $\enspace\widehat{}\enspace$ to denote objects associated to the induced Riemannian metric $h_{ij}$, in particular $\widehat C_{ijk}$, $\widehat L_{ij}$ and $\widehat B_{ij}$
denote the Cotton, Schouten and Bach tensor, respectively, of $h_{ij}$. Recall
that they are defined by 
%\marc{Added. Check conventions. Not sure what you call ``Bach tensor''. Please add definition}
\begin{align}
\widehat C_{ijk} &:= \widehat \nabla_k \widehat L_{ij} - \widehat \nabla_j
\widehat L_{ik}\;, \quad \quad \widehat L_{ij} := \widehat R_{ij} - 
\frac{1}{4} \widehat R h_{ij} \;, \label{Cotton+Schouten}\\
\widehat B_{ij} & := -\widehat\nabla^k\widehat C_{ijk} = \widehat\nabla^k\widehat \nabla_i\widehat L_{jk} -
\widehat \nabla_k \nabla^k \widehat L_{ij}\label{Bach} 
\;.
\end{align}
Note that due to \eq{christoffel} the actions of $\widetilde \nabla_t$ and $\partial_t$, as well as $\widetilde \nabla_i$ and $\widehat \nabla_i$, respectively, coincide on $\scri^-$, so we can use them interchangeably.

Whenever $X^{\mu}$ is  a KVF of the physical spacetime, the vector field 
\begin{equation}
Y^i := \widetilde X^i|_{\scri^-}
\end{equation}
 is a conformal Killing vector field (CKVF) 
 of $(\scri^-, h_{ij})$, i.e.\
\begin{equation}
(\mcL_Y h)_{ij} \, \equiv \, 2\widehat\nabla_{(i}Y_{j)} \,=\, \frac{2}{3} \widehat\nabla_kY^k h_{ij} 
\;,
\label{CKVF_eqn}
\end{equation}
 which fulfills the \emph{KID equations} \cite{ttp2}
\begin{equation}
   \mcL_Y D_{ij} + \frac{1}{3}D_{ij}\widehat \nabla_k Y^k = 0
\;,
\label{reduced_KID}
\end{equation}
and vice versa: 
%\marc{I don't understand the meaning of the phrase
%``in the unphysical spacetime'' in the theorem
%\\ -- \\
%tim: you apply the Choquet-Bruhat-Geroch argument in the unphysical spacetime. It's actually not clear to me what one  can say about the physical spacetime
%(this ``issue'' already appears in Theorem 2.1)}

%\marc{A symbol for covariant derivative in (\ref{2.51}) changed}
\begin{theorem}[\cite{ttp2}]
Let   $(\Sigma,h_{ij})$ be a Riemannian 3-manifold, $D_{ij}$ a symmetric $(0,2)$-tensor on $\Sigma$ and $\Lambda>0$.
Then, the tuple
 $(\Sigma,h_{ij}, D_{ij},  Y^i)$  defines
an  (up to isometries) unique, in the unphysical spacetime maximal globally hyperbolic
$\Lambda$-vacuum spacetime with a smooth $\scri^-$, represented %in the unphysical space-time
by  $\iota(\Sigma)$, with $\iota^* \widetilde g_{ij}|_{\Sigma}=h_{ij}$ and $\iota^* \widetilde d_{titj}|_{\Sigma}=D_{ij}$,
which contains a Killing vector field  $X$
 with  $\ol {\widetilde X}^i=Y^i$, if and only if  $D_{ij}$ is a TT tensor and  $Y$
is a conformal Killing vector field on    $(\Sigma,h_{ij})$ which satisfies the  KID equations
 (\ref{reduced_KID}). 
%
%\begin{equation}
% \mcL_{Y}D_{ij} + \frac{1}{3}D_{ij}  \widehat \nabla_k Y^k\,=\, 0
%\;.
%\end{equation}

Moreover, $\widetilde X^{\mu}$ satisfies
\begin{equation}
 \ol {\widetilde X}^t=0\;,  \quad  \ol {\widetilde \nabla_t \widetilde X^t} = \frac{1}{3}\widehat \nabla_i Y^i\;, \quad \ol{\widetilde \nabla_t \widetilde X^i}=0\;.
\label{rel_KVF}
\end{equation}
\end{theorem}

From what has been shown  in \cite{ttp2} one easily derives the following expressions on $\scri$,
\begin{eqnarray}
 \ol{\widetilde F} &=&\frac{4}{3}\widehat\nabla_i Y^i\;,
\label{Killing_rel_first}
\\
\Delta_h Y_i
% &=& -\widehat R_{ij}Y^j - \frac{1}{3}\widehat \nabla_i \widehat\nabla_j Y^j
%\\
 &=& -\widehat L_{ij}Y^j   - \frac{1}{4} \widehat R Y_i 
- \frac{1}{3}\widehat \nabla_i \widehat\nabla_j Y^j
\;,
\\
\Delta_h \ol{\widetilde   F} &=& -Y^i\widehat \nabla_i\widehat  R - \frac{1}{2}\widehat R\ol{\widetilde F}
\label{Killing_div}
%\;,
%\\
%\Delta \widehat\nabla_i Y^i &=&  - \frac{3}{4}Y^i\widehat \nabla_i \widehat R  -\frac{1}{2}\widehat R \widehat\nabla_i Y^i
\\
\ol{\widetilde\nabla_t\widetilde \nabla_t\widetilde X_t} % &=& \frac{1}{2}\ol{\widetilde \nabla_0 \widetilde F} \\
&=& 0
\;,
\\
\ol{\widetilde \nabla_t\widetilde \nabla_t\widetilde X_i} 
%&=& \Delta_h Y_i + 2\widehat L_{ij}Y^j + \frac{1}{2}\nabla_i F
%\\
&=&  \widehat L_{ij}Y^j - \frac{1}{4}\widehat R Y_i
+ \frac{1}{3}\widehat \nabla_i \widehat \nabla_j Y^j\;,
\\
\ol{\widetilde\nabla_t\widetilde\nabla_t\widetilde\nabla_t\widetilde X_t} &=& -\frac{1}{4}\Delta_h \ol{\widetilde  F}
\;,
\\
\ol{\widetilde\nabla_t\widetilde\nabla_t\widetilde\nabla_t\widetilde X_k} &=& 
%-Y^l\ol{\widetilde\nabla_t\widetilde R_{tktl}} \,=\,  
-   2 \sqrt{\frac{\Lambda}{3}} D_{kl} Y^l
\;,
\\
\ol{\widetilde \nabla_t \widetilde F} &=& %\nabla_0\nabla_{\mu} X^{\mu} \,=\,  
%- \nabla_0\nabla_{0} X_{0}\,=\, - \frac{1}{2} \nabla_0 F?
0
\;,
\\
\ol{\widetilde \nabla_t\widetilde \nabla_t \widetilde F} &=&\Delta_h \ol{\widetilde  F}
\;.
\label{Killing_rel_last}
\end{eqnarray}

\subsection{The function $Q$}
\label{functionQ}
%\subsection{Vanishing of the Mars-Simon tensor on $\scri$}

\subsubsection{A necessary condition for vanishing MST}
%\jose{Title changed}
\label{subsec_nec_cond}
Our aim is to characterize initial data on a spacelike $\scri^-$ which lead to a vanishing MST.
We have not specified the function $Q$ yet. Nonetheless, let us assume for the time being that $\Theta^{-4} Q$ does not tend to zero 
at
$\scri$. Then, it follows from \eq{formula_C} and \eq{formula_Q} that a necessary condition for the MST to vanish on $\scri$ is
\begin{equation}
\Big[\widetilde{\mathcal{H}}_{\mu\nu}\widetilde{\mathcal{H}}_{\sigma}{}^{\rho} 
-\frac{1}{3}\widetilde{\mathcal{H}}^2 \widetilde{\mathcal{I}}_{\mu\nu\sigma}{}^{\rho}\Big]\Big|_{\scri} \,=\, 0
\;.
\end{equation}
A straightforward computation  on a spacelike $\scri$
in the wave map gauge \eq{gauge_conditions_compact}
%  in adapted coordinates $(x^0, x^i)$ with $\widetilde g_{00}|_{\scri^-}=-1$ and $\widetilde g_{0i}|_{\scri^-}=0$ 
shows that this is the case if and only if
\begin{equation}
0
%\,=\,   \Big[\widetilde H_{0i}\widetilde H_{0j}   +\frac{1}{6}\widetilde H^2 \widetilde g_{ij}\Big]\Big|_{\scri}
\,=\, \Big[\widetilde{\mathcal{H}}_{ti}\widetilde{\mathcal{H}}_{tj} - \frac{1}{3}\widetilde{\mathcal{H}}^2\widetilde{\mathcal{I}}_{titj}\Big]\Big|_{\scri}
\,=\,  \frac{\Lambda}{3} (Y_iY_j)_{\mathrm{tf}}
\quad 
\Longleftrightarrow \quad  Y^i\,=\,0
\;.
\end{equation}
This already implies \cite{ttp2} that the KVF $X^{\mu}$ is trivial.
Hence, $\Theta^{-4} Q$ must necessarily go to zero 
whenever the  MST vanishes on a spacelike $\scri$. In the next section
we in fact show that
\begin{equation}
 Q\,=\,O(\Theta^5)
\end{equation}
holds automatically for an appropriate definition of $Q$.

%\subsection{The function $Q$}

\subsubsection{Definition and asymptotic behavior of the MST}
%\jose{Rephrased to avoid repetitions with previous subsection}
%\marc{Some text added}
In order to analyze
the situation where $\mathcal{S}_{\mu\nu\sigma\rho}$ vanishes, it is natural
to define $Q$ in such a way that a certain scalar constructed from
$\mathcal{S}_{\mu\nu\sigma\rho}$ vanishes automatically. This tensor has the
same algebraic properties as the Weyl tensor, so all its traces
are identically zero and cannot be used to define $Q$.
A convenient choice is to require
\begin{equation}
\mathcal{S}_{\mu\nu\sigma\rho}\mathcal{F}^{\mu\nu}\mathcal{F}^{\sigma\rho} \,=\,0
\;,
\end{equation}
%
%Contracting \eq{dfn_mars-simon} with $\mathcal{F}^{\mu\nu}\mathcal{F}^{\sigma\rho}$ and assuming
%that the MST vanishes yields
%\tim{make use of the relations xxx}
%\tim{last eqn... sign}
%\tim{do sth with $\mathcal{F}$}
%
or, equivalently,
\begin{equation}
 Q\mathcal{F}^{4}\,=\, \frac{3}{2} \mathcal{F}^{\mu\nu}\mathcal{F}^{\sigma\rho} \mathcal{C}_{\mu\nu\sigma\rho}
\,=\,  6F^{\mu\nu} F^{\sigma\rho} \mathcal{C}_{\mu\nu\sigma\rho} 
\label{definition_Q}
%\\
% &=& -\frac{3}{2}\mathcal{F}^{-4}   \mathcal{Q}_{\mu\nu\sigma\rho} \mathcal{C}^{\mu\nu\sigma\rho} 
%\\
%&=& \sqrt{\frac{3}{2}}\mathcal{F}^{-2}   \sqrt{\mathcal{C}_{\mu\nu\sigma\rho} \mathcal{C}^{\mu\nu\sigma\rho} }
\;.
\end{equation}
The function $Q$ necessarily needs to satisfy \eq{definition_Q}
%This choice of the function $Q$ is necessary 
whenever  the MST  vanishes.
%\tim{mention that $\mathcal{F}^2$ cannot have zeros in the $\Lambda>0$-case, so that this is not restrictive in the $\Lambda>0$-case?}
Let us restrict attention to the case where $\mathcal{F}^2$ has no zeros.
In fact,  $\mathcal{F}^2=-\frac{4}{3}\Lambda \Theta^{-2} |Y|^2 + O(\Theta^{-1})$, so, at least sufficiently close to $\scri$,  it suffices to assume that $Y$ has no zeros on $\scri$.
%\tim{what happens away from $\scri^-$?}
Then \eq{definition_Q} determines $Q$.  
%we regard \eq{definition_Q}   as the definition of $Q$
%(it is equivalent to \cite[Equation (24)]{mars_senovilla}).
From now on this choice of $Q$ will be denoted by $Q_0$,
\begin{equation}
 Q_0\,:=\, \frac{3}{2} \mathcal{F}^{-4}\mathcal{F}^{\mu\nu}\mathcal{F}^{\sigma\rho} \mathcal{C}_{\mu\nu\sigma\rho}
%\,=\, - 6\mathcal{F}^{-4}F^{\mu\nu} F_{\rho}{}^{\sigma} \mathcal{C}_{\mu\nu\sigma}{}^{\rho} 
\;,
\label{definition_Q0}
\end{equation}
and the corresponding MST by $\mathcal{S}^{(0)}_{\mu\nu\sigma\rho}$. 
When we want to emphasize the metric $g$ with respect to
which $\mathcal{S}^{(0)}_{\mu\nu\sigma\rho}$ is defined, we will write
$\mathcal{S}^{(0)}_{\mu\nu\sigma\rho}[g]$.

As has already been done for the other fields appearing in the definition of the MST, we express $Q_0$ in terms of the unphysical fields.
First of all we set
\begin{equation}
\widetilde{ \mathcal{D}}_{\mu\nu\sigma\rho} \,=\, \Theta^{-1}  \widetilde{\mathcal{C}}_{\mu\nu\sigma\rho}
\;.
\end{equation}
Making use of the various relations \eq{formula_C}-\eq{formula_XH} we find that\footnote{Not all orders given here and later in several instances are needed for our calculations. Nevertheless, we have chosen to write them down for the sake of completeness.}
%\tim{not all orders needed\\ -- \\ Jose: I have added a footnote}
%
\begin{eqnarray}
 Q_0 &=&  - 6\mathcal{F}^{-4} F^{\mu\nu} F_{\rho}{}^{\sigma} \mathcal{C}_{\mu\nu\sigma}{}^{\rho} 
\\
&=&   6\Theta^4 \widetilde{\mathcal{C}}_{\mu\nu\sigma\rho}\frac{(\widetilde H^{\mu\nu}  + \widetilde F^{\mu\nu}\Theta  )(\widetilde H^{\sigma\rho}  + \widetilde F^{\sigma\rho}\Theta  )  }{[ \widetilde{\mathcal{H}}^2 + 2 {\widetilde{ \mathcal{F}}}_{\alpha\beta}\widetilde{\mathcal{H}}^{\alpha\beta}\Theta 
+\widetilde{ \mathcal{F}}^2\Theta^2  ]^2}
\\
&=&   6\Theta^5 \widetilde{\mathcal{D}}_{\mu\nu\sigma\rho}\frac{
\widetilde H^{\mu\nu}  \widetilde H^{\sigma\rho}  + 2\widetilde H^{\mu\nu}\widetilde F^{\sigma\rho}\Theta  
 + \widetilde F^{\sigma\rho}  \widetilde F^{\mu\nu}\Theta ^2
 }{[ \widetilde {\mathcal{H}}^2  + 2 {\widetilde{ \mathcal{F}}}_{\alpha\beta}\widetilde{\mathcal{H}}^{\alpha\beta}\Theta 
+\widetilde{ \mathcal{F}}^2\Theta^2  ]^2}
\\
&=&   \frac{3}{2}\Theta^5  \widetilde H^{-4}\widetilde{\mathcal{D}}_{\mu\nu\sigma\rho}\Big(
 \widetilde H^{\mu\nu}  \widetilde H^{\sigma\rho}   
+ 
 2\widetilde H^{\mu\nu}\widetilde F^{\sigma\rho}\Theta 
 - 2\widetilde H^{-2}  \widetilde H^{\mu\nu}  \widetilde H^{\sigma\rho} {\widetilde{ \mathcal{F}}}_{\alpha\beta}\widetilde{\mathcal{H}}^{\alpha\beta}   \Theta 
\Big)
\nonumber
\\
&&
+ O(\Theta^7)
\;.
\end{eqnarray}
Using  $\Lambda>0$  and the relations \eq{relation_H2}-\eq{relation_FH}, which in particular imply
\begin{equation}
  \widetilde H^{-2}
\,=\,
- \frac{3}{2}\Lambda^{-1} \widetilde X^{-2}
%-9\Lambda^{-2}|\widetilde X|^{-2} \widetilde s \Theta 
+O(\Theta^2)
\end{equation}
(note that $\widetilde s=O(\Theta)$  due to \eq{gauge_conditions_compact}),
 we find the following expression for $Q_0$,
%\tim{not all orders needed}
%\tim{MCFE, unphys Killing}
%
\begin{eqnarray}
 Q_0
&=&   \frac{27}{8}\Theta^5 \Lambda^{-2} \widetilde X^{-4} \widetilde{\mathcal{D}}_{\mu\nu\sigma\rho}\Big(
 \widetilde H^{\mu\nu}  \widetilde H^{\sigma\rho}   
+ 
 2\widetilde H^{\mu\nu}\widetilde F^{\sigma\rho}\Theta 
\nonumber
\\
&&
\hspace{6em}
 +6i\Lambda^{-1} \widetilde X^{-2}  \widetilde H^{\mu\nu}  \widetilde H^{\sigma\rho} \widetilde H^{\alpha\beta} \widetilde F^{\star}_{\alpha\beta} \Theta 
\Big)
+ O(\Theta^7)
\;.
\label{expansion_Q}
\end{eqnarray}
We conclude that, in the wave map gauge \eq{gauge_conditions_compact},
\begin{equation}
 (\Theta^{-5} Q_0) |_{\scri}
\,=\,   \frac{27}{8}\Lambda^{-2} \widetilde X^{-4} \widetilde{\mathcal{D}}_{\mu\nu\sigma\rho}
 \widetilde H^{\mu\nu}  \widetilde H^{\sigma\rho}   
\,=\,   \frac{9}{2}\Lambda^{-1}  |Y|^{-4} Y^{i}Y^{j}\widetilde{\mathcal{D}}_{titj}
\;.
\label{leading_order_Q}
\end{equation}
%
%Let us also compute its transverse derivative,
%\tim{not needed anymore}
%%
%\begin{eqnarray}
% \partial_0 (\Theta^{-5}Q_0)
%&=&   -4(\Theta^{-5}Q)|_{\scri} |\widetilde X|^{-2}\widetilde X^{\mu} \widetilde \nabla_0 \widetilde X_{\mu}
%+ \frac{27}{8} \Lambda^{-2} |\widetilde X|^{-4}   \widetilde H^{\mu\nu}  \widetilde H^{\sigma\rho}  \widetilde\nabla_0\widetilde{\mathcal{D}}_{\mu\nu\sigma\rho}
%\nonumber
%\\
% &&+ \frac{27}{4} \Lambda^{-2} |\widetilde X|^{-4} \widetilde{\mathcal{D}}_{\mu\nu\sigma\rho}\Big(
% \widetilde H^{\mu\nu}   \widetilde\nabla_0\widetilde H^{\sigma\rho}   
%+ 
% \widetilde H^{\mu\nu}\widetilde F^{\sigma\rho} \widetilde\nabla_0\Theta
%\nonumber
%\\
%&& 
% +3i\Lambda^{-1} |\widetilde X|^{-2}  \widetilde H^{\mu\nu}  \widetilde H^{\sigma\rho} \widetilde H^{\alpha\beta} \widetilde F^{\star}_{\alpha\beta} \widetilde \nabla_0\Theta 
%\Big)
%+ O(\Theta)
%\;,
%\end{eqnarray}
%%
%whence
%%\tim{the last sign was wrong in (1.65)}
%%
%\begin{eqnarray}
% \partial_0 (\Theta^{-5}Q_0)|_{\scri}
%&=&   \frac{9}{2} \Lambda^{-1} |Y|^{-4} Y^iY^j   \widetilde  \nabla_0\widetilde{\mathcal{D}}_{0i0j}
%- \frac{9}{2} i\Lambda^{-1} |Y|^{-4}  Y^i N^j \widetilde{\mathcal{D}}_{0i0j}
%\nonumber
%\\
% &&+ 9 i\Lambda^{-1} |Y|^{-6}Y^iY^j Y^kN_k \widetilde{\mathcal{D}}_{0i0j} 
%\;.
%\end{eqnarray}
%
%Here 
%%
%\begin{equation}
% N^k \,=\, \mathrm{curl} \,Y^k \,=\,\widehat\volform^{ijk}\widehat\nabla_iY_j
%\end{equation}
%%
%denotes the curl of $Y$.

\subsection{Properties of the MST on $\scri$}

\begin{proposition}
\label{prop_reg_S}
 Consider a spacetime $(\mcM,g)$, solution to Einstein's vacuum field equations with  $\Lambda>0$, which admits a smooth conformal extension through $\scri$ and
which contains  a KVF $X$ with $\widetilde X^2|_{\scri} > 0$.
%\tim{non-trivial near $\scri$}
%\tim{!!}
Then  the MST
%\marc{reference to the metric dropped}
$\mathcal{S}^{(0)}_{\mu\nu\sigma}{}^{\rho}[\Theta^{-2}\widetilde g_{\alpha\beta}]$ 
corresponding to $X$ with $Q=Q_0$ defined by \eq{definition_Q0} vanishes on $\scri$.%
%\footnote{
%$|\widetilde X|^2\ne 0$ implies that $|X|^2$ and $\mathcal{F}^2$ are non-vanishing  at least sufficiently ``close'' to $\scri$, so $Q$ is indeed well-defined.
%}
%\tim{?}
\end{proposition}
\begin{proof}
The Weyl tensor is known to vanish on $\scri$.
%\tim{add reference}
Since  $\mathcal{U}_{\mu\nu\sigma}{}^{\rho}=O(\Theta^{-4})$ by \eq{formula_Q} and $Q_0=O(\Theta^5)$ by \eq{expansion_Q},
the lemma is proved.
\qed
\end{proof}

\begin{corollary}
 The rescaled MST
 $$\widetilde{\mathcal{T}}^{(0)}_{\mu\nu\sigma}{}^{\rho}[\Theta, \widetilde g_{\alpha\beta}] 
:=\Theta^{-1} \mathcal{S}^{(0)}_{\mu\nu\sigma}{}^{\rho}[\Theta^{-2}\widetilde g_{\alpha\beta}]$$
is regular at $\scri$.
\end{corollary}

\subsection{The rescaled MST on $\scri$}
\label{section_recaled_MS}

In this section we determine  the behavior of the rescaled MST $\widetilde{\mathcal{T}}^{(0)}_{\mu\nu\sigma}{}^{\rho}$ at $\scri$.
%We make the convention that indices of physical fields such as $\mathcal{T}^{(0)}_{\mu\nu\sigma}{}^{\rho}$ are raised and lowered using the \emph{unphysical} metric (as ``natural'' index positions we take the ones in \eq{formula_C}-\eq{formula_Q}).
For the tensor $\mathcal{U}_{\mu\nu\sigma}{}^{\rho} $ we find, using
(\ref{formula_Q}), (\ref{relation_H2}) and (\ref{relation_FH})
%\tim{not all orders needed}
%
\begin{eqnarray}
\Theta^4\mathcal{U}_{\mu\nu\sigma}{}^{\rho}
%  &=& \Theta^{2}\widetilde{\mathcal{U}}_{\mu\nu\sigma}{}^{\rho} -\Theta \big( 
% \widetilde{ \mathcal{F}}_{\mu\nu}\widetilde{\mathcal{H}}_{\sigma}{}^{\rho} 
%+ 
%\widetilde{\mathcal{H}}_{\mu\nu}\widetilde{ \mathcal{F}}_{\sigma}{}^{\rho}
%-\frac{2}{3} {\widetilde{ \mathcal{F}}}_{\alpha\beta}\widetilde{\mathcal{H}}^{\alpha\beta}\widetilde{\mathcal{I}}_{\mu\nu\sigma}{}^{\rho}  \big)
%\nonumber
%\\
%&&
%  -\big( \widetilde{\mathcal{H}}_{\mu\nu}\widetilde{\mathcal{H}}_{\sigma}{}^{\rho} 
%-\frac{1}{3}\widetilde{\mathcal{H}}^2 \widetilde{\mathcal{I}}_{\mu\nu\sigma}{}^{\rho} \big)
%\nonumber
% \\
 &=&
  -\big( \widetilde{\mathcal{H}}_{\mu\nu}\widetilde{\mathcal{H}}_{\sigma}{}^{\rho} 
+\frac{4}{9} \Lambda \widetilde X^2 \widetilde{\mathcal{I}}_{\mu\nu\sigma}{}^{\rho} \big)
\nonumber
\\
&&
 -\Theta \big( 
\widetilde{ \mathcal{F}}_{\mu\nu}\widetilde{\mathcal{H}}_{\sigma}{}^{\rho} 
+ 
\widetilde{\mathcal{H}}_{\mu\nu}\widetilde{ \mathcal{F}}_{\sigma}{}^{\rho}
-\frac{4}{3}i  \widetilde H^{\alpha\beta} \widetilde F^{\star}_{\alpha\beta}\widetilde{\mathcal{I}}_{\mu\nu\sigma}{}^{\rho}  \big)
 + O(\Theta^2)
\;.
\phantom{xx}
\label{theta4U}
\end{eqnarray}
%
%\marc{ Is (2.79) necessary? It seems to me that it is not used, and in fact
%it is repeated in (2.81)
%\\ -- \\ tim: I'm a bit confused  which eqns you're refering to
%\\ -- \\ marc: the numbers where messed up from a previous version. Slight
%change right before (2.71)
%\\ -- \\
%tim: indeed, there was one eqn which was not needed, I've removed it plus some reorganizations
%} 
%Employing the various relations collected in Section~\ref{constraints}, it follows that, in the wave map gauge \eq{gauge_conditions_compact},
%
%\begin{eqnarray}
%\label{mathycalU}
%(\Theta^4\mathcal{U}_{i0j}{}^{0})|_{\scri}  &=&   \frac{\Lambda}{3}(Y_iY_j)_{\mathrm{tf}}
%%%%%%%%\;,
%\\
%\widetilde \nabla_0(\Theta^4\mathcal{U}_{0i0j}]|_{\scri} 
% &=&
%-\partial_0 \Theta \big(   i \widetilde{ F}^{\star}_{0 i}\widetilde H_{0 j} 
%+ 
%i\widetilde H_{0 i}  \widetilde{ F}^{\star}_{0 j}+\frac{2}{3}i  \widetilde H^{0k} \widetilde F^{\star}_{0k}h_{ij}  \big)
%\\
%&&
%  +2 \widetilde H_{0 (i} \nabla_{|0|}\widetilde{H}_{ j)0} 
%  +2 i\widetilde H_{0 (i} \widetilde \nabla_{|0|}\widetilde{H}^{\star}_{j)0} 
%\nonumber
%\\
% &=&
%i \frac{\Lambda}{3} (Y_{(i} N_{j)})_{\mathrm{tf}}
%%%%%
%\;.
%\end{eqnarray}
Now we are ready to evaluate the rescaled MST  $\widetilde{\mathcal{T}}^{(0)}_{\mu\nu\sigma\rho} =\widetilde g_{\rho\alpha}  
\widetilde{\mathcal{T}}^{(0)}_{\mu\nu\sigma}{}^{\alpha} $ on~$\scri$. From (\ref{leading_order_Q}) and (\ref{theta4U}),
\begin{equation}
\widetilde{\mathcal{T}}^{(0)}_{\mu\nu\sigma\rho}\big|_{\scri} \,=\,
 \widetilde {\mathcal{D}}_{\mu\nu\sigma\rho} 
-\frac{9}{2}\Lambda^{-1}  |Y|^{-4} Y^{i}Y^{j}\widetilde{\mathcal{D}}_{titj}
( \widetilde{\mathcal{H}}_{\mu\nu}\widetilde{\mathcal{H}}_{\sigma\rho} 
-\frac{1}{3}\widetilde{\mathcal{H}}^2 \widetilde{\mathcal{I}}_{\mu\nu\sigma\rho}) 
\;.
\end{equation}
Since the rescaled MST is a self-dual Weyl field, its independent components
on $\scri$
 are 
$ \widetilde{\mathcal{T}}^{(0)}_{titj}|_{\scri}$.
%\begin{equation}
%\widetilde{\mathcal{T}}^{(0)}_{0i0j}|_{\scri}\,=\, \widetilde {\mathcal{D}}_{0i0j}    -\frac{3}{2}
%  |Y|^{-4} Y^{k}Y^{l}\widetilde{\mathcal{D}}_{0k0l}(Y_iY_j)_{\mathrm{tf}}
%\;.
%\label{MS_scri}
%\end{equation}
%From (\ref{mathycalU}), i.e.
Employing the various relations collected in Section~\ref{constraints}, it follows that, in the wave map gauge \eq{gauge_conditions_compact},
\begin{eqnarray}
(
\widetilde{\mathcal{H}}_{ti}\widetilde{\mathcal{H}}_{tj}  -\frac{1}{3}\widetilde{\mathcal{H}}^2 \widetilde{\mathcal{I}}_{titj}) |_{\scri} &=& \frac{\Lambda}{3} (Y_iY_j)_{\mathrm{tf}}
\;,
\label{H_Y_relation}
\\
\widetilde{\mathcal{D}}_{titj}|_{\scri} &=&D_{ij} - i  \  \sqrt{\frac{3}{\Lambda}}    \widehat C_{ij}
\label{MS_scri2}
\;.
\end{eqnarray}
%
%We deduce that
%%
%\begin{equation}
%\widetilde{\mathcal{T}}^{(0)}_{0i0j}|_{\scri}\,=\, \widetilde {\mathcal{D}}_{0i0j}    -\frac{3}{2}
%  |Y|^{-4} Y^{k}Y^{l}\widetilde{\mathcal{D}}_{0k0l}(Y_iY_j)_{\mathrm{tf}}
%\;.
%\label{MS_scri}
%\end{equation}
%
%
%\begin{equation}
%\widetilde{\mathcal{D}}_{0i0j}|_{\scri} \,=\, D_{ij} - i  \  \sqrt{\frac{3}{\Lambda}}    \widehat C_{ij}
%\;,
%\label{MS_scri2}
%\end{equation}
%
Here $\widehat C_{ij}$ denotes the Cotton-York tensor
\begin{equation}
 \widehat C_{ij} \,=\, -\frac{1}{2}\widehat \volform_i{}^{kl}\widehat C_{jkl} \quad \Longleftrightarrow \quad
\widehat C_{ijk} \,=\,  -\widehat\volform_{jk}{}^{l}\widehat C_{il} \label{cotton-york}
\;,
\end{equation}
which is a TT tensor,
%\jose{added} 
and  $\widehat \volform_{jkl}$ denotes the canonical volume 3-form relative to $h_{\ij}$.

Note that $D_{ij}$ and $\widehat C_{ij}$ correspond to the asymptotic electric and magnetic part, respectively, of the conformal
Weyl tensor.
We observe that  \eq{MS_scri2} immediately implies that $\scri$ will be  
locally conformally flat, i.e.\ has  vanishing Cotton-York tensor, if and only if the
magnetic part of the rescaled Weyl tensor $\widetilde d_{\mu\nu\sigma\rho}$  vanishes at $\scri$, cf.\ \cite{ashtekar}.

%We also compute the transverse derivative of the rescaled MST,
%%\tim{last sign!!}
%%
%\begin{eqnarray}
%\widetilde\nabla_0 \mathcal{T}^{(0)}_{0i0j}|_{\scri}
%&=&
%\widetilde\nabla_0 \widetilde{\mathcal{D}}_{0i0j} + \widetilde\nabla_0(\Theta^{-5}Q)(\Theta^4 \mathcal{U}_{0i0j})
%\nonumber
%\\
%&&
%+ (\Theta^{-5}Q)\widetilde\nabla_0(\Theta^4 \mathcal{U}_{0i0j})
%\\
%&=& \widetilde\nabla_0\widetilde {\mathcal{D}}_{0i0j} 
% +i \frac{3}{2}   |Y|^{-4} Y^{k}Y^{l}\widetilde{\mathcal{D}}_{0k0l} (Y_{(i}N_{j)})_{\mathrm{tf}}
%\nonumber
%\\
%&&    -\frac{3}{2} |Y|^{-4} (Y_iY_j)_{\mathrm{tf}} \Big(   Y^kY^l   \widetilde  \nabla_0\widetilde{\mathcal{D}}_{0k0l} -i  Y^kN^l\widetilde{\mathcal{D}}_{0k0l}
%\nonumber
%\\
%&&
%+ 2 i |Y|^{-2}Y^kY^l Y^mN_m\widetilde{\mathcal{D}}_{0k0l} 
%\Big)
%\;,
%\label{0MS_scri}
%\end{eqnarray}
%%
%with 
%%
%\begin{equation}
%\widetilde\nabla_0 \widetilde{\mathcal{D}}_{0i0j}|_{\scri} \,=\, \widehat\volform_i{}^{kl}\widehat\nabla_k \Big( \sqrt{\frac{3}{\Lambda}}\widehat C_{jl} +i  D_{jl}   \Big)
%\label{0MS_scri2}
%\end{equation}
%%
%(note that $\widehat B_{ij}\,=\, - \widehat \nabla^k \widehat C_{ijk} \,=\,   \widehat\volform_{i}{}^{kl} \widehat \nabla_k\widehat C_{lj}$).

\begin{proposition}
\label{rescaledMST_scri}
 Consider a spacetime $(\mcM,g)$, solution to Einstein's vacuum field equations with  $\Lambda>0$, which admits a smooth conformal extension through $\scri$ and
which contains  a KVF $X$ with $\widetilde X^2|_{\scri} > 0$.
Then, the rescaled MST $\widetilde{\mathcal{T}}^{(0)}_{\mu\nu\sigma}{}^{\rho}$ satisfies
\begin{equation*}
\widetilde{\mathcal{T}}^{(0)}_{titj}\big|_{\scri}\,=\, D_{ij} -\frac{3}{2}   |Y|^{-4} Y^{k}Y^{l} D_{kl}(Y_iY_j)_{\mathrm{tf}} 
 - i  \  \sqrt{\frac{3}{\Lambda}} \Big(   \widehat C_{ij}  
-     \frac{3}{2}   |Y|^{-4} Y^{k}Y^{l} \widehat C_{kl}(Y_iY_j)_{\mathrm{tf}}\Big)
\;.
\end{equation*}
(Recall that $\widetilde{\mathcal{T}}^{(0)}_{titj}\big|_{\scri}$ comprises all independent components.)
\end{proposition}

According to Proposition~\ref{rescaledMST_scri},
%\eq{MS_scri}-\eq{MS_scri2},
the rescaled MST vanishes on $\scri$ if and only if
\begin{eqnarray}
 D_{ij} - \frac{3}{2}  |Y|^{-4}     Y^{k} Y^{l} D_{kl}   (Y_iY_j)_{\mathrm{tf}} &=& 0\;,
\label{condition_D}
\\
 \widehat C_{ij}
- \frac{3}{2} |Y|^{-4}     Y^{k} Y^{l}   \widehat C_{kl} (Y_iY_j)_{\mathrm{tf}} &=& 0
\;.
\label{condition_C}
\end{eqnarray}
%
%Note that the vanishing of the Cotton-York tensor is equivalent to the vanishing of the Cotton tensor and thus to the property that the manifold is locally conformally flat.
%
We solve \eq{condition_D} on $\scri^{-}$.
\eq{condition_C} can be treated in exactly the same manner.

We define 
%\marc{$D_i$ was not necessary, right?}
%
\begin{equation}
% D_i \,:=\, \widehat \nabla_j D_{i}{}^j\;, \quad 
d\,:=\, Y^iY^jD_{ij}
\;.
\end{equation}
Applying $\widehat\nabla^j$ to \eq{condition_D} and employing the fact that the constraints equations enforce $D_{ij}$ to be a TT tensor, we are led to the equation
\begin{equation}
 Y_i \Big(Y^j \widehat \nabla_j  d 
+ \frac{1}{3} d \widehat\nabla_j Y^j \Big)
 - \frac{1}{3} |Y|^2 \widehat \nabla_i  d
- \frac{1}{6}d\widehat \nabla_i |Y|^2 \,=\, 0
\;,
\label{divergence_eqn}
\end{equation}
after using the following two consequences of the conformal Killing equation for $Y$,
\begin{align}
Y^j \widehat \nabla_j |Y|^2 & = \frac{2}{3} |Y|^2 \widehat \nabla_l Y^l, \\
Y^j \widehat \nabla_j Y_i & = \frac{2}{3} Y_i \widehat \nabla_l Y^l 
- \frac{1}{2} \widehat \nabla_i |Y|^2.
\end{align}
Contraction of (\ref{divergence_eqn})  with $Y^i$ gives
\begin{equation}
 Y^j \widehat \nabla_j  d 
+ \frac{1}{3} d \widehat\nabla_j Y^j 
\,=\, 0
\;.
\label{divergence_eqn2}
\end{equation}
Inserting this into \eq{divergence_eqn} yields
\begin{equation}
  2 \widehat \nabla_i  d +d\widehat \nabla_i\log  |Y|^2=0
\;.
\label{divergence_eqn3}
\end{equation}
The general solution of this equation is, using that $\scri^{-}$ is connected,
%We observe that \eq{divergence_eqn3} implies \eq{divergence_eqn2} and is thus equivalent to \eq{divergence_eqn}.
%It follows that
%\tim{...}
%
\begin{equation}
% D_i=0  \quad \Longleftrightarrow \quad  \eq{divergence_eqn3} \text{ holds}  
%\quad\Longleftrightarrow \quad
 d = \frac{2}{3}  \Dconst |Y|^{-1}\;, \enspace   \Dconst=\mathrm{const.}
\end{equation}
It follows that necessarily
%\tim{change notation of the constants $B$, $C$? 
%\\ -- \\
%done}
%
\begin{equation}
  D_{ij}
\,=\,\Dconst  |Y|^{-5} (Y_i Y_j)_{\mathrm{tf}}   
\;,
\label{condition_on_D}
\end{equation}
which is, indeed, a TT tensor satisfying \eq{condition_D}:

\begin{lemma}
Let $(\Sigma,h)$ be an $n$-dimensional Riemannian manifold.
%equipped with a metric $h=h_{ij}\mathrm{d}x^i\mathrm{d}x^j$ of signature $(p,q)$.
Let $Y$ be a vector field on $\Sigma$ with $|Y|^2\ne 0$, and denote by $\widehat\nabla$ the connection associated to $h$.
Then $D_{ij}:=|Y|^{-n-2} (Y_iY_j)_{\mathrm{tf}}$ is  a TT-tensor if and only
\begin{equation}
Y^j(\widehat\nabla_{(i}Y_{j)})_{\mathrm{tf}}\, = \,0
\;.
\label{Y_TT-cond}
\end{equation}
(So in particular if $Y$ is a CKVF.)
\end{lemma}
\begin{proof}
We compute the divergence of $D_{ij}$,
\begin{equation}
\widehat\nabla^jD_{ij} \,=\,
- (n+2)|Y|^{-n-4} Y_iY ^j  Y^k(\widehat\nabla_{(j} Y_{k)})_{\mathrm{tf}} 
+2  |Y|^{-n-2} Y^j(\widehat\nabla_{(i} Y_{j)})_{\mathrm{tf}} 
\label{eqn_TT_cond}
\;,
\end{equation}
and  observe that  $D_{ij}$ is a TT-tensor if \eq{Y_TT-cond} holds.

Conversely, contraction of \eq{eqn_TT_cond} with $Y^i$ yields
\begin{equation}
Y^i\widehat\nabla^jD_{ij} \,=\,
-n  |Y|^{-n-2} Y^iY^j(\widehat\nabla_{(i} Y_{j)})_{\mathrm{tf}} 
\label{eqn_TT_cond2}
\;,
\end{equation}
which we insert into \eq{eqn_TT_cond},
\begin{equation}
\widehat\nabla^jD_{ij} \,=\,
 \frac{n+2}{n}|Y|^{-2} Y_iY^j\widehat\nabla^kD_{jk}
+2  |Y|^{-n-2} Y^j(\widehat\nabla_{(i} Y_{j)})_{\mathrm{tf}} 
\label{eqn_TT_cond3}
\;.
\end{equation}
It follows that if $D_{ij}$ is a TT-tensor then \eq{Y_TT-cond} holds, 
%and the lemma is proved.
which completes the proof of the lemma.
\qed
\end{proof}
%
%One checks that $|Y|^{-5}(Y_i Y_j)_{\mathrm{tf}} $ is a TT tensor if and only if
%%
%\begin{eqnarray*}
% 5Y_iY^jY^k(\widehat\nabla_{j}Y_{k})_{\mathrm{tf}} -2|Y|^2Y^j(\widehat\nabla_{(i}Y_{j)})_{\mathrm{tf}} = 0
%\;.
%\end{eqnarray*}
%%
%Contraction with $Y^i$ yields
%%
%$
%Y^jY^k(\widehat\nabla_{j}Y_{k})_{\mathrm{tf}}= 0
%$,
%%
%and we conclude that the condition
%%
%\begin{eqnarray*}
%Y^j(\widehat\nabla_{(i}Y_{j)})_{\mathrm{tf}} = 0
%\end{eqnarray*}
%%
%characterizes TT tensors of the form $|Y|^{-5}(Y_i Y_j)_{\mathrm{tf}} $.

Similarly, one shows that for some constant $\Cconst$
\begin{equation}
 \widehat C_{ij} \,=\, \CconstL |Y|^{-5}(Y_iY_j)_{\mathrm{tf}}
\;.
\label{condition_on_C}
\end{equation}

\begin{remark}
{\rm
If $Y^i$ is a CKVF, \eq{condition_on_D} defines, away from zeros
of $Y$, a TT-tensor $D_{ij}$ which satisfies the  KID equations
\eq{reduced_KID}. On the other hand, a solution of  \eq{reduced_KID}
always satisfies \eq{divergence_eqn2}.
}
\end{remark}

Up to this stage we had to assume that $|Y|^2>0$ on $\scri$. In fact, the above considerations reveal that this follows from
the assumption of the existence of a smooth $\scri$
whenever the rescaled tensor 
$\widetilde{\mathcal{T}}^{(0)}_{\mu\nu\sigma}{}^{\rho}$ 
vanishes there: 
The CKVF $Y$ is not allowed to vanish in some open region of $\scri$, because this would imply that
the corresponding KVF would vanish in the domain of dependence of that region. 
Let us assume that $|Y(p)|^2=0$ for some $p\in\scri$. Then it follows from \eq{condition_on_D}-\eq{condition_on_C} that, for either $\Dconst \ne 0 $ or $\Cconst \ne 0$,
%\tim{what about the de Sitter case?}
%\tim{$det (g)$?}
%\marc{expression with $B,C$ added}
%that
%
\begin{eqnarray}
 \widetilde d_{\mu\nu\sigma\rho} \widetilde d^{\mu\nu\sigma\rho}|_{\scri}
&=& (4\widetilde d_{t itj} \widetilde d^{titj}  + 4\widetilde d_{t ijk} \widetilde d^{tijk} 
+  \widetilde d_{ijkl} \widetilde d^{ijkl})|_{\scri}
\\
&=& 8D_{ij}D^{ij}- \frac{24}{\Lambda}\widehat C_{ij}\widehat C^{ij}
= \frac{16}{3 |Y|^6} \left ( \Dconst^2 - \Cconst^2 \right )
\end{eqnarray}
or
\begin{eqnarray}
 \widetilde d^{\star}_{\mu\nu\sigma\rho} \widetilde d^{\mu\nu\sigma\rho}|_{\scri}
&=&
 (4\widetilde\volform_{ti}{}^{jk}\widetilde d_{jktl} \widetilde d^{titl}
+2 \widetilde\volform_{ti}{}^{jk}\widetilde d_{jklm} \widetilde d^{tilm}
)|_{\scri}
\\
&=&
-16\sqrt{\frac{3}{\Lambda}} \widehat C_{ij} D^{ij}
= - \frac{32}{3 |Y|^6}  \Dconst\Cconst
\end{eqnarray}
(we used that $\widetilde d_{ijkl}|_{\scri} = -\widetilde\volform_{ij}{}^{tm}\widetilde\volform_{kl}{}^{tn}\widetilde d_{tmtn}$), diverges at $p$, so that $p$ actually cannot belong 
to the (unphysical) manifold. 
%\marc{Text added}
This argument does not apply when $\Dconst=\Cconst=0$. In this case the metric $h$ is conformally flat and 
$D_{ij}$ vanishes, so the data at $\scri^{-}$ correspond to data for the de
Sitter metric. The maximal de Sitter data is $\scri^- = \mathbb{S}^3$
with  $h$ the standard round metric. This space has ten linearly
independent conformal Killing vectors, which generically
%\tim{the question is whether the one wrt which the MST vanishes has zeros (is there more than one?)\\
%Jose: but, with definition (\ref{definition_Q0}) for $Q$, \underline{all} MST vanish identically in de Sitter.
%tim: footnote added }
vanish at some points. In this case the points where the conformal Killing vector
vanishes do belong to $\scri^{-}$.%
\footnote{
Note that in the de Sitter case we have $\mathcal{C}_{\alpha\beta\mu\nu}=0=Q_0$, so
% \emph{all} MSTs
the MST associated to \emph{any} KVF
 vanishes identically.
}
 This is why we need to 
exclude de Sitter explicitly in the following Theorem.

\begin{theorem}
\label{thm_nec_cond}
  Consider a spacetime $(\mcM,g)$ solution to Einstein's vacuum field equations with  $\Lambda>0$, which admits a smooth conformal extension through $\scri$ and
which contains  a KVF $X$. Denote  by $h$ the Riemannian metric induced by $\widetilde g=\Theta^2 g$  on $\scri$, and  by $Y$ the CKVF induced by $X$ on $\scri$.
Assume that $(\mcM,g)$ is not locally isometric to the de Sitter spacetime.
Then  $|Y|^2 > 0$, and the rescaled MST $\widetilde{\mathcal{T}}^{(0)}_{\mu\nu\sigma}{}^{\rho}=\Theta^{-1} \mathcal{S}^{(0)}_{\mu\nu\sigma}{}^{\rho}$ corresponding to $X$ with $Q=Q_0$ defined by \eq{definition_Q0}  vanishes on a connected
component $\scri^{-}$ of $\scri$ if and only if
the following relations hold: 
\begin{enumerate}
\item[(i)] $ \widehat C_{ij} = \CconstL |Y|^{-5}(Y_iY_j -\frac{1}{3}|Y|^2 h_{ij})$ for some constant $\Cconst$, where $\widehat C_{ij}$ is the Cotton-York tensor of the Riemannian  3-manifold $(\scri^{-}, h)$, and
\item[(ii)]    $D_{ij} = \widetilde d_{titj}|_{\scri^-} =\Dconst  |Y|^{-5}(Y_iY_j -\frac{1}{3}|Y|^2 h_{ij})$ for some constant $\Dconst$.
%\tim{meaning of $C$?}
\end{enumerate}
%
%In that case also the transverse derivative of the  rescaled MST vanishes on  $\scri$.
%\tim{last statement not needed}
%\tim{KVF may become trivial}
%The transverse derivative of the  Mars-Simon tensor vanishes on  $\scri$ if and only if, in addition,
%\begin{enumerate}
%\item[(iii)] $Y_i\, \mathrm{curl}\,Y^i =0$ or $B=C=0$.
%\end{enumerate}
\end{theorem}

%\tim{evolution equation needed to deduce the vanishing of  $\mathcal{T}_{\mu\nu\sigma}{}^{\rho}$ everywhere}

\section{The functions $c$ and $k$ and their restrictions to~$\scri$}

\subsection{The functions $c$ and $k$ and their constancy}
\label{section_constant}

Following \cite{mars_senovilla}, we define four real-valued functions
 $b_1$, $b_2$, $c$ and $k$  by the  system
(we make the assumption $Q\mathcal{F}^2-4\Lambda\ne 0$; later on it will become clear that this holds automatically near
a regular  $\scri$)
%\tim{``near a regular  $\scri$'' added}
%\tim{near $\scri$, or globally when supposing that the MST vanishes}
%\tim{emphasize later on}
%\tim{sign of $b_1$ and $b_2$}
%
\begin{eqnarray}
 b_2-ib_1 &=& - \frac{ 36 Q (\mathcal{F}^2)^{5/2} }{(Q\mathcal{F}^2-4\Lambda)^3}
\label{equation_b1b2}
\;,
\\
c &=& - X^2-  \mathrm{Re}\Big(  \frac{6\mathcal{F}^2(Q\mathcal{F}^2+2\Lambda)}{(Q\mathcal{F}^2-4\Lambda)^2} \Big) 
\;,
\label{equation_c}
\\
k  &=&  \Big|\frac{36\mathcal{F}^2}{(Q\mathcal{F}^2-4\Lambda)^2}\Big| \nabla_{\mu}Z\nabla^{\mu}Z -  b_2Z +cZ^2 +\frac{\Lambda}{3} Z^4 
\;,
\label{equation_k}
\end{eqnarray}
where
\begin{equation}
 Z = 6\,\mathrm{Re} \Big( \frac{\sqrt{\mathcal{F}^2}}{Q\mathcal{F}^2-4\Lambda}\Big)
\;.
\end{equation}
\label{rem_compl_sr}
We note that the expression (\ref{equation_k}) for $k$
as given in \cite{mars_senovilla} has two typos both in the statement 
of Theorem 1 and of Theorem 6.

A remark is in order concerning the appearance of square roots of the complex function $\mathcal{F}^2$. 
%\marc{Text added}
In the setting of \cite{mars_senovilla}, the function $\mathcal{F}^2$ is shown 
to be nowhere vanishing so we can prescribe 
the choice of square root at one point and extend it by continuity to the
whole manifold. Since $\mathcal{F}^2$ does not vanish, no branch point of 
the root is ever met and $\sqrt{\mathcal{F}^2}$ is smooth everywhere.
Moreover, the function $\mathcal{F}^2$ has strictly negative real
part in a neighborhood of $\scri$ (see (\ref{used_relation_F2}) below).
We can thus fix the square root $\sqrt{\mathcal{F}^2}$ in this neighborhood
%\tim{but if the MST does not vanish (and in this paper we do not restrict attention to vanishing MSTs), $\mathcal{F}^2$ might become zero off some neighborhood of $\scri$...?}
by choosing the positive branch near $\scri$, namely the branch that
takes positive real numbers and gives positive real values. We will
use this prescription for any function which is non-zero in a neighborhood
of infinity.

%square root.
% 
%We choose the principal square root on $\mathbb{C}$ with the negative imaginary% axis used as the branch cut.
%%At one point sufficiently close to  $\scri^-$ we choose the branch for which $5\sqrt{1}=1$, which fixes the square root globally. 
%t will be shown below that the real part of $\mathcal{F}^2$ does not change si5%gn
%near $\scri^{-}$, cf.\ \eq{used_relation_F2}, so that 
%$\sqrt{1}=1$ holds true in some neighborhood of $\scri^-$.
%the discontinuity of the square root does not cause  any problems.

The following result is proven in \textcolor{blue}{
\cite{mars_senovilla}. Unfortunately, the corresponding statements in 
Theorems 4 and 6 in that reference have a missing hypothesis. This mistake 
has been corrected in the arXiv version \cite{mars_senovilla_arx} of the
paper.}
%\cite[Theorems 4 \& 6]{mars_senovilla}
%
\begin{theorem}
\label{constancy}
Let $(\mcM,g)$ be a $\Lambda$-vacuum spacetime which admits a KVF $X$ such that the MST vanishes for some function
$Q$. Assume further that the functions $Q \mathcal{F}^2$ and $Q\mathcal{F}^2-4\Lambda$ are not identically zero \textcolor{blue}{and that
$y := \mbox{Re} \left ( \frac{6 i \sqrt{\mathcal{F}^2}}{Q \mathcal{F}^2 - 4 \Lambda} \right ) $
has non-zero gradient somewhere}. Then:
\begin{enumerate}
\item[(i)] $\mathcal{F}^2$ and $Q\mathcal{F}^2-4\Lambda$  are nowhere vanishing, 
%\marc{Why the question marks? \\ -- \\ Is it assumed that the spacetime is connected or sth like that? 
%\\ -- \\
%Jose: in \cite{mars_senovilla} the manifold is connected from the very beginning},
\item[(ii)] $Q$ is given by \eq{definition_Q0}, i.e.\ $Q=Q_0$,  and
\item[(iii)]  $b_1$, $b_2$, $c$ and $k$  are constant.
%\tim{and $\sigma_{\mu}$ is exact}
\end{enumerate}
\end{theorem}
%the Kerr-NUT-de-Sitter-parameters $m$, $a$ and $l$ can be determined from the constants
%
\begin{remark}
\label{rem_const}
{\rm
If $\Lambda>0$ and $(\mcM,g)$ admits a smooth $\scri$, has vanishing MST and
is not locally isometric to the de Sitter spacetime,
it follows from \eq{Taylor_Q}-\eq{used_relation_F2} below  that
$Q$, $\mathcal{F}^2$ and $Q\mathcal{F}^2-4\Lambda$ \textcolor{blue}{are non-zero
near $\scri$. Furthermore \eq{Taylory} implies that $y$ has non-zero
gradient in a neighbourhood of $\scri$. Thus, all the hypotheses of the theorem
hold and we conclude that }
% will never  be identically zero. 
%In other words, 
 $b_1$, $b_2$, $c$ and $k$  are constant whenever the MST vanishes in a spacetime $(\mcM,g)$ as above. 
%\marc{added that the spacetime needs to have conformal infinity}
}
\end{remark}

Combining Theorem~\ref{constancy} and Remark~\ref{rem_const} it follows that
a $\Lambda>0$-vacuum spacetime admitting a KVF $X$ with vanishing
associated MST for some $Q$ and for which $\mathcal{F}^2=0$ somewhere 
cannot admit a smooth $\scri$, unless the spacetime is locally isometric
to de Sitter. Although a priori interesting, this result turns out to be empty
since it has been proven in \cite{mars_senovilla_null} that all spacetimes 
with vanishing MST and null Killing form $\mathcal{F}$ (somewhere, and hence everywhere)
have necessarily $\Lambda \leq 0$.
%, unless the spacetime is locally isometric to de Sitter.

%\begin{lemma}
%%Let $(\mcM,g)$ be a $\Lambda>0$-vacuum spacetime
% which admits a KVF $X$ such that the associated MST vanishes for some $Q$.
%Assume further that the function $\mathcal{F}^2$ vanishes at least at one point
%and that $(\mcM,g)$ is not locally isometric to the de Sitter spacetime.
%Then $(\mcM,g)$ does not admit a smooth $\scri$.
%\end{lemma}

%This lemma is interesting because the classification 
%of spacetimes with vanishing MST splits into two cases, namely when $\mathcal{F}^2$ is not zero somewhere and the case when $\mat%hcal{F}^2 \equiv 0$. 
%Only the first case was
%considered in \cite{mars_senovilla}. %A classification of spacetimes in the second class is still lacking. 
%The result above tells a priori that none 
%5of the spacetimes 
%which belong to the second  class with $\Lambda>0$, if they exist,
%will admit a smooth conformal compactification.
%%\tim{remark that such a spacetime does not exist?}

The above functions \eq{equation_b1b2}-\eq{equation_k}, or rather their restrictions to $\scri$, turn out to be crucial for the classification of vacuum spacetimes 
with vanishing MST (and conformally flat $\scri$, cf.\ \cite{mpss}).
Our next aim will therefore be to find explicit expressions for them
in terms of the data at $\scri$
under the assumption that
the MST vanishes for some choice of $Q$. 
We wish to find expressions
at null infinity that make sense (and generally cease to be constant)
for any $\Lambda$-vacuum spacetime with a smooth conformal compactification
and a KVF.

Employing   the  relations collected in Sections~\ref{constraints},
\ref{functionQ} and \ref{section_recaled_MS}  we find that, 
under the assumption that the MST vanishes,
\begin{eqnarray}
 (\Theta^{-5} Q_0)
&=&   3\Lambda^{-1}  |Y|^{-5}\big(\Dconst  - i  \Cconst\big)+ O(\Theta)
\;,
\label{Taylor_Q}
\\
Q_0\mathcal{F}^2-4\Lambda &=&   -4\Lambda  + O(\Theta^3)
\label{expansion_QF_Lambda}
\;,
\\
\mathcal{F}^2 &=& -\frac{4}{3}\Lambda \Theta^{-2} \widetilde X^2  + 2\widetilde F^2  -\frac{1}{4}\widetilde F^2 + 2i {\widetilde F}^{\star}_{\alpha\beta}{\widetilde F}^{\alpha\beta}
\nonumber
\\
&& + 2  \widetilde X^{\alpha}\widetilde \nabla_{\alpha}\widetilde F 
+ 8 \widetilde X^{\alpha} \widetilde X^{\beta} \widetilde L_{\alpha\beta}
  +4i \Theta^{-1} \widetilde F^{\mu\nu} \widetilde H^{\star}_{\mu\nu}
\\
 &=& -\frac{4}{3}\Lambda \Theta^{-2} \widetilde X^2 
+ 4i \sqrt{\frac{\Lambda}{3}}\Theta^{-1}Y_kN^k 
 + |N|^2  -\frac{4}{9}f^2
\nonumber
\\
&& 
+ \frac{8}{3}  Y^{i}\widehat \nabla_{i}f
+ 8 Y^iY^j\widehat  L_{ij}
  +  O(\Theta)
\label{used_relation_F2}
%\\
% &=& -\frac{4}{3}\Lambda \Theta^{-2} |Y|^2
%-4i \sqrt{\frac{\Lambda}{3}}\Theta^{-1}Y_kN^k 
% +\widehat R |Y|^2
% + |N|^2 
%\nonumber
%\\
%&& 
%+ \frac{4}{3}  Y^{i}\widehat \nabla_{i}\widehat \nabla_{j}Y^j
%+ 4 Y^iY^j\widehat  L_{ij}
%  +  O(\Theta)
\;. \\
\textcolor{blue}{\frac{6i \sqrt{\mathcal{F}^2}}{Q_0 \mathcal{F}^2 - 4 \Lambda}}
& \textcolor{blue}{=} & \textcolor{blue}{ \frac{1}{\Theta} \sqrt{\frac{3}{\Lambda}} \sqrt{ 
\widetilde X^2 } + O(1)}.
\label{Taylory} 
\end{eqnarray}
Here 
\begin{equation}
 N^k \,:=\, \mathrm{curl} \,Y^k \,=\,\widehat\volform^{ijk}\widehat\nabla_iY_j \label{curl_of_Y}
\end{equation}
denotes the curl of $Y$, and
\begin{equation}
f \,:=\, \widehat\nabla_iY^i \label{div_of_Y}
\end{equation}
its divergence.

Let us determine  the trace of  \eq{equation_b1b2}-\eq{equation_k} in the unphysical, conformally rescaled spacetime on $\scri$ under the
assumption that the MST vanishes for some choice of $Q$.
With \eq{Taylor_Q} and \eq{used_relation_F2}, we observe that, on $\scri$, equation \eq{equation_b1b2} yields
%\tim{sign of $b_1$ and $b_2$, $\sqrt{-1}=i$...otherwise $\pm$}
%
\begin{equation}
(  b_2-ib_1)|_{\scri} \,=\,  \Big(\frac{9}{16}\Lambda^{-3} (\Theta^{-5}Q_0) (\Theta^2\mathcal{F}^2)^{5/2}\Big)\Big|_{\scri}
%\,=\,  9i\Lambda^{-2} \sqrt{ \frac{\Lambda}{3}}  |Y|Y^{i}Y^{j}\widetilde{\mathcal{D}}_{0i0j} 
\,=\,  \frac{2}{\Lambda} \sqrt{ \frac{3}{\Lambda}}   \big(      \Cconst
% + 6iC\Lambda^{-2} \sqrt{ \frac{\Lambda}{3}}   
 + i\Dconst \big)
\;.
\end{equation}
From \eq{equation_c} and \eq{used_relation_F2} we conclude that
%\tim{...which is indeed a constant, I think this is what Marc can show... assumptions $Y$ CKVF,  \eq{condition_on_C},...}
%
\begin{eqnarray}
\frac{\Lambda}{3} c\Big|_{\scri} &=& \frac{\Lambda}{3}\Big( -\Theta^{-2}\widetilde X^2 - \frac{3}{4}\Lambda^{-1}\mathrm{Re} ( \mathcal{F}^2 )\Big)\Big|_{\scri}
\\
&=& - \frac{1}{4}   |N|^2  + \frac{1}{9}f^2
- \frac{2}{3} Y^i\widehat  \nabla_{i}f
-2 Y^iY^j \widehat L_{ij}
\;.
\label{expression_c}
\end{eqnarray}
Note that this implies that
\begin{eqnarray}
\mathcal{F}^2 
&=& 
 -\frac{4}{3}\Lambda \Theta^{-2} \widetilde X^2   +4i \Theta^{-1} \widetilde F^{\mu\nu} \widetilde H^{\star}_{\mu\nu}-\frac{4}{3}\Lambda c 
\nonumber
\\
&&  -8D_{ij}Y^iY^j \Theta  - 4i \sqrt{\frac{3}{\Lambda}}N^k\widetilde \nabla_t\widetilde \nabla_t\widetilde X_k \Theta  + O(\Theta^2)
\label{used_relation_F2v2}
\\
 &=& -\frac{4}{3}\Lambda (\Theta^{-2} |\widetilde X|^2  +c)
+4i \sqrt{\frac{\Lambda}{3}}\Theta^{-1}Y_kN^k 
  +  O(\Theta)
\label{expansion_F2}
\;.
\end{eqnarray}
Next, we compute the function $Z$ (here an overbar means ``complex conjugation''),
\begin{eqnarray*}
 Z &=& \mathrm{Re} \Big( \frac{6\sqrt{\mathcal{F}^2}}{Q_0\mathcal{F}^2-4\Lambda}\Big)
\\
&=& \frac{6\, \mathrm{Re} \Big( \sqrt{\mathcal{F}^2}(\ol {Q_0\mathcal{F}^2}-4\Lambda)\Big)}{(Q_0\mathcal{F}^2-4\Lambda)(\ol {Q_0\mathcal{F}^2}-4\Lambda)}
\\
&=& \Big(\frac{3}{8}\Lambda^{-2}  + O(\Theta^3)\Big)\mathrm{Re} \Big( \sqrt{\mathcal{F}^2}(\ol {Q_0\mathcal{F}^2}-4\Lambda)\Big)
\\
&=&-\frac{3}{2}\Lambda^{-1} \Theta^{-1}\mathrm{Re} \big( \sqrt{\Theta^2\mathcal{F}^2}\big) +  O(\Theta^2)
%\\
%&=&\pm \frac{3}{2}\Lambda^{-1} \sqrt{\frac{1}{2}\Big(\mathrm{Re}(\mathcal{F}^2) +\sqrt{\mathcal{F}^2\ol{\mathcal{F}^2}} \Big)}+  O(\Theta^2)
\;.
\end{eqnarray*}
Equation \eq{expansion_F2}  yields
%\tim{formula was wrong in an earlier version}
%
\begin{equation}
\mathrm{Re} \big( \sqrt{\Theta^2\mathcal{F}^2}\big)  \,=\, |Y|^{-1}Y_kN^k\Theta + O(\Theta^3)
\;.
\end{equation}
Thus
\begin{equation}
 Z \,=\,  -\frac{3}{2}\Lambda^{-1} |Y|^{-1}Y_kN^k  +  O(\Theta^2)
\;,
\end{equation}
and we deduce from \eq{equation_k} that 
\begin{eqnarray*}
k|_{\scri} &=&\Big( \big| -  \frac{9}{4\Lambda^2}\Theta^2\mathcal{F}^2 (\widetilde\nabla_{t}Z)^2 
+\frac{9}{4\Lambda^2}\Theta^{2}\mathcal{F}^2 \widetilde\nabla_{i}Z\widetilde \nabla^i Z\big|  
-b_2Z +cZ^2+\frac{1}{3}\Lambda Z^4\Big)\Big|_{\scri}
\\
&=&\Big(\frac{3}{\Lambda}|Y|^2\widehat\nabla_{i}Z\widehat \nabla^i Z 
-b_2Z +cZ^2+\frac{1}{3}\Lambda Z^4\Big)\Big|_{\scri}
\;.
\end{eqnarray*}
From the conformal Killing equation for $Y$ we find that
\begin{eqnarray*}
\Lambda |Y| \widehat\nabla_i Z \big|_{\scri}  &=&  -\frac{1}{2}f|Y|^{-2}Y_iY_kN^k 
-\frac{3}{4}|Y|^{-2}Y_kN^k\widehat \volform_{ijl}Y^j N^l
 +  \frac{3}{2} \widehat\nabla_i(Y_kN^k)
\;,
\end{eqnarray*}
whence, using \eq{expression_c},  
\begin{eqnarray*}
\Lambda^2|Y|^2 \widehat\nabla_i Z \widehat\nabla^i Z\big|_{\scri} 
 &=&
   \frac{1}{4}f^2|Y|^{-2}(Y_kN^k )^2 
   -\frac{3}{4}f|Y|^{-2} Y^i\widehat\nabla_i(Y_kN^k)^2
\\
&& 
+\frac{9}{16}|Y|^{-2}(Y_kN^k)^2|N|^2
- \frac{9}{16}|Y|^{-4}(Y_kN^k)^4
\\
&&  -\frac{9}{8}|Y|^{-2}\widehat\volform_{ijl}Y^j N^l  \widehat\nabla^i(Y_kN^k)^2
 +\frac{9}{4} \widehat\nabla_i(Y_jN^j) \widehat\nabla^i(Y_kN^k)
\\
&=&
   \frac{1}{2}f^2|Y|^{-2}(Y_kN^k )^2 
   -\frac{3}{4}f|Y|^{-2} Y^i\widehat\nabla_i(Y_kN^k)^2
\\
&& 
- \frac{3}{2}|Y|^{-2}(Y_kN^k)^2 Y^i\widehat  \nabla_{i}f
- \frac{9}{2}|Y|^{-2}(Y_kN^k)^2 Y^iY^j \widehat L_{ij}
\\
&&  -\frac{9}{8}|Y|^{-2}\widehat\volform_{ijl}Y^j N^l  \widehat\nabla^i(Y_kN^k)^2
 +\frac{9}{4} \widehat\nabla_i(Y_jN^j) \widehat\nabla^i(Y_kN^k)
\\
&& -\frac{1}{3}\Lambda^3 c Z^2  
- \frac{1}{9}\Lambda^{4}Z^4
\;,
\end{eqnarray*}
and  
%\marc{B changed in terms of $Y$}
%
\begin{eqnarray}
\label{constk}
\Big(\frac{\Lambda}{3}\Big)^3 k \big|_{\scri}
&=& \frac{1}{18}|Y|^{-2}(Y_kN^k )^2 \Big(f^2
- 3Y^i\widehat  \nabla_{i}f
- 9 Y^iY^j \widehat L_{ij}
\Big)
+ \frac{1}{2}  (\widehat C_{ij} Y^i Y^j)
%\frac{1}{3}B  |Y|^{-1}
Y_kN^k
\nonumber
\\
&&  -\frac{1}{8}\widetilde\nabla_i\log |Y|^2 \widetilde\nabla^i(Y_kN^k)^2
 +\frac{1}{4} \widetilde\nabla_i(Y_jN^j) \widetilde\nabla^i(Y_kN^k)
\;.
\label{simplified_k}
\end{eqnarray}
where we have used $b_2  = 6  \Lambda^{-2}\CconstL$ and
have replaced $\Cconst$ in terms of the Cotton-York tensor $\Cconst= \frac{3}{2}\sqrt{\frac{3}{\Lambda}} |Y|
\widehat C_{ij} Y^i Y^j$.
Expression (\ref{constk}) 
provides a simplified formula for $k$ on $\scri$ in terms of $Y$.

It follows from Theorem~\ref{constancy} and Remark~\ref{rem_const} that the right-hand sides of \eq{expression_c} and \eq{simplified_k}
are constant, whenever the MST vanishes for  some function~$Q$.

\subsection{The functions $\widehat c(Y)$ and $\widehat k(Y)$}

In the previous section we have introduced the spacetime functions $c$ and $k$ and computed their
restrictions on $\scri$ in terms of the induced metric $h$
and the CKVF $Y$, whenever
the MST  vanishes.
Here we regard these restrictions as functions which are intrinsically defined on some Riemannian 3-manifold, whence we are led to the following
%\marc{$B$ replaced by Cotton-York}
\begin{definition}
 Let $(\Sigma,h)$ be a Riemannian 3-manifold which admits a
CKVF~ $Y$. Then, guided by \eq{expression_c} and \eq{simplified_k}, we set
\begin{eqnarray}
 \widehat c(Y) %\,:=\, \frac{\Lambda}{3}c|_{\scri}  
&:=& - \frac{1}{4} |N|^2  + \frac{1}{9}f^2
- \frac{2}{3} Y^i\widehat  \nabla_{i}f
-2  Y^iY^j \widehat L_{ij}
\;,
\label{defn_c(Y)}
\\
\widehat k(Y) 
%\,:=\, \Big(\frac{\Lambda}{3}\Big)^3  k|_{\scri}
 &:=& \frac{1}{18}|Y|^{-2}(Y_kN^k )^2 \Big(f^2
- 3Y^i\widehat  \nabla_{i}f
- 9 Y^iY^j \widehat L_{ij}
\Big)
+ \frac{1}{2}  (\widehat C_{ij} Y^i Y^j)
%\frac{1}{3}B  |Y|^{-1}
Y_kN^k
\nonumber
\\
&&  -\frac{1}{8}\widehat \nabla_i\log |Y|^2 \widehat \nabla^i(Y_kN^k)^2
 +\frac{1}{4} \widehat\nabla_i(Y_jN^j) \widehat\nabla^i(Y_kN^k)
\;.
\label{defn_k(Y)}
\end{eqnarray}
%
%for some $B\in \mathbb{R}$.
\end{definition}
%In this section we regard $\widehat c(Y)$ and $\widehat k(Y)$ as functions which are intrinsically defined on some Riemannian 3-manifold $(\Sigma,h)$ which admits a CKVF $Y$.

The spacetime functions $c$ and $k$, \eq{equation_c} and \eq{equation_k}, have been introduced in \cite{mars_senovilla} in the setting of a vanishing
MST, where they arise naturally as integration constants.
Concerning $\widehat c(Y)$ and $\widehat k(Y)$, in general
one cannot expect them  to be constant.
However, let us assume that the Cotton-York tensor satisfies condition (\ref{condition_on_C}).
%\jose{changed to avoid repetitions}
%
%\begin{equation}
%\widehat C_{ij} \,=\, B|Y|^{-5}(Y_iY_j)_{\mathrm{tf}}
%\;.
%\label{condition_C-Y}
%\end{equation}
%
Our main result, Theorem~\ref{first_main_thm2},
(which is a reformulation of Theorem~\ref{first_main_thm}) now implies the
following: Choosing an initial $D_{ij}$ according to (\ref{condition_on_D}), $(\Sigma,h)$  can be extended to a $\Lambda>0$ vacuum spacetime for
which $\Sigma$ represents $\scri^-$ and to which $Y$ extends as a KVF such that the associated MST vanishes for some function $Q$.
One then deduces from the results in \cite{mars_senovilla}
that $c$ and $k$, and therefore also $\widehat c(Y)$ and $\widehat k(Y)$ are constant: 
%\marc{$B$ constant added in the lemma}

\begin{lemma}
Let $(\Sigma, h)$ be a Riemannian 3-manifold which admits a CKVF ~$Y$ with $|Y|^2>0$ and such that
$\widehat C_{ij} = C|Y|^{-5}(Y_iY_j)_{\mathrm{tf}}$
with $C$ constant. Then the functions $\widehat c(Y)$ and $\widehat k(Y)$ as given by
\eq{defn_c(Y)} and \eq{defn_k(Y)} are constant.
\end{lemma}

In particular in the case of $\widehat k(Y)$ it  is far from obvious that the condition  \eq{condition_on_C} implies that this function is constant, 
but the proof via the extension of  $(\Sigma,h)$  to a  vacuum spacetime provides an elegant tool to prove that.
As already indicated above,
the constants $\widehat c$ and $\widehat k$ play a decisive role in the classification of $\Lambda>0$-vacuum
spacetimes which admit a conformally flat $\scri$
and a KVF w.r.t.\ which the associated MST vanishes \cite{mpss}.

\subsection{Constancy of $\widehat c(Y)$}

\def\defi{:=}
\def\F{{\mathcal F}}
\def\nab{\nabla}

Let us focus attention on the function $\widehat c(Y)$. In Section~\ref{sec_alt_Q}
 we will introduce an alternative definition of the function $Q$ which permits the derivation of a evolution equations. 
It turns out that the associated MST will in general not be regular at $\scri$, and that the constancy of $\widehat c(Y)$ is a necessary condition to ensure regularity.
Let us therefore consider the issue under which condition the function $\widehat c(Y)$ is constant.
The aim of this section is to prove the following Lemma.
%\tim{notation adapted...should be crosschecked}
%\marc{One $Z$ changed to $N$}

\begin{lemma}
%\jose{Some changes in this lemma}
\label{lem_constancy_c}
Let $(\Sigma,h)$ be a 3-dimensional oriented
Riemannian manifold %with volume form $\widehat\volform_{ijk}$.
which admits a CKVF $Y$, $\mcL_Y h = \frac{2}{3} f  h$
% (then $f=  \widehat\nabla_iY^i$). Let $N_i \defi \widehat\volform_{ijk} \widehat\nabla^j Y^k$, 
with $f$ and $N_i$ as defined in (\ref{curl_of_Y})-(\ref{div_of_Y}). Then the function $\widehat c(Y) $
introduced in (\ref{defn_c(Y)})
%$\widehat c(Y) \defi  -\frac{1}{4}|N|^2 + \frac{1}{9} f^2 - \frac{2}{3} Y^i \nab_i f -2 Y^i Y^j \widehat L_{ij}$. Then, the following identity holds
satisfies the following identity
\begin{align*}
\widehat\nab_l \widehat c(Y)  = -2 \widehat\volform^m_{\phantom{m}li} Y^i \widehat C_{mj} Y^j = \widehat{C}_{jli}Y^j Y^i .
\end{align*}
%where $\widehat C_{ij}$ is  the Cotton-York tensor 
%\tim{repetitive?}
%\marc{Yes. We should drop this once we make sure that all the conventions
%agree
%\\ -- \\
%tim: they do}
%\begin{align*}
%\widehat C_{ij} \defi - \widehat\volform_{i}^{\phantom{i}kl} \widehat\nab_l \widehat L_{jk}, \quad
% \widehat L_{jk} \defi \widehat  R_{jk} - \frac{1}{4} h_{jk} \widehat R.
%\end{align*}
In particular, if, for some smooth function $H : \Sigma \mapsto \mathbb{R}$, the Cotton-York tensor (\ref{cotton-york})  satisfies
\begin{align}
\widehat C_{ij} =  H ( Y_i Y_j)_{\mathrm{tr}}
\end{align}
and $\Sigma$ is connected, 
then the proportionality function necessarily takes the form $H=C|Y|^{-5}$ where $C$ is a constant, and $\widehat c(Y)$ is constant over the manifold.
\end{lemma}

\begin{remark}
{\rm The lemma implies in particular that $\hat{c}$ is constant if and only if $Y^j$ is an eigenvector of the Cotton-York
tensor.}
\end{remark}

\begin{proof}
  From the conformal Killing equation $\nab_{(i} Y_{j)} = \frac{1}{3} f h_{ij}$
it follows 
%\begin{equation}
%Y^{i} \widehat 
%\nabla_{i}Y_{j}+\frac{1}{2} \widehat \nabla_{j} |Y|^2 = \frac{2}{3} f Y_{j}, 
%\hspace{1cm} Y^{i} \widehat \nabla_{i} |Y|^2 = \frac{2}{3} f  |Y|^2 \label{123}
%\end{equation}
%and the equations
%$\nab_i Y^i = 3 \Psi$ and the equations
\begin{align}
\widehat\nab_j \widehat \nab_k f & = - 3(\mcL_{Y} \widehat L)_{jk} \label{HessPsi}
\;,
\\
\widehat\nab_i \widehat\nab_j Y_{l} &= Y_m\widehat  R^{m}_{\phantom{m}ijl} 
+\frac{1}{3}\Big( h_{jl} \widehat\nab_i f + h_{il} \widehat\nab_j f - h_{ij} \widehat\nab_l f \Big)
\;.
\label{HessY}
\end{align}
Evaluating $|N|^2$ as 
\begin{eqnarray*}
|N|^2 &=& \widehat\volform_{ijk} \widehat\volform^{ilm} \widehat\nab^{j} Y^k\widehat \nab_{l} Y_{m} \,=\,
\left ( \delta^{l}_{j} \delta^m_{k} - \delta^l_{k} \delta^m_{j} \right )
\widehat\nab^{j} Y^k \widehat\nab_{l} Y_{m} 
\\
&=& 
\widehat\nabla^j Y^k\widehat \nab_j Y_k -\widehat \nab^j Y^k \widehat\nab_k Y_j
\;,
\end{eqnarray*}
we can write $\widehat c(Y) =-\frac{1}{4}\widehat \nab_i Y_j  ( \widehat\nab^i Y^j - \widehat\nab^j Y^i  )
+ \frac{1}{9} f^2 - \frac{2}{3} Y^i \widehat\nab_i f  -2  Y^i Y^j \widehat L_{ij}$.
It is convenient to split $\widehat c(Y)$ in two terms
\begin{align*}
\widehat c_1 &\defi  -\frac{2}{3} Y^i \widehat\nab_i f -2  Y^i Y^j \widehat L_{ij} + \frac{1}{9} f^2
\;,
 \\
 \widehat c_2 &\defi   -\frac{1}{4}\widehat \nab_i Y_j  ( \widehat\nab^i Y^j - \widehat\nab^j Y^i  )
\;.
\end{align*}
so that $\widehat c = \widehat c_1 +\widehat  c_2$. We start with $\widehat \nab_l c_1$,
\begin{align}
\widehat\nab_{l} \widehat c_1  = & -\frac{2}{3}\widehat\nab_l Y^i \widehat\nab_i f - \frac{2}{3}Y^i \widehat\nab_l \widehat\nab_i f -4
(\widehat\nab_l Y^i) Y^j \widehat L_{ij}
-2  Y^i Y^j \widehat\nab_l \widehat L_{ij}  + \frac{2}{9}f \widehat\nab_l f
 \nonumber  
\\
 = & - \frac{2}{3}\widehat\nab_l Y^i \widehat\nab_i f +2 Y^i (\mcL_{Y} \widehat L)_{li}
-4 (\widehat\nab_{l} Y^i) Y^j \widehat L_{ij}
-2 \volform^m_{\phantom{m}li} Y^i\widehat C_{mj} Y^j 
  \nonumber  \\
&-2 Y^j [ (\mcL_Y L)_{lj} -\widehat L_{lm} \widehat\nab_j Y^m
-\widehat L_{mj} \widehat\nab_l Y^m ]
+ \frac{2}{9}f \widehat\nab_l f
\nonumber
 \\ 
= &  
-\frac{2}{3}\widehat\nab_l Y^i \widehat\nab_i f-4 Y^j  \widehat L_{i[j} \widehat\nab_{l]} Y^i 
-2 \widehat\volform^m_{\phantom{m}li} Y^i \widehat C_{mj} Y^j 
+ \frac{2}{9}f \widehat\nab_l f 
\;,
\label{derc1}
\end{align}
where in the second equality we inserted (\ref{HessPsi}) and
$Y^i \widehat\nab_l\widehat L_{ij} = Y^i \widehat\nab_i\widehat L_{lj} + Y^i 
\widehat\volform^{m}_{\phantom{m}li} \widehat C_{mj}
= (\mcL_{Y} \widehat L)_{lj} - \widehat L_{lm} \widehat\nab_j Y^m -\widehat L_{mj} \widehat\nab_{l} Y^m 
+ Y^i \widehat\volform^{m}_{\phantom{m}li} \widehat C_{mj}$ and in the third one obvious
cancellations have been applied. Concerning $\widehat\nab_l \widehat c_2$ we
find, after a simple rearrangement of indices,
\begin{align*}
\widehat\nab_l \widehat c_2 & = -\frac{1}{2} \widehat\nab_{l} \widehat\nab_{i} Y_j (\widehat \nab^i Y^j - \widehat\nab^j Y^i) 
=  - Y_m \widehat R^m_{\phantom{m}lij} \widehat \nab^i Y^j  
- \frac{1}{3}\widehat\nab_i f   ( \widehat \nab^i Y_l - \widehat\nab_l Y^i ) 
\;,
\end{align*}
where in the second equality we used (\ref{HessY}) and the
antisymmetry of $(\widehat\nab^i Y^j - \widehat\nab^j Y^i)$. We now use the Riemann tensor
decomposition in three dimensions,
\begin{align*}
\widehat R^{m}_{\phantom{m}lij} = \delta^m_i \widehat L_{lj} - \delta^m_j \widehat L_{li} 
+\widehat  L^m_{\phantom{m}i} h_{lj} 
-  \widehat L^m_{\phantom{m}j} h_{li} 
\;,
\end{align*}
to obtain
\begin{align}
\widehat\nab_l\widehat  c_2 = -( Y^m \widehat L_{mi} +\frac{1}{3}\widehat\nabla_i f) ( \widehat \nab^i Y_l - \widehat\nab_l Y^i  )
- Y_i \widehat L_{lj}  ( \widehat \nab^i Y^j - \widehat \nab^j Y^i  ) 
\;.
\label{derc2}
\end{align}
Combining (\ref{derc1}) and (\ref{derc2}) 
\begin{align*}
\widehat\nab_l \widehat c(Y)  = & - Y^m \widehat L_{mi}  ( \widehat\nab^i Y_l  + \widehat\nab_l Y^i  )
+ Y_i\widehat  L_{lj}  ( \widehat\nab^i Y^j + \widehat\nab^j Y^i  ) 
- \frac{1}{3} \widehat\nab_i f  ( \widehat\nab^i Y_l + \widehat\nab_l Y^i  ) \\
& + \frac{2}{9} f \widehat\nab_l f 
-2  \widehat\volform^m_{\phantom{m}li} Y^i \widehat C_{mj} Y^j  \\
= &  -2 \widehat\volform^m_{\phantom{m}li} Y^i \widehat C_{mj} Y^j 
\;,
\end{align*}
where in the second equality we used the conformal Killing equation. 

Now, whenever \eq{condition_on_C} holds we have $\widehat C_{mj} Y^j = \frac{2}{3} H |Y|^2 Y_m$
and $\widehat\nab_l \widehat c(Y) =0$ so that $\widehat c(Y)$ is constant over the (connected) manifold
$\Sigma$. The fact that $H$ is necessarily of the form
$H = C |Y|^{-5}$ was already shown in the proof 
of Proposition \ref{rescaledMST_scri}.
%\marc{The necessary form of $H$ was already found of the proof of Proposition 2.5. I have  added the basic equations used in the proof there.}
%Finally, using that $C_{ij}$ is divergence-free, from (\ref{C}) we deduce
%$$
%0=\widehat \nabla^{i}C_{ij}=Y_{j} Y^{i}\widehat \nabla_{i}H -\frac{1}{3} |Y|^2 \che%ck \nabla_{j}H + H\left(Y_{j} \widehat \nabla_{i} Y^{i}  + Y^{i}\widehat \nabla_{i}%Y_{j} -\frac{1}{3} \widehat \nabla_{j} |Y|^2 \right)
%$$
%and using here the second in (\ref{123})
%\begin{equation}
%\frac{1}{3} |Y|^2 \widehat \nabla_{j}H = Y_{j} Y^{i}\widehat \nabla_{i}H +5H f Y_{j%} -\frac{5}{6} H \widehat \nabla_{j} |Y|^2 . \label{step}
%\end{equation}
%Contracting here with $Y^{j}$ and using the third in (\ref{123})
%$$
%\frac{2}{3} |Y|^2 Y^{j}\widehat \nabla_{j}H +\frac{5}{3} H Y^{j}\widehat \nabla_{j}% |Y|^2 =0
%$$
%which renders (\ref{step}), upon use of the 3rd in (\ref{123}) again, as
%$$
%2 |Y|^2 \widehat \nabla_{j}H +5 H \widehat \nabla_{j} |Y|^2 =0
%$$
%ergo $H|Y|^{5}=B$ for a constant $B$.
\qed
\end{proof}

%\vspace{3mm}

\begin{remark}
%\jose{Notation for the sign changed}
{\rm
A similar Lemma holds for three-dimensional manifolds 
of arbitrary signature. The term $|N|^2$ in $\widehat c(Y)$ needs to be replaced 
by $\epsilon |N|^2$ where $\epsilon$ is an appropriate sign depending
on the signature.
}
\end{remark}

%\vspace{3mm}

Another   problem of interest is to find necessary and sufficient conditions which ensure the constancy of
$\widehat k(Y)$. Since this expression is of higher order in $Y$ than $\widehat c(Y)$, this
is expected to be  somewhat more involved.

\section{Evolution of the MST}

\subsection{The Ernst potential on $\scri$}

In this section we make no assumption concerning the
MST, so all the results hold generally for any
$\Lambda$-vacuum
spacetime admitting a KVF $X$ and a smooth conformal compactification. 

%\marc{$\sigma_{\mu}$ defined}
Using the results of Section~\ref{Mars-Simon_conf} the so-called
Ernst one-form of $X$,
$\sigma_{\mu} := 2X^{\alpha} \mathcal{F}_{\alpha\mu}$, has the following 
asymptotic expansion
\begin{eqnarray}
\sigma_{\mu} &=& 2X^{\alpha} \mathcal{F}_{\alpha\mu}
\\
&=&  2X^{\alpha}\Big(  \Theta^{-3}\widetilde{\mathcal{H}}_{\alpha\mu} +  \Theta^{-2}\widetilde{ \mathcal{F}}_{\alpha\mu}\Big)
\\
&=&  2 \Theta^{-3}\widetilde X^{\alpha}\Big( \widetilde H_{\alpha\mu}  +\frac{i}{2} 
\widetilde\volform_{\alpha\mu}{}^{\nu\sigma}\widetilde H_{\nu\sigma}  +  \Theta
(\widetilde F_{\alpha\mu}  + \frac{i}{2} \widetilde\volform_{\alpha\mu}{}^{\nu\sigma}  \widetilde F_{\nu\sigma} )\Big)
\\
&=&  2 \Theta^{-3}\widetilde X^{\alpha}\Big(  2\widetilde  X_{[\alpha}\widetilde\nabla_{\mu]}\Theta 
%+ i \widetilde\volform_{\alpha\mu}{}^{\nu\sigma} \widetilde  X_{\nu}\widetilde\nabla_{\sigma}\Theta 
+  \Theta
((\widetilde \nabla_{\alpha}\widetilde X_{\mu})_{\mathrm{tf}}
%-\frac{1}{4}g_{\alpha\mu}\widetilde\nabla_{\beta}\widetilde X^{\beta} 
+ \frac{i}{2} \widetilde\volform_{\alpha\mu}{}^{\nu\sigma} \widetilde\nabla_{\nu}\widetilde X_{\sigma} )\Big)
\;.
\end{eqnarray}
It is known (see e.g.\ \cite{KSMH}) that this covector field has an (``Ernst-'') potential 
$\sigma_{\mu}=\partial_{\mu}\sigma$, at least locally.
% supposing that $X$ is a KVF.\jose{Little re-ordering}
Taking the following useful relations into account,
%\tim{mistake in first and  second eqn}
%
\begin{eqnarray*}
Y^i \widehat \volform_{i }{}^{jk}\widetilde\nabla_t\widetilde\nabla_t\widetilde \nabla_{j}\widetilde X_{k} \big|_{\scri}
&=&2\widehat C_{ij}Y^iY^j
\;,
\\
 \widetilde\nabla_t\Big(-|\widetilde X|^2 \Theta^{-2}  +\frac{3}{\Lambda} i \widetilde H^{\alpha\beta} \widetilde F^{\star}_{\alpha\beta} \Theta^{-1}\Big) \Big|_{\scri}
&=&   2\sqrt{\frac{\Lambda}{3}} |Y|^2\Theta^{-3}   - i  Y_i N^i \Theta^{-2}
\\
&&\hspace{-5em}
-  \frac{1}{2}\sqrt{\frac{3}{\Lambda}} \widehat R |Y|^2 \Theta^{-1}
+\frac{2}{3}\sqrt{\frac{3}{\Lambda}} D_{kl}  Y^{k}  Y^l
\\
&&\hspace{-5em}
  +\frac{i}{2} \frac{3}{\Lambda}  \Big( 2\widehat C_{kl}Y^kY^l+ N^k\ol{\widetilde\nabla_t\widetilde\nabla_t\widetilde X_{k}}
  \Big)
   + O(\Theta)
\;,
\end{eqnarray*}
a somewhat lengthy computation  making extensive use of the equations \eq{constr2}-\eq{constr_last}
and the Killing  relations \eq{rel_KVF}-\eq{Killing_rel_last}
reveals that (as before an overbar means ``restriction to $\scri$'')

%\tim{many  cancellations}
%
\begin{eqnarray}
\sigma_{t}
&=&  2 \Theta^{-2}\widetilde X^{t}\widetilde \nabla_{t}\widetilde X_{t}+\frac{1}{2}  \Theta^{-2}\widetilde X^{t}\widetilde\nabla_{\beta}\widetilde X^{\beta}  
 +   2 \Theta^{-3}\widetilde X^{i}\widetilde  X_{i}\widetilde\nabla_{t}\Theta 
\nonumber
\\
&&
+   2 \Theta^{-2}\widetilde X^{i} \widetilde \nabla_{i}\widetilde X_{t}
+i \Theta^{-2}\widetilde X^{i}  \widetilde\volform_{i t}{}^{jk} \widetilde\nabla_{j}\widetilde X_{k} 
 + O(\Theta)
%\\
%&=&   2 \Theta^{-3}\widetilde X^{i}\widetilde  X_{i}\widetilde\nabla_{t}\Theta 
%\\
%&&
%+   2 \Theta^{-2}\widetilde X^{i} \widetilde \nabla_{i}\widetilde X_{t}
%\\
%&&
%+i \Theta^{-2}\widetilde X^{i}  \widetilde\volform_{i t}{}^{jk} \widetilde\nabla_{j}\widetilde X_{k} 
% + O(\Theta)
\\
&=&     2\sqrt{\frac{\Lambda}{3}} |Y|^2\Theta^{-3}   - i  Y_i N^i \Theta^{-2}
-  \frac{1}{2}\sqrt{\frac{3}{\Lambda}} \widehat R |Y|^2 \Theta^{-1}
 + \frac{2}{3}  \sqrt{\frac{3}{\Lambda}}   D_{kl}Y^{k}  Y^l
\nonumber
\\
&&
-\frac{i}{2}\frac{3}{\Lambda}\Big( 2\widehat C_{ij}Y^iY^j+ N^i\ol{\widetilde  \nabla_t\widetilde \nabla_t\widetilde X_{i} } \Big)
%
%-\frac{i}{2}A^2\widetilde X^{i}  \volform_{i }{}^{jk}( \widetilde\nabla_{j}\nabla_t\nabla_t\widetilde X_{k}
% + Y^{l} \nabla_t R_{tjkl}- \frac{2}{3}f  R_{tjtk})
%\\
%&&
%+\frac{i}{4}A^2Y^{i}N^j R_{titj}
%- \frac{i}{4}A^2 Y_kN^kR_{tt}
%-\frac{i}{6}f A^2\widetilde X^{i}  \volform_{i }{}^{jk}R_{tjtk}
 + O(\Theta)
\label{expansion_sigma_t}
\\
&=&  \widetilde\nabla_t\Big[-\widetilde X^2 \Theta^{-2}  +\frac{3}{\Lambda} i \widetilde H^{\alpha\beta} \widetilde F^{\star}_{\alpha\beta} \Theta^{-1}+ q(x^i)
\nonumber
\\
&&-i\frac{3}{\Lambda}\sqrt{\frac{3}{\Lambda}}\Big( 2  \widehat C_{kl}Y^kY^l  +N^k\ol{\widetilde \nabla_t\widetilde \nabla_t\widetilde X_k } \Big)\Theta +  O(\Theta^2) \Big]   
\;
\end{eqnarray}
for some function $q(x^i)$. To determine this function it is necessary to compute $\sigma_i$ up to and including constant order.
Using the relation
\begin{eqnarray}
 \widehat\nabla_i(Y_kN^k) &=& \frac{1}{3}f N_i  +2\widehat \volform_{ij}{}^{k} Y^jY^l  \widehat  L_{kl} 
 + \frac{2}{3}\widehat\volform_{ij}{}^{k} Y^j\widehat \nabla_kf 
\;,
\label{deriv_YN}
\end{eqnarray}
another lengthy calculation gives via
\eq{constr2}-\eq{constr_last}
and  \eq{rel_KVF}-\eq{Killing_rel_last}
%\marc{changed the reference in the bracket term. Using (\ref{HessPsi}) it follows more easily}
%
\begin{eqnarray}
\sigma_{i}
&=&   -  \widehat  \nabla_{i}|Y|^2 \Theta^{-2} +  i\sqrt{\frac{3}{\Lambda}} \widehat\nabla_i(Y_kN^k) \Theta^{-1}
 +  \frac{1}{\Lambda} \Big[ -   Y_i \Big(\Delta_h -\frac{\widehat R}{3}\Big) f
\nonumber
\\
&&  
 -  Y_{i}\ol{ \widetilde \nabla_t\widetilde\nabla_t\widetilde\nabla_t\widetilde X^t}
+3 i\widehat \volform_{ij }{}^{k} Y^{j} \ol{\widetilde\nabla_t\widetilde\nabla_t\widetilde\nabla_{t}\widetilde X_{k} }
+ 3\ol{\widetilde\nabla_t\widetilde\nabla_t\widetilde X^j} \widehat  \nabla_{j}Y_{i}
\nonumber
\\
&&  -\frac{1}{2} |Y|^2\widehat \nabla_{ i}\widehat R
+ \frac{1}{2}  Y_{i}Y^{j} \widehat \nabla_{ j}\widehat R
+ 3 Y^{j}\ol{\widetilde\nabla_t\widetilde\nabla_t\widetilde \nabla_{j}\widetilde X_{i}}
\Big]
+O(\Theta)
\\
&=&   -  \widehat  \nabla_{i}|Y|^2 \Theta^{-2} +  i\sqrt{\frac{3}{\Lambda}} \widehat\nabla_i(Y_kN^k) \Theta^{-1}
\nonumber
\\
&&   +  \frac{1}{\Lambda} \Big[    \frac{2}{3}Y_i\underbrace{ \Big(\Delta_h  f+ \frac{1}{2} \widehat R f+ \frac{3}{4}Y^{j} \widehat \nabla_{ j}\widehat R
\Big)}_{=0 \text{ by \eq{HessPsi}}}
-6 i\sqrt{\frac{\Lambda}{3}}\widehat \volform_{ij }{}^{k}  D_{kl} Y^{j}Y^l
\nonumber
\\
&&
%- 3\widetilde \nabla_{i}(Y^{j}\nabla_t\nabla_tX_j )
%+ \frac{1}{3}\widehat\nabla_i f^2
%- \frac{3}{8} \nabla_{i}(\widehat R|Y|^2)
%\nonumber
%\\
%&&
%-\frac{1}{2} |Y|^2\widetilde\nabla_{ i}\widehat R
%+ \frac{3}{2}\nabla_{k}(  L_{i}{}^{k}|Y|^2)
- 3\ol{\widetilde \nabla_t\widetilde\nabla_t\widetilde X^j }\widehat  \nabla_{i}Y_{j}
+ 2  f\ol{\widetilde\nabla_t\widetilde\nabla_t\widetilde X_i}
  -\frac{1}{2} |Y|^2\widehat \nabla_{ i}\widehat R
\nonumber
\\
&&
-  3 Y^{j}\ol{\widetilde\nabla_t\widetilde\nabla_t\widetilde \nabla_{i}\widetilde X_{j}}
\Big]
+O(\Theta)
\label{expansion_sigma_i}
\\
&=&    \widetilde \nabla_{i}\Big(-\widetilde X^2 \Theta^{-2} 
  +\frac{3}{\Lambda} i \widetilde H^{\alpha\beta} \widetilde F^{\star}_{\alpha\beta} \Theta^{-1}
-a 
+O(\Theta)\Big) 
\nonumber
\\
&&-2 i\sqrt{\frac{3}{\Lambda}}\widehat \volform_{ij }{}^{k}  D_{kl} Y^{j}Y^l
\label{sigma_i_prelim}
\end{eqnarray}
%
%since
%%
%\begin{eqnarray*}
% \widetilde\nabla_i(-|\widetilde X|^2\Theta^{-2}) 
%&=&
% -  \widehat \nabla_i |Y|^2 \Theta^{-2}
% - 3\Lambda^{-1}\widetilde \nabla_t\widetilde\nabla_tX^j \widetilde \nabla_{i}Y_{j}
%+ 2\Lambda^{-1}  f\widetilde\nabla_t\widetilde\nabla_t\widetilde X_i
%\\
%&&
%  -\frac{1}{2}\Lambda^{-1} |Y|^2\widehat \nabla_{ i}\widehat R
%-  3 \Lambda^{-1}Y^{j}\widetilde\nabla_t\widetilde\nabla_t\widetilde \nabla_{i}\widetilde X_{j}
% +  O(\Theta)
%%\\
%%&=&- \widehat \nabla_i|Y|^2 \Theta^{-2}
%% +\frac{1}{3\Lambda }\widehat\nabla_i f^2
%%- \frac{3}{\Lambda}\widehat\nabla_i(Y^j\widetilde \nabla_t \widetilde\nabla_t\widetilde  X_j)
%%\\
%%&&
%%+ \frac{3}{2\Lambda}\widehat \nabla_{k}(\widehat  L_{i}{}^{k}|Y|^2)
%%- \frac{3}{8\Lambda}\widehat \nabla_i( \widehat R|Y|^2) 
%% -\frac{1}{2\Lambda } |Y|^2 \widehat \nabla_i\widehat R
%% + O(\Theta)
%\;.
%\end{eqnarray*}
%
%\tim{last term corrected}
for some complex  constant $a$ (the ``$\sigma$-constant'' introduced in the Introduction).
%\tim{added \\
%--
%\\ Jose: little addition}
%\tim{$t$-dependent...}
As one should expect, the last term in \eq{sigma_i_prelim} has, at least locally, a potential, supposing that $Y$ is a CKVF and
$D_{ij}$ a TT tensor which together satisfy the KID equations \eq{reduced_KID}: Indeed, setting
%\tim{corrected}
%
\begin{equation}
P_i :=-2 \sqrt{\frac{3}{\Lambda}} \widehat \volform_{ij }{}^{k}  D_{kl} Y^{j}Y^l
\;,
\end{equation}
we find that
%\tim{corrected}
%
\begin{eqnarray}
\widehat\volform_i{}^{jk}\widehat\nabla_{j} P_{k} &=& 2 \sqrt{\frac{3}{\Lambda}} Y^j\Big(
\mcL_YD_{ij} + \frac{1}{3} fD_{ij}
\Big)
\,=\, 0
\\
\Longrightarrow \quad \widehat\nabla_{[i} P_{j]}  &=& 0
\;.
\end{eqnarray}
On the simply connected components of  the initial 3-manifold 
%\tim{do we want to assume this?... Theorem \ref{constancy}}
this implies
\begin{equation}
P_i =\widehat\nabla_i p\quad \text{for some real-valued function} \quad p=p(x^i)
\;.
\end{equation}
(The fact that $p$ is only determined up to some constant is reflected in the $\sigma$-constant $a$
 introduced above.)
Thus,
\begin{equation}
\sigma_{i}
=    \widetilde \nabla_{i}\Big(-\widetilde X^2 \Theta^{-2} 
  +\frac{3}{\Lambda} i \widetilde H^{\alpha\beta} \widetilde F^{\star}_{\alpha\beta} \Theta^{-1} + ip(x^j)
-a  +O(\Theta)\Big) 
\;.
\label{sigma_i}
\end{equation}
Altogether, we conclude that $q(x^i) = i p(x^i) -a$ for some not yet
specified $a\in\mathbb{C}$,
% (recall that $\partial_{\mu} \sigma = \sigma_{\mu}$),
%
and that 
\begin{eqnarray}
 \sigma 
&=& -\widetilde X^2 \Theta^{-2} 
  +\frac{3}{\Lambda} i \widetilde H^{\alpha\beta} \widetilde F^{\star}_{\alpha\beta} \Theta^{-1}
+ ip(x^j) -a
\nonumber
\\
&&
-i\frac{3}{\Lambda}\sqrt{\frac{3}{\Lambda}}\Big( 2  \widehat C_{kl}Y^kY^l  +N^k\widetilde \nabla_t\widetilde \nabla_t\widetilde X_k  \Big)\Theta
+O(\Theta^2)
\\
&=&  -|Y|^2 \Theta^{-2}  + i \sqrt{\frac{3}{\Lambda}} Y_iN^i \Theta^{-1}  
 +\frac{1}{3\Lambda} f^2  - \frac{3}{\Lambda}Y^{i}\ol{\widetilde\nabla_t\widetilde \nabla_t  \widetilde X_{i}}
\nonumber
\\
&&
+ ip(x^j) -a
  +\frac{3}{\Lambda}\Big(
\frac{2}{3} D_{kl}Y^kY^l  - i \sqrt{\frac{3}{\Lambda}}\widehat C_{kl}Y^kY^l 
\nonumber
\\
&&- \frac{i}{4}  \sqrt{\frac{3}{\Lambda}}N^k(2\ol{\widetilde \nabla_t\widetilde \nabla_t\widetilde X_k } + Y_k\widehat R)   \Big)\Theta
+O(\Theta^2)
\;.
\label{expansion_sigma}
\end{eqnarray}

\begin{proposition}
Consider a $\Lambda>0$-vacuum spacetime which admits a KVF and a smooth $\scri$.
Then the Ernst potential $\sigma$ can be computed explicitly near $\scri$, where it admits the expansion \eq{expansion_sigma}.
\end{proposition}

\subsection{Alternative definition of the function $Q$}
\label{sec_alt_Q}

In order to derive evolution equations it is convenient 
(cf.\ \cite{ik} for the $\Lambda=0$-case)
to define, {\em for each Ernst potential $\sigma$}, a new function $Q=Q_{\mathrm{ev}}$ 
%\marc{Dropped ``up to some complex constant''. I don't thing there was a freedom there
%\\ -- \\ tim: no new freedom, but there is the freedom coming from the Ernst-potential... I just wanted to emphasize that this is actually a 1-parmeter family of $Q_{\mathrm{ev}}$'s\\
%--\\
%Jose: A little rewriting added that gives a compromise} 
by the following set of equations,
\begin{align}
%\nabla_{\mu}\sigma \,=\, \sigma_{\mu} \,=\,2X^{\alpha} \mathcal{F}_{\alpha\mu}
%\;,
%\\
Q_{\mathrm{ev}} \, & :=\, \frac{3J}{R} - \frac{\Lambda}{R^2}
\;,
 \label{ev_dfn_Q}
\\
 R \, & :=\, - \frac{1}{2}i\sqrt{\mathcal{F}^2}
\label{defR}
\;,
\\
J \, & :=\, \frac{R +  \sqrt{R^2 - \Lambda\sigma}}{\sigma} 
\;.
 \label{defJ}
\end{align}

%\marc{text added}
and all square roots are chosen with the same prescription as explained
above, cf.\ p.~\pageref{rem_compl_sr}.
Alternatively, we could have defined $R = + (i/2)\sqrt{\mathcal{F}^2}$.
Then the expression for $J$ would have changed accordingly. The choice
(\ref{defR}) is preferable because then
the real part of $R$ approaches minus infinity at $\scri$, in agreement with 
the usual behavior of Boyer-Lindquist type coordinates near infinity
 in Kerr-de Sitter and related metrics \cite{mars_senovilla}.
Note that the definition of $J$ above implies the identity
\begin{eqnarray}
\sigma J^2 -2J R +\Lambda \,=\,0
\;,
\label{quadratic}
\end{eqnarray}
which will be useful later.
The MST associated with the choice $Q=Q_{\mathrm{ev}}$ will be denoted by $\mathcal{S}^{\mathrm{(ev)}}_{\mu\nu\sigma\rho}$.

It follows from \eq{used_relation_F2v2} that
%\tim{change $c$ to $\widehat c$?}
%
\begin{eqnarray}
 R^2& =& -\frac{1}{4}\mathcal{F}^2 
\\
&=&
 \frac{\Lambda}{3} \Theta^{-2} \widetilde X^2   - i \Theta^{-1}\widetilde H^{\alpha\beta} \widetilde F^{\star}_{\alpha\beta}+ \frac{\Lambda}{3} c(x^j)
\nonumber
\\
&&  +\Big(2 D_{ij}Y^iY^j   + i \sqrt{\frac{3}{\Lambda}}N^k\ol{\widetilde \nabla_t\widetilde \nabla_t\widetilde X_k} \Big)\Theta  + O(\Theta^2)
\\
&=&
\frac{\Lambda}{3} \Theta^{-2}|Y|^2 
- i \sqrt{\frac{\Lambda}{3}}Y_{i}N^i\Theta^{-1}
-\frac{1}{9} f^2
+ Y^{i} \ol{\widetilde\nabla_t\widetilde \nabla_t \widetilde X_{i}  }
+ \frac{\Lambda}{3} c(x^j)
\nonumber
\\
&&
 +\Big(\frac{4}{3} D_{kl}Y^kY^l   -  i   \sqrt{\frac{3}{\Lambda}}\widehat C_{kl}Y^kY^l
 + \frac{i}{4} \sqrt{\frac{3}{\Lambda}}N^k(2\ol{\widetilde \nabla_t\widetilde \nabla_t\widetilde X_k} + Y_k\widehat R ) \Big)\Theta  
\nonumber
\\
&&
+ O(\Theta^2)
\;.
\label{expansion_R}
\end{eqnarray}
One remark is in order: In this section we do \emph{not} assume that the MST vanishes, so there is no reason
why the real  function $c$, which has been defined on $\scri$ in \eq{expression_c}, should be constant.
%\tim{add remark there}
%Let us  write
%\tim{change notation}
%%
%\begin{eqnarray*}
% A&:=& |Y|^2 \, > \, 0
%\;,
%\\
%B&:=&  i \sqrt{\frac{3}{\Lambda}} Y_iN^i
%\;,
%\\
%C&:=& \frac{1}{3\Lambda} f^2  - \frac{3}{\Lambda}Y^{i}\widetilde\nabla_t\widetilde \nabla_t  \widetilde X_{i}
%\;,
%\\
%D&:=& -\frac{3}{\Lambda} D_{kl}Y^kY^l
%\;,
%\\
%E &:=&  i\frac{3}{\Lambda}  \sqrt{\frac{3}{\Lambda}}\widehat C_{kl}Y^kY^l 
%\;,
%\\
%F&:=&   -\frac{i}{4} \frac{3}{\Lambda}  \sqrt{\frac{3}{\Lambda}}N^k\Big( 2\widetilde \nabla_t\widetilde \nabla_t\widetilde X_k + Y_k\widehat R \Big)
%\;.
%\end{eqnarray*}
%
From \eq{expansion_sigma} and \eq{expansion_R} we observe that
%%
%\begin{eqnarray*}
%\sigma &=&  - A \Theta^{-2}  +B \Theta^{-1}  
% +C
%+a -\Big( \frac{2}{3} D + E - F\Big)\Theta
%+O(\Theta^2)
%\;,
%\\
%%P\,:=\,
%\frac{3}{\Lambda}R^2 &=& 
% A\Theta^{-2} - B\Theta^{-1}- C+  c 
%-\Big(\frac{4}{3}D+  E+F\Big)\Theta + O(\Theta^2) 
%\;.
%\end{eqnarray*}
%%
%Observe that
%\tim{!!!!}
%
\begin{eqnarray}
 \Xi\,:=\, \sigma +\frac{3}{\Lambda}R^2 \,=\,
 c(x^j)+ ip(x^j) -a+\frac{6}{\Lambda}\widetilde{\mathcal{D}}_{tktl}|_{\scri}Y^kY^l  \Theta
 + O(\Theta^2)  
\;,
\end{eqnarray}
whence
\begin{eqnarray}
Q_{\mathrm{ev}} &=&  \frac{3R^2-\Lambda\sigma  + 3 R\sqrt{R^2 - \Lambda\sigma} }{\sigma R^2} 
\label{expr_Q1}
%\\
%&=&  \frac{\frac{3}{\Lambda}R^2 -\sigma - 3 \sqrt{\frac{R^2}{\Lambda}}\sqrt{\frac{R^2}{\Lambda} - \sigma} }{\sigma \frac{R^2}{\Lambda}} 
%\\
%&=&3 \frac{2P -\Xi -  \sqrt{P(4P - 3\Xi ) }}{P(\Xi- P) } 
\\
&=&3 \frac{2R^2  - \frac{\Lambda}{3}\Xi -  2R^2\sqrt{1  - \frac{\Lambda}{4}\frac{\Xi}{R^2} } }{R^2 (\Xi- \frac{3}{\Lambda}R^2 ) } 
\label{expr_Q2}
\\
&=&  \frac{\Lambda^2}{12}\frac{\Xi}{R^4} +O\Big(\frac{\Xi^2}{R^6}\Big) 
\;.
\label{expr_Q3}
\end{eqnarray}

\begin{remark}
{\rm
The expressions \eq{expr_Q1} and \eq{expr_Q2} for $Q_{\mathrm{ev}}$ rely on the choice of $R=-\sqrt{R^2}$. The final expression  (\ref{expr_Q3}) is however
independent of this choice, as it must be. This final expression ensures that,
in an appropriate setting, $Q_{\mathrm{ev}}$ coincides with $Q_0$, as will be shown in
Theorem~\ref{first_main_thm}.
It should also be emphasized that this expression does not admit
a limit $\Lambda \searrow 0$. This is because, when $\Lambda=0$, the function
$\mathcal{F}^2$ approaches zero at infinity and the definition 
of square root needs to be worked out differently.
%For $\Lambda=0$, the other sign in the definition of $R$ is preferable, cf.\ \c%ite{mars}.
%%\tim{does $R$ appear there?... Klainerman}
%Consequently, the corresponding $Q_{\mathrm{ev}}$ cannot be obtained as a 
%limit 
}
\end{remark}

We have 
\begin{eqnarray}
Q_{\mathrm{ev}} =\begin{cases}
  \frac{3}{4}\big(c(x^j)+ ip(x^j) -a\big)|Y|^{-4}\Theta^4 +O(\Theta^5)  &  \text{if $a \not\equiv c(x^j)+ ip(x^j)$}
\;,
\\
  \frac{9}{2}\Lambda^{-1}|Y|^{-4}\widetilde{\mathcal{D}}_{tktl}Y^kY^l  \Theta^5 +O(\Theta^6)  & \text{if $a\equiv c(x^j)+ ip(x^j)$}
\;.
\end{cases}
\label{final_Q_ev}
\end{eqnarray}
The rescaled MST with $Q = Q_{\mathrm{ev}}$
will be regular on $\scri$ if and only if $Q_{\mathrm{ev}} =O(\Theta^5)$, i.e.\ if and only if
both functions $c$ and $p$ are constant and the 
%complex
$\sigma$-constant $a$
 has been  chosen such that
\begin{equation}
\mathrm{Re}(a) \,=\, c \quad \text{and} \quad \mathrm{Im}(a) \,=\, p
\;.
\label{afix}
\end{equation}
We remark that with this choice of $a$ the function $Q_{\mathrm{ev}} $ is completely determined.

Note that the potential $p$ will be constant if and only if the covector field $P_i$ vanishes.
The constancy of $c$ has been analyzed in Lemma~\ref{lem_constancy_c}.
Comparison with \eq{leading_order_Q} then leads to the following result, 
a shortened version of which has been stated as Theorem~\ref{short}
in the Introduction:

%\tim{if and only if}
%xxxxxxxxxxxxxxxxxxx
%parameterization/
%classification of $\Lambda>0$-vacuum space-time with a smooth $\scri$ and vanishing Mars-Simon tensor
%xxxxxxxxxxxxxxx

\begin{theorem}
\label{prop_Qs}
%\tim{rewordings and $|Y|^2>0$-condition removed}
Consider a $\Lambda>0$-vacuum spacetime%
\footnote{Since $P_i$ needs to vanish in this setting, $P_i$ is exact and we need
not assume that $\scri^-$ is simply connected in order to get a globally defined Ernst potential.}
 which admits a smooth
$\scri^-$ and a KVF $X$.
Denote by $Y$ the CKVF induced, in the conformally rescaled spacetime, by $X$ on $\scri^-$.
If and only if
\begin{enumerate}
\item[(i)] 
$\widehat \volform_{ij }{}^{k}  \widehat C_{kl} Y^{j}Y^l=0$ (so that the function $c=\frac{3}{\Lambda}\widehat c(Y)$ is constant on $\scri^-$),
%$\widehat c(Y) = - \frac{1}{4}\Big( |\,\mathrm{curl}\,Y|^2  - \frac{4}{9}(\widehat\nabla_i Y^i)^2 + \frac{8}{3} Y^i\widehat  \nabla_{i}\widehat\nabla_jY^j + 8 Y^iY^j \widehat L_{ij}
%\Big) $ is constant, 
and 
\item[(ii)]  $ \widehat \volform_{ij }{}^{k}  D_{kl} Y^{j}Y^l=0$ (so that $P_i=0$ whence its potential $p$ is constant),
%$P_i =-6 \sqrt{\frac{\Lambda}{3}} \widehat \volform_{ij }{}^{k}  D_{kl} Y^{j}Y^l=0$, so that its potential $p$  is constant,
\end{enumerate}
there exists a
%\jose{little rewriting} 
unique %complex
$\sigma$-constant $a$, given by $a=c+ip$,  which leads via $Q_{\mathrm{ev}}$ to a rescaled
MST $\widetilde{\mathcal{T}}^{(\mathrm{ev})}_{\mu\nu\sigma\rho}$ which is regular at $\scri^-$. In that case the leading order terms of  $Q_{\mathrm{ev}}=O(\Theta^5)$ and $Q_0=O(\Theta^5)$ coincide,
\begin{equation*}
 \lim_{\Theta \rightarrow 0} (\Theta^{-5}Q_{\mathrm{ev}})=  \lim_{\Theta \rightarrow 0} (\Theta^{-5}Q_0)
\;.
\end{equation*}
In particular,
\begin{equation*}
\widetilde{\mathcal{T}}^{(\mathrm{ev})}_{\mu\nu\sigma}{}^{\rho}|_{\scri^-} \,=\, \widetilde{\mathcal{T}}^{(0)}_{\mu\nu\sigma}{}^{\rho}|_{\scri^-}
\;.
\end{equation*}
\end{theorem}

\begin{remark}
{\rm
For initial data of the form \eq{condition_on_D}-\eq{condition_on_C}, which are necessary for the MST to vanish for some choice of $Q$, 
the conditions (i) and (ii) are satisfied.
% Indeed, \eq{condition_on_C} implies $c=\mathrm{const}$, while it follows from \eq{condition_on_D} that $P_i=0$ and thus $p=\mathrm{const}$.
}
\end{remark}
% 
%Given a conformal Killing vector field $Y$ on $\scri$ the  condition \eq{condition_on_D}, which is necessary for the rescaled Mars-Simon tensor to vanish on $\scri$, ensures that $Y$ extends to a space-time vector field  $\widetilde X$ satisfying the unphysical Killing equations, and thereby to a Killing vector field $X$ in the physical space-time, which guarantees that $\sigma_{\mu}$ has a potential.

It is worth to emphasize the roles of $\widehat C_{ij}$ and $D_{ij}$ which enter Theorem~\ref{thm_nec_cond}
as well as Theorem~\ref{prop_Qs} in a completely symmetric manner.

%\marc{text added}
%\tim{new subsection ``Asymptotically KdS-like space-times''?
%\\ --\\
%Jose: done}

\subsection{(Asymptotically) KdS-like spacetimes}
\label{subsec:AKdS}

%\tim{additions and rewordings throughout the  section}
In Theorem \ref{prop_Qs} we have obtained necessary and sufficient conditions for 
the rescaled MST 
$\widetilde{\mathcal{T}}^{(\mathrm{ev})}_{\mu\nu\sigma\rho}$ 
 to be regular at $\scri^-$.
As we shall see
%\jose{Added} 
in Section~\ref{subsec:FOSH},
this tensor satisfies a homogeneous symmetric hyperbolic Fuchsian system
with data prescribed at $\scri^-$. The zero data is such that its propagation
stays zero. The resulting spacetime has vanishing MST and hence either 
a Kerr-de Sitter metric or one of the related metrics classified in \cite{mars_senovilla}. 
We call such spacetimes \emph{KdS-like}:

\begin{definition}
\label{KdS_like}
Let  $(\mcM,g)$
 be a $\Lambda>0$-vacuum spacetime 
admitting smooth  conformal compactification and corresponding null
infinity $\scri$.  $(\mcM,g)$ is called
``Kerr-de Sitter-like'' at a connected component
$\scri^{-}$ of $\scri$ if it admits a KVF $X$ which induces
a CKVF $Y$ on $\scri^-$, 
such  that the rescaled MST $\widetilde{\mathcal{T}}^{(\mathrm{ev})}_{\mu\nu\sigma\rho}$  vanishes  at $\scri^-$.
\end{definition}

Note that also the Kerr-NUT-de Sitter spacetime  belongs to the class of KdS-like space-times.
In \cite{mpss} we analyze in detail KdS-like space-times which admit a conformally flat $\scri$.

The case where the tensor 
$\widetilde{\mathcal{T}}^{(\mathrm{ev})}_{\mu\nu\sigma\rho}$ is merely assumed to be finite
at $\scri^-$ obviously includes the zero case (i.e. Kerr-de Sitter and
related metrics) 
and at the same time excludes many other $\Lambda$-vacuum
spacetimes with a smooth $\scri^-$. It makes sense to call such spacetimes
{\it asymptotically Kerr-de Sitter-like}. We put forward the following
definition:
%\tim{motivation... use evolution eqns to prove some kind of stability result for this class of spacetimes...?}

%\marc{reworded, because a metric may asymptotically Kerr-de Sitter like in one $\scri$  but not in others}
\begin{definition}
\label{asympt_KdS_like}
Let  $(\mcM,g)$
 be a $\Lambda>0$-vacuum spacetime 
admitting smooth  conformal compactification and corresponding null
infinity $\scri$.  $(\mcM,g)$ is called
``asymptotically Kerr-de Sitter-like'' at a connected component
$\scri^{-}$ of $\scri$ if it admits a KVF $X$ which induces
a CKVF $Y$ on $\scri^-$,  which satisfies $|Y|^2>0$,
 such that the conditions (i) and (ii) in Theorem~\ref{prop_Qs} are satisfied, or, equivalently,
such that the rescaled MST $\widetilde{\mathcal{T}}^{(\mathrm{ev})}_{\mu\nu\sigma\rho}$  is regular at $\scri^-$.
%\tim{check compatibility with KID equation}
\end{definition}

\begin{remark}
{\rm
As will be shown  later (cf.\ Corollary~\ref{cor_charact_KdS}), KdS-like space-times have a vanishing MST, whence, as shown in Section~\ref{section_recaled_MS}, the condition
$|Y|^2>0$ follows automatically. In the asymptotically KdS-like case, though, the conditions (i) and (ii) in Theorem~\ref{prop_Qs} might be compatible with zeros of $Y$.
}
\end{remark}

%\tim{paragraph should be reworded or removed (if we expect that an existence result a la Ames et al. can be established)} 
%Given the regular evolution  for 
%$\widetilde{\mathcal{T}}^{(\mathrm{ev})}_{\mu\nu\sigma\rho}$ to be derived
%\jose{added} 
%in the next sections, 
% ``asymptotically Kerr-de Sitter-like'' spacetimes encompass precisely
%those $\Lambda>0$-vacuum spacetimes
%for which the rescaled MST 
% can be determined from asymptotic  initial data 
%$\widetilde{\mathcal{T}}^{(\mathrm{ev})}_{\mu\nu\sigma\rho}|_{\scri^-}$.

An interesting open problem is to classify ``asymptotically Kerr-de
Sitter-like'' spacetimes.

\subsection{Derivation of  evolution equations for the (rescaled) MST}
\label{sect_deriv_ev_MST}

Based on the corresponding derivation for $\Lambda=0$ in \cite{ik}, we will show that the MST
\begin{eqnarray}
\mathcal{S}^{(\mathrm{ev})}_{\mu\nu\sigma\rho} &=& \mathcal{C}_{\mu\nu\sigma\rho}  + Q_{\mathrm{ev}} \mathcal{U}_{\mu\nu\sigma\rho}
\;,
%\\
%&=& \mathcal{C}_{\alpha\beta\mu\nu}  + 6 \mathcal{R}_{\alpha\beta\mu\nu}
%\;,
\end{eqnarray}
with $ Q_{\mathrm{ev}}$ as defined in \eq{ev_dfn_Q}-\eq{defJ},
%where
%%
%\begin{eqnarray*}
%\mathcal{R}_{\alpha\beta\mu\nu} &:=& \frac{1}{6} Q_{\mathrm{ev}} \mathcal{Q}_{\alpha\beta\mu\nu}
%\\
%&=& \frac{1}{2}\Big(\frac{\Lambda}{3} - RJ \Big)\Big(   \mathcal{W}_{\alpha\beta}\mathcal{W}_{\mu\nu} - \frac{1}{3}\mathcal{W}^2\mathcal{I}_{\alpha\beta\mu\nu} \Big)
%\;,
%\\
% \mathcal{W}_{\alpha\beta} &:=& R^{-1}\mathcal{F}_{\alpha\beta} 
%\;.
%\end{eqnarray*}
%
satisfies a symmetric hyperbolic system of evolution equations as well as a system of wave equations.

\subsubsection{An analog to the Bianchi equation}
First, we derive an analog to  the Bianchi equation $\nabla_{\rho}C_{\mu\nu\sigma}{}^{\rho}=0$
for the MST $\mathcal{S}^{(\mathrm{ev})}_{\mu\nu\sigma}{}^{\rho}$.
For this we set
\begin{eqnarray}
%\nabla_{\rho}\mathcal{S}_{\mu\nu\sigma}{}^{\rho} &=& 6\mathcal{J}(\mathcal{S})_{\sigma\mu\nu}
%\;,
%\label{phys_ev}
%\\
%\mathcal{J}(\mathcal{S})_{\sigma\mu\nu} &:=& \frac{1}{2R}\Big(\frac{\Lambda}{3}-RJ\Big) X^{\alpha}\mathcal{S}_{\alpha\beta\gamma\delta}\Big(\mathcal{W}_{\sigma}{}^{\beta}\delta_{\mu}{}^{\gamma}\delta_{\nu}{}^{\delta}-\frac{2}{3}\mathcal{W}^{\gamma\delta}\mathcal{I}_{\mu\nu\sigma}{}^{\beta}\Big)
%\nonumber
%\\
%&& -\frac{\Lambda}{24R} \frac{R+2\sigma J}{R-\sigma J} X^{\alpha}\mathcal{S}_{\alpha\beta\gamma\delta}\mathcal{W}^{\gamma\delta}\mathcal{W}_{\sigma}%{}^{\beta}\mathcal{W}_{\mu\nu}
%\;,
%\label{ev_eqn_rhs}
%\\
 \mathcal{W}_{\alpha\beta} &:=& R^{-1}\mathcal{F}_{\alpha\beta} 
\;,
\end{eqnarray}
so that 
\begin{eqnarray}
Q_{\mathrm{ev}} \mathcal{U}_{\alpha\beta\mu\nu} &=& \Big( \frac{3 J}{R} - \frac{\Lambda}{R^2}\Big) \Big(- \mathcal{F}_{\alpha\beta}  \mathcal{F}_{\mu\nu}+ \frac{1}{3}\mathcal{F}^2\mathcal{I}_{\alpha\beta\mu\nu} \Big)
\\
&=&   (  \Lambda -3 J R ) \Big( \mathcal{W}_{\alpha\beta}  \mathcal{W}_{\mu\nu} +  \frac{4}{3}\mathcal{I}_{\alpha\beta\mu\nu} \Big)
\;.
\end{eqnarray}
Differentiation yields
\begin{eqnarray}
\nabla_{\rho} (Q_{\mathrm{ev}} \mathcal{U}_{\alpha\beta\mu\nu})
&=&  -3 \nabla_{\rho}(J R) \Big( \mathcal{W}_{\alpha\beta}  \mathcal{W}_{\mu\nu} +  \frac{4}{3}\mathcal{I}_{\alpha\beta\mu\nu} \Big)
\nonumber
\\
&&
+ (  \Lambda -3 J R )(  \mathcal{W}_{\mu\nu} \nabla_{\rho}\mathcal{W}_{\alpha\beta}   +  \mathcal{W}_{\alpha\beta}  \nabla_{\rho} \mathcal{W}_{\mu\nu}  )
\;.
\end{eqnarray}
First of all let us calculate the covariant derivative of $\mathcal{W}_{\mu\nu}$. From
\begin{eqnarray}
\nabla_{\rho}\mathcal{F}_{\mu\nu} &=& X^{\sigma}\Big(\mathcal{C}_{\mu\nu\sigma\rho} + \frac{4}{3}\Lambda \mathcal{I}_{\mu\nu\sigma\rho}\Big)
\\
 &=& X^{\sigma}\Big(\mathcal{S}_{\mu\nu\sigma\rho} - Q_{\mathrm{ev}} \mathcal{U}_{\mu\nu\sigma\rho} + \frac{4}{3}\Lambda \mathcal{I}_{\mu\nu\sigma\rho}\Big)
\\
 &=&   \frac{ 3 J R - \Lambda}{2R^2} \sigma_{\rho}\mathcal{F}_{\mu\nu} 
+4   J RX^{\sigma} \mathcal{I}_{\mu\nu\sigma\rho} + X^{\sigma}\mathcal{S}_{\mu\nu\sigma\rho} 
\;,
\\
\nabla_{\rho}\mathcal{F}^2 &=&  2\mathcal{F}^{\mu\nu}\nabla_{\rho}\mathcal{F}_{\mu\nu}
\\
&=& -4 ( 2 J R - \Lambda) \sigma_{\rho}
+ 2X^{\sigma} \mathcal{F}^{\mu\nu}\mathcal{S}_{\mu\nu\sigma\rho} 
\\
&=& \frac{1}{3} ( 2 Q_{\mathrm{ev}}\mathcal{F}^2  +4 \Lambda) \sigma_{\rho}
+ 2X^{\sigma} \mathcal{F}^{\mu\nu}\mathcal{S}_{\mu\nu\sigma\rho} 
\;,
\\
\nabla_{\rho} R &=& -\frac{1}{2} \nabla_{\rho}\sqrt{-\mathcal{F}^2}\,=\, -\frac{1}{8}R^{-1} \nabla_{\rho}\mathcal{F}^2
\\
&=&\frac{1}{6} \Big( 2 Q_{\mathrm{ev}}R - \frac{\Lambda}{R}\Big) \sigma_{\rho}
 -\frac{1}{4}X^{\sigma} \mathcal{W}^{\mu\nu}\mathcal{S}_{\mu\nu\sigma\rho} 
\\
&=& \Big( J-\frac{\Lambda}{2R} \Big) \sigma_{\rho}
 -\frac{1}{4}X^{\sigma} \mathcal{W}^{\mu\nu}\mathcal{S}_{\mu\nu\sigma\rho} 
\;,
\end{eqnarray}
we deduce that
\begin{eqnarray}
 \nabla_{\rho}\mathcal{W}_{\mu\nu} &=&  R^{-1} \nabla_{\rho}\mathcal{F}_{\mu\nu} -R^{-1} \mathcal{W}_{\mu\nu} \nabla_{\rho} R
\\
 &=&   
 \frac{  J  }{2R} \sigma_{\rho}\mathcal{W}_{\mu\nu} 
+4   J X^{\sigma} \mathcal{I}_{\mu\nu\sigma\rho}
 +  R^{-1}  X^{\sigma}\mathcal{S}_{\mu\nu\sigma\rho} 
\nonumber
\\
&& +  \frac{1}{4R} X^{\sigma} \mathcal{W}_{\mu\nu}\mathcal{W}^{\alpha\beta}\mathcal{S}_{\alpha\beta\sigma\rho} 
\;,
\\
 \nabla^{\nu}\mathcal{W}_{\mu\nu} 
 &=&  
 \frac{  J  }{2R} \sigma^{\nu}\mathcal{W}_{\mu\nu}  + 3JX_{\mu}
 +\frac{1}{4R}  X^{\sigma} \mathcal{W}_{\mu}{}^{\rho}\mathcal{W}^{\alpha\beta}\mathcal{S}_{\alpha\beta\sigma\rho} 
\\
 &=&  
 2 J X_{\mu}
 +\frac{1}{4R}  X^{\sigma} \mathcal{W}_{\mu}{}^{\rho}\mathcal{W}^{\alpha\beta}\mathcal{S}_{\alpha\beta\sigma\rho} 
\;,
\end{eqnarray}
where we used that $\sigma^{\rho}\mathcal{F}_{\mu\rho} = \frac{1}{2}\mathcal{F}^2X_{\mu}$ \cite{mars_senovilla}.

Taking the derivative of \eq{quadratic}, we obtain
\begin{eqnarray}
\nabla_{\rho} (JR)
% &=& \nabla_{\rho} \Big( R \frac{R-\sqrt{R^2-\Lambda\sigma}}{\sigma}\Big) 
%\\
&=&
\frac{  J^2\sigma }{R-J\sigma} \Big(  \frac{R}{2\sigma}\sigma_{\rho}  -  \nabla_{\rho}  R\Big)
\\
&=& 
\frac{  J }{R-J\sigma}\Big(  \frac{-3 JR^2 + 2R\Lambda + \Lambda \sigma  J}{2R} \sigma_{\rho}
 + \frac{1}{4}  J\sigma   X^{\sigma} \mathcal{W}^{\mu\nu}\mathcal{S}_{\mu\nu\sigma\rho}  
\Big)
\phantom{xxxx}
\\
&=& 
\frac{J}{2R} (3JR-\Lambda) \sigma_{\rho}
 + \frac{  J^2\sigma }{4(R-J\sigma)}    X^{\sigma} \mathcal{W}^{\mu\nu}\mathcal{S}_{\mu\nu\sigma\rho}  
\;.
\end{eqnarray}

%\tim{up to this stage it coincides with Walter's computation, haven't compared it from now on...}
Altogether we find 
\begin{eqnarray}
\nabla_{\rho} (Q_{\mathrm{ev}} \mathcal{U}_{\alpha\beta\mu}{}^{\rho})
&=&  J  (3JR-\Lambda) 
 \Big(2   X_{\mu} \mathcal{W}_{\alpha\beta}-   \frac{2}{R}\sigma^{\rho} \mathcal{I}_{\alpha\beta\mu\rho}
-4   X^{\sigma}\mathcal{W}_{\mu}{}^{\rho} \mathcal{I}_{\alpha\beta\sigma\rho}
  \Big)
\nonumber
\\
&&
 + (  \Lambda -3 J R) R^{-1}  X^{\sigma}\mathcal{W}_{\mu}{}^{\rho}\Big(    \mathcal{S}_{\alpha\beta\sigma\rho} 
+ \frac{1}{2}\mathcal{W}_{\alpha\beta} \mathcal{W}^{\gamma\delta}\mathcal{S}_{\gamma\delta\sigma\rho}  
\Big)
\nonumber
\\
&&
 -  \frac{ 3 J^2\sigma }{4(R-J\sigma)}    X^{\sigma} \mathcal{W}^{\gamma\delta}\mathcal{S}_{\gamma\delta\sigma\rho}  
\Big( \mathcal{W}_{\alpha\beta}  \mathcal{W}_{\mu}{}^{\rho} +  \frac{4}{3}\mathcal{I}_{\alpha\beta\mu}{}^{\rho} \Big)
\;.
\phantom{xxxxx}
\end{eqnarray}
Using
\begin{equation}
 \mathcal{F}_{(\mu}{}^{\sigma}\mathcal{I}_{\nu)\sigma\alpha\beta}
\,=\,
\frac{1}{4} g_{\mu\nu}\mathcal{F}_{\alpha\beta}
\;,
\end{equation}
cf.\ \cite[Equation (4.37)]{ik},
we obtain
\begin{eqnarray}
 \frac{2}{R}\sigma^{\rho} \mathcal{I}_{\alpha\beta\mu\rho}
+4   X^{\sigma}\mathcal{W}_{\mu}{}^{\rho} \mathcal{I}_{\alpha\beta\sigma\rho}
&=& 4X^{\sigma}(\mathcal{W}_{\sigma}{}^{\rho} \mathcal{I}_{\alpha\beta\mu\rho}
+\mathcal{W}_{\mu}{}^{\rho} \mathcal{I}_{\alpha\beta\sigma\rho})
%\\
%&=& 2g_{\mu\sigma}X^{\sigma}\mathcal{W}^{\gamma\delta}\mathcal{I}_{\alpha\beta\gamma\delta}
\\
&=& 2X_{\mu}\mathcal{W}_{\alpha\beta}
\;,
\end{eqnarray}
and thus
\begin{eqnarray}
\nabla_{\rho} \mathcal{S}^{(\mathrm{ev})}_{\alpha\beta\mu}{}^{\rho}
&=&
 (  \Lambda -3 J R)\mathcal{W}_{\mu}{}^{\rho}\Big(     R^{-1}  X^{\sigma}\mathcal{S}^{(\mathrm{ev})}_{\alpha\beta\sigma\rho} 
+ \frac{1}{2} R^{-1}  X^{\sigma} \mathcal{W}_{\alpha\beta} \mathcal{W}^{\gamma\delta}\mathcal{S}^{(\mathrm{ev})}_{\gamma\delta\sigma\rho}  
\Big)
\nonumber
\\
&&
 -  \frac{ 3 J^2\sigma }{4(R-J\sigma)}    X^{\sigma} \mathcal{W}^{\gamma\delta}\mathcal{S}^{(\mathrm{ev})}_{\gamma\delta\sigma\rho}  
\Big( \mathcal{W}_{\alpha\beta}  \mathcal{W}_{\mu}{}^{\rho} +  \frac{4}{3}\mathcal{I}_{\alpha\beta\mu}{}^{\rho} \Big)
%\\
%&=&
% 6\frac{  Q_{\mathrm{ev}}\mathcal{F}^2 + 2\Lambda  }{  Q_{\mathrm{ev}}\mathcal{F}^2\sigma-  4\Lambda  \sigma -3\mathcal{F}^2 } \mathcal{F}^{-2}  X^{\sigma} \mathcal{F}^{\gamma\delta}\mathcal{Q}_{\alpha\beta\mu\rho} \mathcal{S}_{\gamma\delta\sigma}{}^{\rho}  
%\\
%&&
% -  Q_{\mathrm{ev}}\mathcal{F}_{\mu\rho} X^{\sigma}\Big(   \mathcal{S}_{\alpha\beta\sigma}{}^{\rho}  -2 \mathcal{F}^{-2}  \mathcal{F}_{\alpha\beta} \mathcal{F}^{\gamma\delta}\mathcal{S}_{\gamma\delta\sigma}{}^{\rho}   \Big)
%\\
%&=&
%  - \frac{  3J^2\sigma  }{R^2-JR\sigma} \mathcal{F}^{-2}  X^{\sigma} \mathcal{F}^{\gamma\delta}\mathcal{Q}_{\alpha\beta\mu\rho}\mathcal{S}_{\gamma\delta\sigma}{}^{\rho}  
%\\
%&&
% -4 (  \Lambda -3 J R)\mathcal{F}^{-2}\mathcal{F}_{\mu\rho} X^{\sigma}\Big( \mathcal{I}_{\alpha\beta}{}^{\gamma\delta} 
%-2 \mathcal{F}^{-2}  \mathcal{F}_{\alpha\beta} \mathcal{F}^{\gamma\delta}
%\Big) \mathcal{S}_{\gamma\delta\sigma}{}^{\rho} 
\\
&=&
\frac{  \Lambda -3 J R}{R}\mathcal{W}_{\mu}{}^{\rho}   X^{\sigma}\mathcal{S}^{(\mathrm{ev})}_{\alpha\beta\sigma\rho} 
 -  \frac{  J^2\sigma }{R-J\sigma}    X^{\sigma}
 \mathcal{I}_{\alpha\beta\mu}{}^{\rho}  \mathcal{W}^{\gamma\delta} \mathcal{S}^{(\mathrm{ev})}_{\gamma\delta\sigma\rho}  
\nonumber
\\
&&-\frac{\Lambda}{4R} \frac{ R+  2  J\sigma }{R-J\sigma}  X^{\sigma} \mathcal{W}_{\alpha\beta} \mathcal{W}^{\gamma\delta}  \mathcal{W}_{\mu}{}^{\rho} \mathcal{S}^{(\mathrm{ev})}_{\gamma\delta\sigma\rho}  
%\\
%&=&
%-Q_{\mathrm{ev}} \mathcal{F}_{\mu\rho}   X^{\sigma}\mathcal{S}_{\alpha\beta\sigma}{}^{\rho} 
% +  \frac{  2Q_{\mathrm{ev}}\mathcal{F}^2 +4 \Lambda }{3\mathcal{F}^2  - \sigma Q_{\mathrm{ev}}\mathcal{F}^2 +4 \Lambda\sigma}      X^{\sigma} \mathcal{F}_{\delta}{}^{\gamma}  \mathcal{I}_{\alpha\beta\mu}{}^{\rho}  \mathcal{S}_{\sigma\rho\gamma}{}^{\delta}  
%\\
%&&+4 \Lambda \frac{5Q_{\mathrm{ev}}\mathcal{F}^2 + 4\Lambda }{ Q_{\mathrm{ev}}\mathcal{F}^2 + 8\Lambda}  X^{\sigma} \mathcal{F}^{-4} \mathcal{F}_{\alpha\beta} \mathcal{F}^{\gamma\delta}  \mathcal{F}_{\mu\rho} \mathcal{S}_{\gamma\delta\sigma}{}^{\rho}  
\;.
\end{eqnarray}
Expressed in terms of $\mathcal{F}_{\mu\nu}$ and $Q_{\mathrm{ev}}$ we finally end up with the desired  equation for the MST, which may be regarded
as an analog to the Bianchi equation,
%\tim{actually it's quite a nice expression now}
%
\begin{eqnarray}
\nabla_{\rho} \mathcal{S}^{(\mathrm{ev})}_{\alpha\beta\mu}{}^{\rho}
%&=&
% \Big(4 \Lambda \frac{5Q_{\mathrm{ev}}\mathcal{F}^2 + 4\Lambda }{ Q_{\mathrm{ev}}\mathcal{F}^2 + 8\Lambda}  \mathcal{F}^{-4} \mathcal{F}_{\alpha\beta} \mathcal{F}_{\mu\rho}  -  \frac{  2Q_{\mathrm{ev}}\mathcal{F}^2 +4 \Lambda }{3\mathcal{F}^2  - \sigma Q_{\mathrm{ev}}\mathcal{F}^2 +4 \Lambda\sigma}    
% \mathcal{I}_{\alpha\beta\mu\rho} 
% \Big) 
%\nonumber
%\\
%&&
%\times X^{\sigma} \mathcal{F}^{\gamma\delta}  \mathcal{S}^{(\mathrm{ev})}_{\gamma\delta\sigma}{}^{\rho}  
%-Q_{\mathrm{ev}} \mathcal{F}_{\mu\rho}   X^{\sigma}\mathcal{S}^{(\mathrm{ev})}_{\alpha\beta\sigma}{}^{\rho} 
%\\
&=&
 \Big(4 \Lambda \frac{5Q_{\mathrm{ev}}\mathcal{F}^2 + 4\Lambda }{ Q_{\mathrm{ev}}\mathcal{F}^2 + 8\Lambda} \mathcal{F}_{\alpha\beta} \mathcal{F}_{\mu\rho}  + \frac{2}{3}\mathcal{F}^{2} \frac{ Q_{\mathrm{ev}}^2\mathcal{F}^4 -2 \Lambda Q_{\mathrm{ev}}\mathcal{F}^2 -8\Lambda^2 }{Q_{\mathrm{ev}}\mathcal{F}^2 + 8\Lambda}
 \mathcal{I}_{\alpha\beta\mu\rho} 
 \Big) 
\nonumber
\\
&&
\times   \mathcal{F}^{-4} X^{\sigma} \mathcal{F}^{\gamma\delta}  \mathcal{S}^{(\mathrm{ev})}_{\gamma\delta\sigma}{}^{\rho}  
-Q_{\mathrm{ev}} \mathcal{F}_{\mu\rho}   X^{\sigma}\mathcal{S}^{(\mathrm{ev})}_{\alpha\beta\sigma}{}^{\rho} 
\label{full_eqn_MST}
\\
&=&
-4 \Lambda  \frac{  5  Q_{\mathrm{ev}}\mathcal{F}^2  +4\Lambda }{Q_{\mathrm{ev}}\mathcal{F}^2 + 8\Lambda}
 \mathcal{U}_{\alpha\beta\mu\rho} 
   \mathcal{F}^{-4} X^{\sigma} \mathcal{F}^{\gamma\delta}  \mathcal{S}^{(\mathrm{ev})}_{\gamma\delta\sigma}{}^{\rho}  
\nonumber
\\
&& +  Q_{\mathrm{ev}}X^{\sigma}\Big( \frac{2}{3}
 \mathcal{I}_{\alpha\beta\mu\rho}  \mathcal{F}^{\gamma\delta}  \mathcal{S}^{(\mathrm{ev})}_{\gamma\delta\sigma}{}^{\rho}  
- \mathcal{F}_{\mu\rho} \mathcal{S}^{(\mathrm{ev})}_{\alpha\beta\sigma}{}^{\rho} 
\Big)
\\
&=:& \mathcal{J}(\mathcal{S^{(\mathrm{ev})}})_{\alpha\beta\mu}
\;.
\label{phys_ev}
\end{eqnarray}
%
%\jose{Some additions}
Here we have introduced the shorthand $ \mathcal{J}(\mathcal{S^{(\mathrm{ev})}})_{\alpha\beta\mu}$ for the righthand side, which is a double $(2,1)$-form,
 linear and homogeneous in $\mathcal{S}^{(\mathrm{ev})}_{\alpha\beta\sigma}{}^{\rho}$, with the following properties
\begin{equation}
 \mathcal{J}(\mathcal{S^{(\mathrm{ev})}})_{\alpha\beta\mu}= \mathcal{J}(\mathcal{S^{(\mathrm{ev})}})_{[\alpha\beta]\mu},
 \hspace{3mm}  \mathcal{J}(\mathcal{S^{(\mathrm{ev})}})_{[\alpha\beta\mu]}=0, \hspace{3mm}
  \mathcal{J}(\mathcal{S^{(\mathrm{ev})}})^\rho{}_{\beta\rho}=0 . \label{Jsymprop}
\end{equation}
 It is also self-dual in the first pair of anti-symmetric indices:
$\mathcal{J}^{\star}(\mathcal{S^{(\mathrm{ev})}})_{\alpha\beta\mu}=-i \mathcal{J}(\mathcal{S^{(\mathrm{ev})}})_{\alpha\beta\mu}$.
%
%\begin{eqnarray*}
%\frac{  2Q_{\mathrm{ev}}\mathcal{F}^2 +4 \Lambda }{3\mathcal{F}^2  - \sigma Q_{\mathrm{ev}}\mathcal{F}^2 +4 \Lambda\sigma}    
%&=&
%-\frac{2}{3}\mathcal{F}^{-2} \frac{ Q_{\mathrm{ev}}^2\mathcal{F}^4 -2 \Lambda Q_{\mathrm{ev}}\mathcal{F}^2 -8\Lambda^2 }{Q_{\mathrm{ev}}\mathcal{F}^2 + 8\Lambda}
%\end{eqnarray*}

%\begin{eqnarray*}
%\mathcal{F}^{\alpha\beta}\nabla_{\rho} \mathcal{S}_{\alpha\beta\mu}{}^{\rho}
%&=&
% \Big(4 \Lambda \frac{5Q_{\mathrm{ev}}\mathcal{F}^2 + 4\Lambda }{ Q_{\mathrm{ev}}\mathcal{F}^2 + 8\Lambda}  \mathcal{F}^{-2}  -  \frac{  2Q_{\mathrm{ev}}\mathcal{F}^2 +4 \Lambda }{3\mathcal{F}^2  - \sigma Q_{\mathrm{ev}}\mathcal{F}^2 +4 \Lambda\sigma}      -Q_{\mathrm{ev}} 
% \Big) 
%\mathcal{F}_{\mu\rho}    X^{\sigma} \widetilde{ \mathcal{S}}_{\sigma}{}^{\rho}  
%\end{eqnarray*}
%\begin{eqnarray*}
%\nabla_{\rho}(\mathcal{F}^{\alpha\beta} \mathcal{S}_{\alpha\beta\mu}{}^{\rho})
%&=&
% \Big(4 \Lambda \frac{5Q_{\mathrm{ev}}\mathcal{F}^2 + 4\Lambda }{ Q_{\mathrm{ev}}\mathcal{F}^2 + 8\Lambda}  \mathcal{F}^{-2}  -  \frac{  %2Q_{\mathrm{ev}}\mathcal{F}^2 +4 \Lambda }{3\mathcal{F}^2  - \sigma Q_{\mathrm{ev}}\mathcal{F}^2 +4 \Lambda\sigma}    
% -Q_{\mathrm{ev}} 
% \Big) 
%\mathcal{F}_{\mu\rho}    X^{\sigma} \widetilde{ \mathcal{S}}_{\sigma}{}^{\rho}  
%\\
%&&
%+   \frac{ 3 J R - \Lambda}{2R^2} \sigma_{\rho}\widetilde{ \mathcal{S}}{}_{\mu}{}^{\rho}
%+ X^{\sigma}\mathcal{S}_{\alpha\beta\sigma\rho} \mathcal{S}^{\alpha\beta}{}_{\mu}{}^{\rho}
%\end{eqnarray*}

%\jose{Lemma moved to new subsection \ref{subsec:FOSH}}

Using the fact that the MST $\mathcal{S}^{(\mathrm{ev})}_{\mu\nu\sigma}{}^{\rho}$ has all the algebraic symmetries of the Weyl tensor, we immediately obtain
a Bianchi-like equation  from \eq{phys_ev}
 for the rescaled MST $\widetilde{\mathcal{T}}^{(\mathrm{ev})}_{\mu\nu\sigma}{}^{\rho}$  in the conformally rescaled ``unphysical'' spacetime,
\begin{equation}
 \widetilde\nabla_{\rho}\widetilde{\mathcal{T}}^{(\mathrm{ev})}_{\alpha\beta\mu}{}^{\rho} \,=\, \Theta^{-1}\nabla_{\rho}\mathcal{S}^{(\mathrm{ev})}_{\alpha\beta\mu}{}^{\rho}
\,=\,  \Theta^{-1} \mathcal{J}(\mathcal{S}^{(\mathrm{ev})})_{\alpha\beta\mu}\,=\,   \mathcal{J}(\widetilde{\mathcal{T}}^{(\mathrm{ev})})_{\alpha\beta\mu}
\;.
\label{unphys_ev}
\end{equation}

%\subsubsection{The evolution equations for the (rescaled) MST are symmetric hyperbolic }
\subsubsection{A symmetric hyperbolic system satisfied  by the  (rescaled) MST}
\label{subsec:FOSH}

We now want to show that the
equations \eq{phys_ev} and \eq{unphys_ev} contain a system of linear first-order symmetric hyperbolic equations in their respective spacetimes. 
%\tim{strictly speaking they merely contain evolution equation since they also contain constraint equations???}
Given that \eq{phys_ev} and \eq{unphys_ev} have exactly the same structure, it is enough to perform the analysis for any one of the two systems, or to a model equivalent system in a given spacetime. Then, a further analysis of the regularity 
of the system \eq{unphys_ev} near $\scri$ is also needed, and this will be done later in this section.

Let $\mathcal{S}_{\alpha\beta\lambda}{}^\mu$ represent either $\mathcal{S}^{(\mathrm{ev})}_{\alpha\beta\lambda}{}^{\mu}$ 
or $\widetilde{\mathcal{T}}^{(\mathrm{ev})}_{\alpha\beta\lambda}{}^{\mu}$, or for that
matter, any other self-dual symmetric and traceless  double (2,2)-form satisfying a system of equations such as \eq{phys_ev} or \eq{unphys_ev}:
\begin{equation}
\nabla_\rho {\cal S}_{\gamma\mu\nu}{}^\rho= {\cal J}_{\gamma\mu\nu}({\cal S}), \hspace{1cm} 
%{\cal J}_{\gamma\mu\nu}({\cal T}):=   J^{\rho}{}_{\mu\nu\rho\gamma} 
\label{delta1T}
\end{equation}
where $ {\cal J}_{\gamma\mu\nu}({\cal S})$ is a self-dual double $(2,1)$-form, 
linear and homogeneous in $\mathcal{S}_{\alpha\beta\sigma}{}^{\rho}$, with the properties given in \eq{Jsymprop}. 
Employing the  fact that the rescaled MST satisfies all the algebraic symmetries of the rescaled Weyl tensor 
we find that this system is equivalent to (cf.\ \cite{penrose, ik})
\begin{equation}
\label{d1tildeT}
3 \nabla_{[\sigma}{\mathcal{S}}_{\mu\nu]\alpha\beta} =
-\frac{1}{2} \volform_{\sigma\mu\nu\kappa} \volform^{\gamma\delta\rho\kappa}\nabla_{\gamma}
{\mathcal{S}}_{\delta\rho\alpha\beta} 
= \volform_{\sigma\mu\nu}{}^{\kappa}\nabla_{\gamma}
({\mathcal{S}}_{\alpha\beta\kappa}{}^{\gamma} )^{\star}
= - i \volform_{\sigma\mu\nu}{}^{\kappa}\nabla_{\gamma}
{\mathcal{S}}_{\alpha\beta\kappa}{}^{\gamma} 
\end{equation}
that is to say
\be
3 \nabla_{[\sigma}{\mathcal{S}}_{\mu\nu]\alpha\beta} =
- i\volform_{\mu\nu\sigma}{}^{\rho}
\mathcal{J}({\mathcal{S}})_{\alpha\beta\rho}\label{d1T}
\ee

%As proven in the calculation leading to \eq{d1tildeT}, this system is equivalent to 
%\begin{equation}
%3\nabla_{[\alpha}{\cal S}_{\beta\gamma]\mu\nu}=-i \volform_{\alpha\beta\gamma}{}^\rho \mathcal{J}_{\mu\nu\rho}({\cal S}). \label{d1T}
%\end{equation}

Observe that each of \eq{delta1T} and \eq{d1T} contains 8 complex (16 real) independent equations for only 5 complex (10 real) unknowns, hence they are overdetermined.

Systems of this type have been analyzed many times in the literature (in order to see if they comprise a symmetric hyperbolic system), 
especially in connection with the Bianchi identities \cite{ChY2,Bo}. 
Here, to check that the systems \eq{delta1T}, or \eq{d1T}, and therefore \eq{phys_ev} and \eq{unphys_ev}, contain symmetric hyperbolic evolution equations we use the general ideas exposed in \cite{G}, 
which were applied to systems more general than --and including-- those of type 
(\ref{delta1T}) or (\ref{d1T}) and discussed at length  in \cite{S}. The goal is to find a ``hyperbolization'', in the sense of \cite{G}. 
To that end, simplifying the calculations in \cite{S}, pick {\em any} timelike vector $v^\mu$ and contract \eq{delta1T} with
\begin{equation}
\label{Factor1}
- v_{[\alpha} \delta^\nu_{\beta]}v^{[\gamma} v_{[\delta} \delta^{\mu]}_{\epsilon]}
\end{equation}
and add the result to the contraction of \eq{d1T} with
\begin{equation}
\label{Factor2}
-\frac{1}{2} v^{[\lambda} \delta^\tau_\alpha \delta^{\nu]}_\beta v^{[\gamma} v_{[\delta} \delta^{\mu]}_{\epsilon]}
\end{equation}
to arrive at the following system
\begin{equation}
Q^{\tau}{}^{\lambda\nu\gamma\mu}_{\alpha\beta\delta\epsilon} \nabla_\tau {\cal S}_{\lambda\nu\gamma\mu} = {\cal J}_{\alpha\beta\delta\epsilon} \label{QnablaS}
\end{equation}
with
$$
Q^{\tau}{}^{\lambda\nu\gamma\mu}_{\alpha\beta\delta\epsilon} =v^{[\gamma} v_{[\delta} \delta^{\mu]}_{\epsilon]}\Big(g^{\tau[\lambda}v_{[\alpha} \delta^{\nu]}_{\beta]} +\delta^\tau_{[\alpha}v^{[\lambda} \delta^{\nu]}_{\beta]} -\frac{1}{2} v^\tau \delta^{\lambda}_{[\alpha}\delta^\nu_{\beta]} \Big)\, .
$$
By construction, the righthand side of (\ref{QnablaS}) is linear in ${\cal J}_{\gamma\mu\nu}(S)$ and {\em a fortiori} linear in ${\cal S}_{\gamma\mu\nu\rho}$, 
so that its explicit expression is unimportant. The system \eq{QnablaS} is symmetric hyperbolic. 
To prove it, we have to check two properties of $Q^{\tau}{}^{\lambda\nu\gamma\mu}_{\alpha\beta\delta\epsilon}$: 
it must be Hermitian in $\lambda\nu\gamma\mu \leftrightarrow {\alpha\beta\delta\epsilon}$, 
and there must exist a one-form $u_\tau$ such that its contraction with $Q^{\tau}{}^{\lambda\nu\gamma\mu}_{\alpha\beta\delta\epsilon}$ is positive definite.

The first condition can be easily checked by first noticing that $Q^{\tau}{}^{\lambda\nu\gamma\mu}_{\alpha\beta\delta\epsilon}$ 
happens to be real and then contracting it with two arbitrary self-dual trace-free double (2,2) forms, say ${\cal A}_{\lambda\nu\gamma\mu}$ and ${\cal B}^{\alpha\beta\delta\epsilon}$. The result is %(up to a factor $1/2$)
\begin{eqnarray*}
v^\beta v^\lambda v^\mu \Big[{\cal A}^\tau{}_{\rho\lambda\sigma} {\cal B}_{\beta}{}^\rho{}_\mu{}^\sigma +{\cal B}^\tau{}_{\rho\lambda\sigma} {\cal A}_{\beta}{}^\rho{}_\mu{}^\sigma -\frac{1}{2}\delta^\tau_\beta {\cal A}_{\alpha\rho\lambda\sigma} {\cal B}^{\alpha\rho}{}_\mu{}^\sigma \Big]
\;,
\end{eqnarray*}
which is manifestly symmetric under the interchange of ${\cal A}$ and ${\cal B}$. Thus, the matrix of the system \eq{QnablaS} is Hermitian.

With regard to the second condition, we contract $Q^{\tau}{}^{\lambda\nu\gamma\mu}_{\alpha\beta\delta\epsilon}$ with 
${\cal A}_{\lambda\nu\gamma\mu}$, $\overline{\cal A}^{\alpha\beta\delta\epsilon}$,
and with $u_\tau$.
We stress that in this section an overbar means ``complex conjugation'' rather than ``restriction to $\scri$''.
We get
\begin{eqnarray*}
u_\tau v^\beta v^\lambda v^\mu  ({\cal A}^\tau{}_{\rho\lambda\sigma} \overline{\cal A}_{\beta}{}^\rho{}_\mu{}^\sigma +\overline{\cal A}^\tau{}_{\rho\lambda\sigma} {\cal A}_{\beta}{}^\rho{}_\mu{}^\sigma )= 2\,  u_\tau v^\beta v^\lambda v^\mu  ({\cal A}^\tau{}_{\rho\lambda\sigma} \overline{\cal A}_{\beta}{}^\rho{}_\mu{}^\sigma )
\;
\end{eqnarray*}
and note that the expression in brackets is precisely the Bel-Robinson superenergy tensor $t_{\tau\beta\lambda\mu}$ of the self-dual Weyl-type tensor ${\cal A}_{\alpha\beta\lambda\mu}$ \cite{penrose,S1}. It is known that this tensor satisfies the dominant property \cite{penrose,S1}, that is, $t_{\tau\beta\lambda\mu}v_1^\tau v_2^\beta v_3^\lambda v_4^\mu >0$ for arbitrary future-pointing  timelike vectors $v_1^\tau, v_2^\beta, v_3^\lambda, v_4^\mu$. Thus, the previous expression is positive for {\em any} timelike $u_\tau$ with the same time orientation as $v^\tau$, as required.
Actually, by choosing $u^{\tau}=v^{\tau}$ it becomes (twice) the so-called super-energy density of ${\cal A}_{\alpha\beta\lambda\mu}$ relative to $v^{\tau}$:
\be
W_{v}({\cal A}) := v^\tau v^\beta v^\lambda v^\mu  {\cal A}_{\tau\rho\lambda\sigma} \overline{\cal A}_{\beta}{}^\rho{}_\mu{}^\sigma \label{s-e}
\ee
which is  non-negative, vanishing if and only if so does the full ${\cal A}_{\alpha\beta\lambda\mu}$ \cite{S1}.
 
In order to see the relation between the found symmetric hyperbolic system \eq{QnablaS}
and the original equations \eq{delta1T} or \eq{d1T}, we take into account that $Q^{\tau}{}^{\lambda\nu\gamma\mu}_{\alpha\beta\delta\epsilon}$ is 
non-degenerate as an operator acting on 2-forms in its 
"$\gamma\mu , \delta\epsilon$" part so that \eq{QnablaS} is actually fully equivalent to
$$
3v^\lambda \nabla_{[\lambda}{\cal S}_{\tau\nu]\gamma\mu}
+2\nabla_\rho {\cal S}_{\gamma\mu[\nu}{}^\rho v_{\tau]} 
=i v^\lambda  {\cal J}_{\gamma\mu\sigma}(S) \volform^\sigma{}_{\lambda\tau\nu}
+2 {\cal J}_{\gamma\mu[\nu}(S)v_{\tau]}
\;.
$$
There is some redundancy here due to the equivalence of \eq{delta1T} and \eq{d1T}. To optimize the expression of this symmetric hyperbolic system we note that, 
via the identity (\ref{d1tildeT}), it can be rewritten as
$$
\overline{\cal I}^{\tau\nu}{}_{\lambda\sigma}\left(\nabla_\rho {\cal S}_{\gamma\mu[\nu}{}^\rho- {\cal J}_{\gamma\mu[\nu}({\cal S})\right)v_{\tau]}=0
$$
which is easily seen to be equivalent to
%\be
%v^\alpha\left(3\nabla_{[\alpha}{\cal S}_{\beta\gamma]\mu\nu}-J_{\alpha\beta\gamma\mu\nu}({\cal T})�\right)=0 \label{evol1}
%\ee
%with $v^\rho$ any timelike vector, that we can chose to be unit ($v^\rho v_\rho =-1$) without loss of generality. An alternative, but equivalent, form of the evolution equations (adapted to (\ref{delta1T})) is
\be
\left(\nabla_\rho {\cal S}_{\gamma\mu[\nu}{}^\rho- {\cal J}_{\gamma\mu[\nu}({\cal S})\right)v_{\tau]}=0 \label{evol2}
\ee
with $v_\beta$ any timelike vector. The linear symmetric hyperbolic set (\ref{evol2})
constitutes  the {\em evolution} equations of our system. 
Note that, taking into account trace and symmetry properties, there are precisely 5 complex (10 real) independent equations in (\ref{evol2}), 
which is the number of independent unknowns.

The complete system (\ref{delta1T}) is re-obtained by adding the constraints, which can be written for any given spacelike hypersurface $\Sigma$ with timelike normal 
$n^\mu$ as (cf.\ \cite{S}, section 4)
%\be
%n_{[\sigma}\nabla_{\alpha}{\cal T}_{\beta\gamma]\mu\nu}-n_{[\sigma}J_{\alpha\beta\gamma]\mu\nu}({\cal T}) =0, \label{const1}
%\ee
%or equivalently as
\be
n^\nu\left(\nabla_\rho {\cal S}_{\gamma\mu\nu}{}^\rho- {\cal J}_{\gamma\mu\nu}({\cal S})\right)=0\;. \label{const2}
\ee
Notice, first of all, that only derivatives {\em tangent} to $\Sigma$ appear in 
%(\ref{evol1}) and 
(\ref{const2}). Observe furthermore that (\ref{const2}) contains 3 complex (6 real) independent equations  
which adds up with (\ref{evol2}) to the number of equations of the original system (\ref{delta1T}),
 rendering
%(\ref{evol1}) together with (\ref{const1}) is fully equivalent to the original system (\ref{d1T}). Similarly, 
the former two equations fully equivalent with the latter one.
% Let us check, for instance, the latter. 
To check this directly, contract (\ref{evol2}) with $v^\tau$ to get
%$\jose{I dislike this notation for $|v|^{2}$, and it confused me in a previous version. Cf before...}
$$
(-v^2 \delta^\sigma_\nu +v^\sigma v_\nu)\left(\nabla_\rho {\cal S}_{\gamma\mu\sigma}{}^\rho- {\cal J}_{\gamma\mu\sigma}({\cal S})\right)=0
\;,
$$
where 
\be
h^\sigma{}_\gamma := \delta^\sigma_\gamma -v^{-2}v_\gamma v^\sigma
\label{proj}
\ee
is the projector orthogonal to $v^\mu$, which immediately leads to (\ref{delta1T}) by taking into account (\ref{const2})--- e.g., by simply choosing  $v^\mu$ pointing along $n^\mu$.

As mentioned before, (\ref{evol2}) contains 5 equations for the 5 complex  independent unknowns in ${\cal S}_{\gamma\mu\sigma}{}^\rho$. A convenient way of explicitly expressing this fact is by recalling the following identity
\bean
{\cal S}_{\alpha\beta\lambda\mu}=2\left[(h_{\alpha[\lambda} -v^{-2}v_\alpha v_{[\lambda}){\cal E}_{\mu]\beta} +(h_{\beta[\mu} -v^{-2}v_\beta v_{[\mu}){\cal E}_{\lambda]\alpha}\right.\\
\left. -i v^{{-2}}v^{\rho}\volform_{\rho\lambda\mu\sigma}v_{[\alpha}{\cal E}_{\beta]}{}^{\sigma}-i v^{{-2}}v^{\rho}\volform_{\rho\alpha\beta\sigma}v_{[\lambda}{\cal E}_{\mu]}{}^{\sigma}\right]
\eean
in terms of the spatial ``electric-magnetic'' tensor defined for any timelike $v^{\tau}$ by
%\jose{${\cal E}$ depends on $v$ and ${\cal S}$, but I have not introduced any notation for that because it will only arise once for $\tilde{\cal T}$. OK?}
\be
{\cal E}_{\beta\mu}:= -v^{-2} v^{\alpha} v^{\lambda} {\cal S}_{\alpha\beta\lambda\mu} . \label{e-m}
\ee
Observe the following properties
$$
{\cal E}_{\beta\mu}={\cal E}_{\mu\beta}, \hspace{3mm} {\cal E}_{\beta\mu}v^{\mu} =0,
\hspace{3mm} {\cal E}^{\mu}{}_{\mu}=0.
$$
Thus, ${\cal E}_{\beta\mu}$ contains 5 complex independent components and exactly the same information as the full ${\cal S}_{\gamma\mu\sigma}{}^\rho$. Note that the density (\ref{s-e}) is then expressed simply as
\be
W_{v}({\cal S}) = {\cal E}_{\mu\nu}\overline{\cal E}^{\mu\nu} .\label{s-e2}
\ee
In any orthonormal basis with its timelike `t'-part aligned with $v^{\mu}$, the five independent components of ${\cal E}_{\mu\nu}$ are given simply by
$$
{\cal E}_{ij} = {\cal S}_{titj}, \hspace{1cm} {\cal S}_{tijl}=i \volform^{t}{}_{kjl}{\cal E}_{i}^{k}
$$
where the second equation follows from the self-duality of ${\cal S}$. Using this, the evolution equations (\ref{evol2}) become simply
\be
\nabla_{\rho}{\cal S}_{t(ij)}{}^{\rho} =\nabla_{t} {\cal E}_{ij}+i \volform^{t}{}_{lk(j}\nabla^{l}{\cal E}_{i)}^{k}={\cal J}({\cal S})_{t(ij)}\label{evol3}
\ee
while the constraint equations (\ref{const2}) (with $n^{\tau}$ pointing along $v^{\tau}$) read
\be
\nabla_{\rho}{\cal S}_{tit}{}^{\rho}=\nabla_{j}{\cal E}_{i}{}^{j}={\cal J}({\cal S})_{tit}\, .
\label{const3}
\ee
We will use these expressions later for the case  ${\cal S}_{\alpha\beta\gamma\delta} = 
\widetilde{\mathcal{T}}^{(\mathrm{ev})}_{\alpha\beta\gamma\delta}$, 
to prove uniqueness of the solutions to \eq{unphys_ev}.

All in all, as a generalization of \cite[Theorem~4.5 \& 4.7]{ik} we have obtained
\begin{lemma}
\label{sym_hyp_MST}
\begin{enumerate}
\item[(i)] The MST $\mathcal{S}^{(\mathrm{ev})}_{\mu\nu\sigma\rho}$ satisfies, for any sign of the cosmological constant $\Lambda$, a  linear, homogeneous symmetric hyperbolic system
of evolution equations in $(\mcM,g)$.
\item[(ii)] The rescaled 
MST $\widetilde{\mathcal{T}}^{(\mathrm{ev})}_{\mu\nu\sigma\rho}$ satisfies, for any sign of the cosmological constant $\Lambda$, a  linear, homogeneous symmetric hyperbolic system
of evolution equations  in $(\widetilde{\mcM\enspace}\hspace{-0.5em} ,\widetilde g)$.
\end{enumerate}
\end{lemma}

\begin{remark}
{\rm
An alternative route to arrive at the same result is by using spinors, see \cite{F1,F}. In this formalism \cite{penrose} the (rescaled) MST is represented by a fully symmetric spinor $\Upsilon_{ABCD}$, and equations (\ref{delta1T}) are written in the following form
$$
\nabla^{A}{}_{A'}\Upsilon_{ABCD}=L_{A'BCD}
$$
where $L_{A'BCD}=L_{A'(BCD)}$ is the spinor associated to ${\cal J}_{\gamma\mu\nu}({\cal S})$. Then this is easily put in symmetric hyperbolic form, writing it as in \cite{F1}, section 4, for the Bianchi equations.
}
\end{remark}

\begin{remark}
\label{rem_reg}
{\rm Note that the denominator $Q_{\mathrm{ev}}\mathcal{F}^2 + 8\Lambda$ in the equation \eq{full_eqn_MST} for $\mathcal{S}^{(\mathrm{ev})}_{\mu\nu\sigma}{}^{\rho}$
might have zeros. Furthermore, the Ernst potential may have zeros
so that $Q_{\mathrm{ev}}$  blows up.
An analogous problem arises  for $\widetilde{\mathcal{T}}^{(\mathrm{ev})}_{\mu\nu\sigma}{}^{\rho}$, cf.\  \eq{sym_hyp0} below.
In fact, it follows from \eq{expansion_sigma}, \eq{asympt_exp_Q}, \eq{asympt_exp_F2}
and \eq{asympt_exp_H2} below, that this cannot happen sufficiently close
to $\scri$ (for $\Lambda >0$).
It is  not clear, though, whether the evolution equations remain  regular off some neighborhood of $\scri$.

Moreover, it will be shown in the subsequent section  that 
%\jose{rewritten, and formula number changed} 
$\mathcal{J}(\widetilde{\mathcal{T}}^{(\mathrm{ev})})_{\alpha\beta\mu}$ is singular at $\scri$ due to  the vanishing of the
conformal factor $\Theta$ there --see \eq{sym_hyp1} below--, whence one actually has to deal with a Fuchsian
system.}

\end{remark}

%A similar result holds for the rescaled MST  $\widetilde{\mathcal{T}}^{(\mathrm{ev})}_{\mu\nu\sigma}{}^{\rho}$  in the unphysical spacetime, but 
%we want to make sure that equations \eq{unphys_ev} are regular near $\scri$. 

\subsubsection{Behavior of the Bianchi-like system for $\widetilde{\mathcal{T}}^{(\mathrm{ev})}_{\mu\nu\sigma}{}^{\rho}$  near $\scri$}

Let us analyze the behavior of the system \eq{unphys_ev} near $\scri$. 
Note that we are
{\it not} assuming  a priori that 
$\widetilde{\mathcal{T}}^{(\mathrm{ev})}_{\mu\nu\sigma}{}^{\rho}$  is regular
at $\scri$.
%We have to check that the equations \eq{unphys_ev} are regular near $\scri$.
%\tim{a couple of  additions and reorganizations until the end of the section}
First of all we employ the following expansions which have been derived in Section~\ref{Mars-Simon_conf}, 
and which do not rely on any gauge choice,
\begin{eqnarray}
Q_{\mathrm{ev}} &=& O(\Theta^4)
\label{asympt_exp_Q}
\;,
\\
\mathcal{F}_{\mu\nu} &=& \widetilde{\mathcal{H}}_{\mu\nu}\Theta^{-3} + O(\Theta^{-2})
\;,
\\
\mathcal{F}^2 &=& \Theta^{-2}\widetilde{\mathcal{H}}^2 + O(\Theta^{-1})
\label{asympt_exp_F2}
\;,
\\
g^{\mu\nu} &=& \Theta^2 \widetilde g^{\mu\nu}
\;,
\\
\mathcal{I}_{\alpha\beta\mu\nu} &=&\Theta^{-4} \widetilde {\mathcal{I}}_{\alpha\beta\mu\nu}
\;,
\\
\mathcal{U}_{\alpha\beta\mu}{}^{\nu} &=&- \Big(\widetilde{\mathcal{H}}_{\alpha\beta}\widetilde{\mathcal{H}}_{\mu}{}^{\nu} - 
\frac{1}{3}\widetilde{\mathcal{H}}^2\widetilde{\mathcal{I}}_{\alpha\beta\mu}{}^{\nu}\Big)\Theta^{-4} + O(\Theta^{-3})
\;.
\end{eqnarray}
This yields
\begin{eqnarray}
  \mathcal{J}(\widetilde{\mathcal{T}}^{(\mathrm{ev})})_{\alpha\beta\mu}
&=&
-4 \Lambda  \frac{  5  Q_{\mathrm{ev}}\mathcal{F}^2  +4\Lambda }{Q_{\mathrm{ev}}\mathcal{F}^2 + 8\Lambda}
 \mathcal{U}_{\alpha\beta\mu}{}^{\nu} 
   \mathcal{F}^{-4} X^{\sigma}g_{\rho\nu}g^{\gamma\kappa}g^{\delta\varkappa} \mathcal{F}_{\kappa\varkappa}  \widetilde{\mathcal{T}}^{(\mathrm{ev})}_{\gamma\delta\sigma}{}^{\rho}  
\nonumber
\\
&& +  Q_{\mathrm{ev}}X^{\sigma}\Big( \frac{2}{3}
 \mathcal{I}_{\alpha\beta\mu\rho} g^{\gamma\kappa} g^{\delta\varkappa} \mathcal{F}_{\kappa\varkappa}  \widetilde{\mathcal{T}}^{(\mathrm{ev})}_{\gamma\delta\sigma}{}^{\rho}  
- \mathcal{F}_{\mu\rho} \widetilde{\mathcal{T}}^{(\mathrm{ev})}_{\alpha\beta\sigma}{}^{\rho}
\Big)
\label{sym_hyp0}
\\
&=&
-2 \Lambda
 \mathcal{U}_{\alpha\beta\mu}{}^{\nu} 
   \mathcal{F}^{-4} X^{\sigma}g_{\rho\nu}g^{\gamma\kappa}g^{\delta\varkappa} \mathcal{F}_{\kappa\varkappa}  \widetilde{\mathcal{T}}^{(\mathrm{ev})}_{\gamma\delta\sigma}{}^{\rho}  
\nonumber
\\
&& 
 + (O(\Theta) \widetilde{\mathcal{T}}^{(\mathrm{ev})})_{\alpha\beta\mu}
\\
&=&
2 \Lambda\Theta^{-1}
\Big( \widetilde{\mathcal{H}}_{\alpha\beta}\widetilde{\mathcal{H}}_{\mu\rho} - 
\frac{1}{3}\widetilde{\mathcal{H}}^2\widetilde{\mathcal{I}}_{\alpha\beta\mu\rho}\Big)
  \widetilde{\mathcal{H}}^{-4}  \widetilde{\mathcal{H}}^{\gamma\delta}  \widetilde X^{\sigma} \widetilde{\mathcal{T}}^{(\mathrm{ev})}_{\gamma\delta\sigma}{}^{\rho}  
\nonumber
\\
&& 
 + (O(1) \widetilde{\mathcal{T}}^{(\mathrm{ev})})_{\alpha\beta\mu}
\;.
\label{sym_hyp1}
\end{eqnarray}

In adapted 
%\tim{add ``(sufficiently smooth)''?}
coordinates $(t,x^i)$ where $\scri^-=\{t=0\}$, we have
\begin{equation}
%X^{\mu}\,=\, \widetilde X^{\mu}\,=\, O(1)
X^{i}\,=\, \widetilde X^{i}\,=\, Y^i + O(\Theta)
\;,
\quad
X^{t}\,=\, \widetilde X^{t}\,=\, O(\Theta)
\;.
\label{KVF_scri}
\end{equation}
Let us further assume a gauge where
\begin{equation}
 \widetilde g_{tt}|_{\scri}  \,=\, -1\;, \quad \widetilde g_{ti}|_{\scri} \,=\, 0
\;.
\label{asymp_gauge}
\end{equation}
 In particular this  implies by \eq{conf5} that the conformal factor $\Theta$ satisfies
\begin{equation}
 \Theta \,=\, \sqrt{\frac{\Lambda}{3}}t + O(1)
\;.
\label{exp_theta}
\end{equation}
Moreover, as for the wave map gauge \eq{gauge_conditions_compact}, which is compatible with \eq{asymp_gauge}, we find
\begin{eqnarray}
\widetilde{\mathcal{H}}_{ti} &=& - \sqrt{\frac{\Lambda}{3}} Y_i + O(\Theta)
\;,
\label{H_asymp_gauge1}
\\
\widetilde{\mathcal{H}}_{ij} &=& i\sqrt{\frac{\Lambda}{3}} \widehat\eta_{ijk}Y^k  + O(\Theta)
\;,
\label{H_asymp_gauge2}
\\
\widetilde{\mathcal{H}}^2 &=& - 4 \frac{\Lambda}{3} |Y|^2   + O(\Theta)
\label{asympt_exp_H2}
\label{H_asymp_gauge3}
%\;,
%\;,
%\\
%Q_{\mathrm{ev}} &=& O(\Theta^4)
%\label{asympt_exp_Q}
%\;,
%\\
%\mathcal{F}_{\mu\nu} &=& \widetilde{\mathcal{H}}_{\mu\nu}\Theta^{-3} + O(\Theta^{-2})
%\;,
%\\
%\mathcal{F}^2 &=& -\frac{4}{3}\Lambda |Y|^2\Theta^{-2} + O(\Theta^{-1})
%\label{asympt_exp_F2}
%\;,
%\\
%\sigma &=& -|Y|^2 \Theta^{-2} + O(\Theta^{-1})
%\;,
%\\
%g^{\mu\nu} &=& \Theta^2 \widetilde g^{\mu\nu}
%\;,
%\\
%\mathcal{I}_{\alpha\beta\mu\nu} &=&\Theta^{-4} \widetilde {\mathcal{I}}_{\alpha\beta\mu\nu}
%\;,
%\\
%\mathcal{U}_{\alpha\beta\mu}{}^{\nu} &=&- (\widetilde{\mathcal{H}}_{\alpha\beta}\widetilde{\mathcal{H}}_{\mu}{}^{\nu} - 
%\frac{1}{3}\widetilde{\mathcal{H}}^2\widetilde{\mathcal{I}}_{\alpha\beta\mu}{}^{\nu})\Theta^{-4} + O(\Theta^{-3})
\;.
\end{eqnarray}
%
%The asymptotic expansions in this gauge (which we have specified only asymptotically via \eq{asymp_gauge})  therefore coincide with those computed in Section  in the wave map gauge \eq{gauge_conditions_compact}. This  is due to the fact that the two gauge conditions coincide asymptotically.
%
%That yields
%%
%\begin{eqnarray}
%  \mathcal{J}(\widetilde{\mathcal{T}}^{(\mathrm{ev})})_{\alpha\beta\mu}
%&=&
% \frac{9}{8}\Lambda^{-1} \Theta^{-1}|Y|^{-4}(\widetilde{\mathcal{H}}_{\alpha\beta}\widetilde{\mathcal{H}}_{\mu\rho} - \frac{1}{3}\widetilde{\mathcal{H}}^2\widetilde{\mathcal{I}}_{\alpha\beta\mu\rho}) \widetilde X^{\sigma}\widetilde{\mathcal{H}}^{\gamma\delta}
%\widetilde{\mathcal{T}}^{(\mathrm{ev})}_{\gamma\delta\sigma}{}^{\rho} 
%\nonumber
%\\
%&& 
% + (O(1) \widetilde{\mathcal{T}}^{(\mathrm{ev})})_{\alpha\beta\mu}
%%\label{sym_hyp1}
%\;.
%\end{eqnarray}
%
Using further that
\begin{equation}
%&\widetilde X^{\sigma} \,=\, \delta_k{}^{\sigma} Y^k + O(\Theta)
%\;,&
%\\
%&
\widetilde{\mathcal{H}}^{\gamma\delta}
\widetilde{\mathcal{T}}^{(\mathrm{ev})}_{\gamma\delta\sigma}{}^{\rho}  \,=\,  2 \widetilde{H}^{\gamma\delta}
\widetilde{\mathcal{T}}^{(\mathrm{ev})}_{\gamma\delta\sigma}{}^{\rho}  \,=\,  4 \widetilde X^{\gamma}\widetilde\nabla^{\delta}\Theta
\widetilde{\mathcal{T}}^{(\mathrm{ev})}_{\gamma\delta\sigma}{}^{\rho}  
 \,=\,  4  \sqrt{\frac{\Lambda}{3}} Y^{i}
\widetilde{\mathcal{T}}^{(\mathrm{ev})}_{t  i \sigma}{}^{\rho}   + (O(\Theta) \widetilde{\mathcal{T}}^{(\mathrm{ev})})_{\sigma}{}^{\rho}
\;, 
\label{HT_relation}
%&
\end{equation}
we 
%obtain the following system:
find that the system \eq{unphys_ev} has the following structure near $\scri$,
\begin{align}
\widetilde\nabla_{\rho} \widetilde{\mathcal{T}}^{(\mathrm{ev})}_{\alpha\beta\mu}{}^{\rho}
%= & \,  
%  \mathcal{J}(\widetilde{\mathcal{T}}^{(\mathrm{ev})})_{\alpha\beta\mu} 
%\nonumber \\
\,=\, & \, 
\frac{9}{2} \sqrt{\frac{\Lambda}{3}}\Lambda^{-1} \Theta^{-1}|Y|^{-4}(\widetilde{\mathcal{H}}_{\alpha\beta}\widetilde{\mathcal{H}}_{\mu\rho}
- \frac{1}{3}\widetilde{\mathcal{H}}^2\widetilde{\mathcal{I}}_{\alpha\beta\mu\rho})Y^i Y^k \widetilde{\mathcal{T}}^{(\mathrm{ev})}_{ti  k}{}^{\rho}
\nonumber
\\
&
 + (O(1) \widetilde{\mathcal{T}}^{(\mathrm{ev})})_{\alpha\beta\mu}
\label{sym_hyp}
\;,
\end{align}
in adapted coordinates and whenever  \eq{asymp_gauge} holds.

As explained in the previous section the system  \eq{unphys_ev} splits into a symmetric hyperbolic system of evolution equations and a system of constraint
equations  for  $\widetilde{\mathcal{T}}^{(\mathrm{ev})}_{\alpha\beta\mu\nu}$. However, this requires an appropriate gauge choice. A  convenient way to realize such a gauge 
is to impose the condition \eq{asymp_gauge} also off the initial surface,
\begin{equation}
 \widetilde g_{tt} \,=\, -1\;, \quad \widetilde g_{ti}\,=\, 0
\;.
\label{new_gauge}
\end{equation}
It is well known that  these \emph{Gaussian normal coordinates}  \cite{Wald}
%which yield \eq{new_gauge}
 are obtained
by shooting geodesics normally to $\scri$; the coordinate $t$ is then chosen to be an affine parameter along these geodesics, while the coordinates
$\{x^i\}$ are transported  from $\scri$ by requiring them to be constant along these geodesics.
%
%In addition to \eq{new_gauge} we assume that the conformal factor has been chosen in such a way that
%%
%\begin{equation}
%\widetilde s|_{\scri} \,=\, 0
%\;,
%\label{gauge_conf_factor}
%\end{equation}
%%
%which then implies via the MCFE \eq{conf2} 
%%
%\begin{equation}
% \partial_0\widetilde g_{ij}|_{\scri} \,=\,0 \;, \quad \partial_0\partial_0 \Theta|_{\scri}\,=\, 0
%\;.
%\end{equation}
%%
%In the gauge \eq{new_gauge} \& \eq{gauge_conf_factor} we find the following expansions,
%
%\jose{Some changes in what follows, to join smoothly with other parts} 
%As explained in the previous subsection, this system contains constraint and evolution equations for  $\widetilde{\mathcal{T}}^{(\mathrm{ev})}_{\alpha\beta\mu\nu}$.

Setting ${\cal
S}_{\alpha\beta\gamma\delta}=\widetilde{\mathcal{T}}^{(\mathrm{ev})}_{\alpha\beta\gamma\delta}$ in  (\ref{evol3})
and (\ref{const3})  and using
%the self-duality of $\widetilde{\mathcal{H}}$ and $\widetilde{\mathcal{T}}$ as well as 
\eq{H_Y_relation}, which follows from \eq{H_asymp_gauge1}-\eq{H_asymp_gauge3},
%the formula
%%
%\begin{equation}
%(\widetilde{\mathcal{H}}_{0i}\widetilde{\mathcal{H}}_{0j} -
%\frac{1}{3}\widetilde{\mathcal{H}}^2\widetilde{\mathcal{I}}_{0i0j})|_{\scri}
%\,=\, \frac{\Lambda}{3}(Y_iY_j)_{\mathrm{tf}}
%\;,
%\end{equation}
%
we find in the unphysical spacetime (we have $\mathcal{E}_{ij}\equiv  \widetilde{\mathcal{T}}^{(\mathrm{ev})}_{ti tj}$)
%\jose{new lines, some $\widehat\epsilon$ changed to $\widehat\volform$, ok? \\ -- \\ tim: should be ok}
%
\begin{eqnarray}
\widetilde\nabla_{\rho} \widetilde{\mathcal{T}}^{(\mathrm{ev})}_{tit}{}^{\rho}
&\equiv& 
\widetilde\nabla_{j}{\cal E}_{i}{}^{j}
\\
&=&
\frac{9}{2} \sqrt{\frac{\Lambda}{3}}\Lambda^{-1}
\Theta^{-1}|Y|^{-4}(\widetilde{\mathcal{H}}_{ti}\widetilde{\mathcal{H}}_{tj} -
\frac{1}{3}\widetilde{\mathcal{H}}^2\widetilde{\mathcal{I}}_{titj})Y^l Y^k
\widetilde{\mathcal{T}}^{(\mathrm{ev})}_{tlk}{}^{j}
\nonumber
\\
&&
 + (O(1) \widetilde{\mathcal{T}}^{(\mathrm{ev})})_{tit}
%\\
%&=&
%\frac{3}{2} \sqrt{\frac{\Lambda}{3}}
%\Theta^{-1}|Y|^{-4}(Y_iY_j)_{\mathrm{tf}}\,Y^l Y^k
%\widetilde{\mathcal{T}}^{(\mathrm{ev})}_{0lk}{}^{j}
%\nonumber
%\\
%&&
% + (O(1) \widetilde{\mathcal{T}}^{(\mathrm{ev})})_{0i0}
%\label{sym_hyp_spec}
\\
&=&
\frac{1}{2} \sqrt{\frac{\Lambda}{3}}
\Theta^{-1}|Y|^{-2} Y^l Y^k
\widetilde{\mathcal{T}}^{(\mathrm{ev})}_{tlki}
 + (O(1) \widetilde{\mathcal{T}}^{(\mathrm{ev})})_{tit}
\\
&=&
-\frac{1}{2}i \sqrt{\frac{\Lambda}{3}}
\Theta^{-1}|Y|^{-2} \widehat\volform_{i}{}^{jk}Y^l Y_{[j}
\widetilde{\mathcal{T}}^{(\mathrm{ev})}_{k]tlt}
 + (O(1) \widetilde{\mathcal{T}}^{(\mathrm{ev})})_{tit}
\label{sym_hyp_spec}
\\
&\equiv&
-\frac{1}{2}i \sqrt{\frac{\Lambda}{3}}
\Theta^{-1}|Y|^{-2} \widehat\volform_{i}{}^{jk}Y^l Y_{[j}
{\cal E}_{k]l}
 + (O(1) \mathcal{E})_{tit}
\end{eqnarray}
for the constraint equations,  and
\begin{eqnarray}
\widetilde\nabla_{\rho} \widetilde{\mathcal{T}}^{(\mathrm{ev})}_{t(ij)}{}^{\rho}
&\equiv & \widetilde\nabla_{t} {\cal E}_{ij}+i \widehat\volform_{(j}{}^{lk}\widetilde\nabla_{l}{\cal E}_{i)k} 
\\
&=&
-\frac{9}{2} \sqrt{\frac{\Lambda}{3}}\Lambda^{-1} \Theta^{-1}|Y|^{-4}(\widetilde{\mathcal{H}}_{t(i}\widetilde{\mathcal{H}}_{|t|j)}
- \frac{1}{3}\widetilde{\mathcal{H}}^2\widetilde{\mathcal{I}}_{t(i|t|j)})Y^k Y^l \widetilde{\mathcal{T}}^{(\mathrm{ev})}_{tk t l}
\nonumber
\\
&&+
\frac{9}{2} \sqrt{\frac{\Lambda}{3}}\Lambda^{-1} \Theta^{-1}|Y|^{-4}(\widetilde{\mathcal{H}}_{t(i}\widetilde{\mathcal{H}}_{j)k}
- \frac{1}{3}\widetilde{\mathcal{H}}^2\widetilde{\mathcal{I}}_{t(ij)k})Y^l Y^m \widetilde{\mathcal{T}}^{(\mathrm{ev})}_{tlm}{}^{k}
\nonumber
\\
&&
 + (O(1) \widetilde{\mathcal{T}}^{(\mathrm{ev})})_{(tij)}
\\
&=& \sqrt{\frac{\Lambda}{3}}|Y|^{-2} \Theta^{-1}\Big(
\frac{3}{2}Y_{(i} Y^l \widetilde{\mathcal{T}}^{(\mathrm{ev})}_{|t|j)tl}
- 3|Y|^{-2}Y_iY_jY^k Y^l \widetilde{\mathcal{T}}^{(\mathrm{ev})}_{tk t l}
\nonumber
\\
&&
+ \frac{1}{2} h_{ij} Y^k Y^l \widetilde{\mathcal{T}}^{(\mathrm{ev})}_{tk t l}
\Big)
 + (O(1) \widetilde{\mathcal{T}}^{(\mathrm{ev})})_{t(ij)}
\label{ev_eqns}
\\
&\equiv& \sqrt{\frac{\Lambda}{3}}|Y|^{-2} \Theta^{-1}\Big(
\frac{3}{2}Y_{(i} Y^l \mathcal{E}_{j)l}
- 3|Y|^{-2}Y_iY_jY^k Y^l \mathcal{E}_{k l}
+ \frac{1}{2} h_{ij} Y^k Y^l \mathcal{E}_{k  l}
\Big)
\nonumber
\\
&&
 + (O(1) \mathcal{E})_{t(ij)}
\label{ev_eqns2}
\end{eqnarray}
for the evolution equations.
%\tim{needs a thought: The gauge condition we are using implies $\tilde g{}^{0i}=0$ only on $\scri$ but in general not off $\scri$, 
%so these are  constraint equation merely ``on $\scri$''}

Note that the equations \eq{sym_hyp_spec} and \eq{ev_eqns} hold regardless of the gauge  as long as 
the asymptotic gauge condition  \eq{asymp_gauge} is ensured. However,  the global gauge condition \eq{new_gauge}
(or an analogous one, cf.\  Section~\ref{subsec:FOSH})
is needed to ensure that this realizes   the splitting into constraint and evolution equations.

The divergent terms in both constraint and evolution equations 
%\jose{changed} 
are regular
%(prerequisite to end up with a smooth solution)
if and only if
\begin{align*}
 Y^l Y_{[j}
\widetilde{\mathcal{T}}^{(\mathrm{ev})}_{k]tlt}|_{\scri} =0=Y^l Y_{(j}
\widetilde{\mathcal{T}}^{(\mathrm{ev})}_{k)tlt}|_{\scri}  & \quad\Longleftrightarrow \quad
Y^j\widetilde{\mathcal{T}}^{(\mathrm{ev})}_{titj}|_{\scri} \equiv Y^j \mathcal{E}_{ij}|_{\scri}= 0 \\
%& \quad\Longleftrightarrow\quad
%Y^j \mathcal{E}_{ij}|_{\scri} = 0 \\
& \quad\Longleftrightarrow\quad
%(\widetilde X^{\sigma}\widetilde{\mathcal{T}}_{\mu\nu 0\sigma})|_{\scri} =0
(\widetilde {\mathcal{H}}^{\alpha\beta}\widetilde{\mathcal{T}}^{(\mathrm{ev})}_{\mu\nu\alpha\beta})|_{\scri} =0
\;.
\end{align*}
For the sake of consistency, we check that these conditions hold
if and only if the spacetime is asymptotically KdS-like. Indeed
%
%\begin{equation*}
%\widetilde{\mathcal{T}}^{(\mathrm{ev})}_{0i0j}|_{\scri}\,=\,\widetilde{\mathcal%{T}}^{(0)}_{0i0j}|_{\scri}\,=\, D_{ij} -\frac{3}{2}   |Y|^{-4} Y^{k}Y^{l} D_{kl%}(Y_iY_j)_{\mathrm{tf}}
% - i  \  \sqrt{\frac{3}{\Lambda}} \Big(   \widehat C_{ij}
%-     \frac{3}{2}   |Y|^{-4} Y^{k}Y^{l} \widehat C_{kl}(Y_iY_j)_{\mathrm{tf}}\B%ig)
%\;.
%\end{equation*}
%%
%The problematic  term in the constraint equations will be regular if and only i%f
%%
\begin{eqnarray}
0 \,=\, Y^l Y_{[j}
\widetilde{\mathcal{T}}^{(\mathrm{ev})}_{k]tlt}|_{\scri} \,=\,  Y^l Y_{[j}D_{k]l}
 - i   \sqrt{\frac{3}{\Lambda}}  Y^l Y_{[j}  \widehat C_{k]l}
\;,
\label{asmpt_KdS}
\end{eqnarray}
holds if and only if  $Y^{k}$ is an eigenvector of both $D_{jk}$ and $\widehat{C}_{jk}$, or in other words if and only if
\begin{eqnarray}
 Y^k D_{jk}  =  |Y|^{-2} Y^k Y^l D_{kl}   Y_{j}\;,
\quad
 Y^k \widehat C_{jk}  =  |Y|^{-2} Y^k Y^l \widehat C_{kl}   Y_{j}
\;,
\label{asmpt_KdS2}
\end{eqnarray}
which are precisely the conditions defining 
asymptotically KdS-like spacetimes in Definition \ref{asympt_KdS_like}. 
Analogously, the divergent term in the evolution equations will be regular if and only if
\begin{eqnarray}
0 \,=\, Y^l Y_{(j}
\widetilde{\mathcal{T}}^{(\mathrm{ev})}_{k)tlt}|_{\scri} &=&
Y^l Y_{(j} \mathcal{E}_{k)l}|_{\scri} \nonumber
\\
&=&
  Y^l Y_{(j}D_{k)l} -  |Y|^{-2} Y^{m}Y^{n} D_{mn} Y_jY_k
\nonumber
\\
& - i &  \sqrt{\frac{3}{\Lambda}} \Big(    Y^l Y_{(j}\widehat C_{k)l}
-       |Y|^{-2} Y^{m}Y^{n}  \widehat C_{mn}Y_jY_k\Big)
\;,
 \phantom{xxx}
\label{cond_ev}
\end{eqnarray}
which  is automatically true in the asymptotically KdS-like setting as follows from \eq{asmpt_KdS2}.

In summary, the evolution equations (\ref{evol3}) for 
${\cal S}_{\mu\nu\sigma}{}^{\rho}=\widetilde{\mathcal{T}}^{(\mathrm{ev})}_{\mu\nu\sigma}{}^{\rho}$ constitute a symmetric hyperbolic system in the unphysical spacetime with a righthand side of the form
\begin{equation}
 \mathcal{J}(\widetilde{\mathcal{T}}^{(\mathrm{ev})})_{\alpha\beta\mu}
= \frac{1}{\Theta} 
 \mathcal{N}(\widetilde{\mathcal{T}}^{(\mathrm{ev})})_{\alpha\beta\mu}
\end{equation}
where ${\mathcal N}(\widetilde{\mathcal{T}}^{(\mathrm{ev})})$ denotes a linear map
${\mathcal N}(\widetilde{\mathcal{T}}^{(\mathrm{ev})})_{\alpha\beta\mu} = 
{\mathcal  N}_{\alpha\beta\mu}{}^{\rho \nu \sigma}{}_{\kappa} 
\widetilde{\mathcal{T}}^{(\mathrm{ev})}_{\rho\nu\sigma}{}^{\kappa}$ being
${\mathcal  N}_{\alpha\beta\mu}{}^{\rho \nu \sigma}{}_{\kappa}$ a smooth tensor field up to and including $\scri$, at least in some neighborhood of $\scri$, cf.\ Remark~\ref{rem_non_reg}.
 Equations with such divergent
terms are called {\it Fuchsian} in the literature. We state the existence
and properties of the evolution equation for
$\widetilde{\mathcal{T}}^{(\mathrm{ev})}_{\rho\sigma\mu}{}^{\nu}$ 
as a lemma.
\begin{lemma}
\label{lemma_evolution}
 The rescaled MST 
$\widetilde{\mathcal{T}}^{(\mathrm{ev})}_{\rho\sigma\mu}{}^{\nu}$ 
with $Q=Q_{\mathrm{ev}}$, satisfies a symmetric hyperbolic, linear, homogeneous
Fuchsian system of evolution equations near $\scri$.
\end{lemma}

\begin{remark}
\label{rem_non_reg}
{\rm
As already discussed in Remark~\ref{rem_reg},
%Note that the denominator $Q_{\mathrm{ev}}\mathcal{F}^2 + 8\Lambda$ in the equation \eq{sym_hyp0} for $\widetilde{\mathcal{T}}^{(\mathrm{ev})}_{\mu\nu%\sigma}{}^{\rho}$ 
%might have zeros off some neighborhood of $\scri^-$.
%%\tim{addition}
%Similarly, the Ernst potential may become zero off some neighborhood of $\scri^-$ so that $Q_{\mathrm{ev}}$  blows up.
%It is therefore
it is  not clear that the evolution system remains  regular outside some neighborhood of 
$\scri$, in its whole domain of dependence.
%
%If $g$ is the Kerr-de Sitter metric, in Boyer-Lindquist-type coordinates say,
%\tim{with $\xi=0$}
%and $X=\partial_t$, then
}
\end{remark}

\subsubsection{A wave equation satisfied by the (rescaled) MST}

We now recall that \eq{d1T} holds in particular for    
$\widetilde{\mathcal{T}}^{(\mathrm{ev})}_{\alpha\beta\lambda}{}^{\mu}$,
and with tildes on all quantities. Application of $\widetilde\nabla^{\sigma}$
yields, together with with \eq{unphys_ev}, the linear, homogeneous wave equation 

\begin{eqnarray}
\Box_{\widetilde g} \widetilde{\mathcal{T}}^{(\mathrm{ev})}_{\alpha\beta\mu\nu}
 &=& -  2\widetilde\nabla^{\sigma}\widetilde\nabla_{[\mu}\widetilde{\mathcal{T}}^{(\mathrm{ev})}_{\nu]\sigma\alpha\beta} 
  - i\widetilde\volform_{\mu\nu\sigma}{}^{\rho}\widetilde\nabla^{\sigma}\mathcal{J}(\widetilde{\mathcal{T}}^{(\mathrm{ev})})_{\alpha\beta\rho}
%\\
%&=&
%-  2\widetilde\nabla_{[\mu}\widetilde\nabla_{|\sigma}\widetilde{\mathcal{T}}^{(\mathrm{ev})}_{\alpha\beta|\nu]}{}^{\sigma} 
%   - i\widetilde\volform_{\mu\nu}{}^{\sigma\rho}\widetilde\nabla_{\sigma}\mathcal{J}(\widetilde{\mathcal{T}}^{(\mathrm{ev})})_{\alpha\beta\rho}
% -  2\widetilde R_{\kappa [\mu }\widetilde{\mathcal{T}}^{(\mathrm{ev})}_{|\alpha\beta|\nu]}{}^{\kappa} 
%\nonumber
%\\
%&&
% -  2\widetilde R_{\alpha\kappa[\mu }{}^{\sigma}\widetilde{\mathcal{T}}^{(\mathrm{ev})}_{\nu]\sigma\beta}{}^{\kappa} 
% +  2\widetilde R_{\beta\kappa[\mu }{}^{\sigma}\widetilde{\mathcal{T}}^{(\mathrm{ev})}_{\nu]\sigma\alpha}{}^{\kappa} 
% + \widetilde R_{\mu \nu\sigma}{}^{\kappa}\widetilde{\mathcal{T}}^{(\mathrm{ev})}_{\alpha\beta\kappa}{}^{\sigma} 
\\
&=&
-  2\widetilde\nabla_{[\mu} \mathcal{J}(\widetilde{\mathcal{T}}^{(\mathrm{ev})})_{|\alpha\beta|\nu]}
   - i\widetilde\volform_{\mu\nu}{}^{\sigma\rho}\widetilde\nabla_{\sigma}\mathcal{J}(\widetilde{\mathcal{T}}^{(\mathrm{ev})})_{\alpha\beta\rho}
 -  2\widetilde R_{\kappa [\mu }\widetilde{\mathcal{T}}^{(\mathrm{ev})}_{|\alpha\beta|\nu]}{}^{\kappa} 
\nonumber
\\
&&
 -  2\widetilde R_{\alpha\kappa[\mu }{}^{\sigma}\widetilde{\mathcal{T}}^{(\mathrm{ev})}_{\nu]\sigma\beta}{}^{\kappa} 
 +  2\widetilde R_{\beta\kappa[\mu }{}^{\sigma}\widetilde{\mathcal{T}}^{(\mathrm{ev})}_{\nu]\sigma\alpha}{}^{\kappa} 
 + \widetilde R_{\mu \nu\sigma}{}^{\kappa}\widetilde{\mathcal{T}}^{(\mathrm{ev})}_{\alpha\beta\kappa}{}^{\sigma} 
\;.
\label{wave_eqn_general}
\end{eqnarray}

Of course the same reasoning can be applied to
$S_{\alpha\beta\gamma\delta}^{(ev)}$, and we are led to the following 

\begin{lemma}
\label{wave_eqn_MST}
\begin{enumerate}
\item[(i)] The MST $\mathcal{S}^{(\mathrm{ev})}_{\mu\nu\sigma\rho}$ satisfies, for any sign of the cosmological constant $\Lambda$, a  linear, homogeneous system of  wave equations in $(\mcM,g)$.
\item[(ii)] The MST $\widetilde{\mathcal{T}}^{(\mathrm{ev})}_{\mu\nu\sigma\rho}$ satisfies, for any sign of the cosmological constant $\Lambda$, a   linear, homogeneous system of  
wave equations in $(\widetilde{\mcM\enspace}\hspace{-0.5em} ,\widetilde g)$.
\end{enumerate}
\end{lemma}

Some care is needed concerning the regularity of the coefficients in these wave equations, cf.\ Remark~\ref{rem_reg}.

It follows from  \eq{sym_hyp1} that  \eq{wave_eqn_general} is a  linear, homogeneous wave equation
\emph{of Fuchsian type} at $\scri$.
Indeed, using adapted coordinates $(x^0,x^i)$ and imposing the asymptotic gauge condition \eq{asymp_gauge}
a more careful calculation which uses 
 \eq{KVF_scri},  \eq{exp_theta}, \eq{H_asymp_gauge1}-\eq{H_asymp_gauge3} and \eq{HT_relation}
 shows 
 (note that 
%in the gauge \eq{new_gauge}
 $\mathcal{E}_{ij} \equiv \widetilde{\mathcal{T}}^{(\mathrm{ev})}_{titj}$ encompasses
all independent components of the rescaled MST),
\begin{eqnarray}
\Box_{\widetilde g} \mathcal{E}_{ij}
%&=& 
%-  4\Lambda \widetilde\nabla_{[\mu} \Big(
%\Theta^{-1}
%\Big( \widetilde{\mathcal{H}}_{|\alpha\beta|}\widetilde{\mathcal{H}}_{\nu]\kappa} - 
%\frac{1}{3}\widetilde{\mathcal{H}}^2\widetilde{\mathcal{I}}_{|\alpha\beta|\nu]\kappa}\Big)
%  \widetilde{\mathcal{H}}^{-4}  \widetilde{\mathcal{H}}^{\gamma\delta}  \widetilde X^{\sigma} \widetilde{\mathcal{T}}^{(\mathrm{ev})}_{\gamma\delta\sigma}{}^{\kappa}  \Big)
%\nonumber
%\\
%&&
%   - 2i \Lambda \widetilde\volform_{\mu\nu}{}^{\varkappa\rho}\widetilde\nabla_{\varkappa}
%\Big( \Theta^{-1}
%\Big( \widetilde{\mathcal{H}}_{\alpha\beta}\widetilde{\mathcal{H}}_{\rho\kappa} - 
%\frac{1}{3}\widetilde{\mathcal{H}}^2\widetilde{\mathcal{I}}_{\alpha\beta\rho\kappa}\Big)
%  \widetilde{\mathcal{H}}^{-4}  \widetilde{\mathcal{H}}^{\gamma\delta}  \widetilde X^{\sigma} \widetilde{\mathcal{T}}^{(\mathrm{ev})}_{\gamma\delta\sigma}{}^{\kappa}  \Big)
%\nonumber
%\\
%&&
%+  (O(1)\widetilde{ \mathcal{T}}^{(\mathrm{ev})})_{\alpha\beta\mu\nu} +   (O(1) %\widetilde\nabla\widetilde{\mathcal{T}}^{(\mathrm{ev})})_{\alpha\beta\mu\nu}
%\\
&=& 
3 \sqrt{\frac{\Lambda}{3}} |Y|^{-4} \Theta^{-1}\Big(
  (Y_iY_j)_{\mathrm{tf}}
  Y^kY^l   \widetilde\nabla_{t}{\mathcal{E}}_{kl} 
-  \frac{1}{2}|Y|^2
( Y^k Y_{(i}  \widetilde\nabla_{|t|} \mathcal{E}_{j)k}   )_{\mathrm{tf}}
\Big)
\nonumber
\\
&&
 - \Lambda  |Y|^{-4}  \Theta^{-2}\Big( (Y_iY_j)_{\mathrm{tf}}
  Y^kY^l
\mathcal{E}_{kl}
-  \frac{1}{2} |Y|^{2} 
(  Y^kY_{(i} \mathcal{E}_{j)k})_{\mathrm{tf}}
\Big)
\nonumber
\\
&&
+ i  \sqrt{\frac{\Lambda}{3}} |Y|^{-4} \Theta^{-1}\widehat\volform_{(i}{}^{lm} \Big( 
 \frac{1}{2}|Y|^2  Y^kY_{|l}
\widetilde\nabla_{m|} \mathcal{E}_{j)k}
-
 \frac{1}{2}|Y|^2
  Y^kY_{|l|}
 \widetilde\nabla_{j)} \mathcal{E}_{mk}
\nonumber
\\
&&
  - 3  Y_{j)}Y_l
    Y^kY^n
\widetilde\nabla_{m} \mathcal{E}_{kn}
    + |Y|^2Y_{j)}   Y^k
\widetilde\nabla_{m} \mathcal{E}_{kl}
\Big)
\nonumber
\\
&&
+   (O(1) \widetilde\nabla{\mathcal{E}})_{ij}
+  (O(\Theta^{-1}){ \mathcal{E}})_{ij} 
%\\
%&=& 
% (O(\Theta^{-2})\widetilde{ \mathcal{T}}^{(\mathrm{ev})})_{\alpha\beta\mu\nu} +   (O(\Theta^{-1}) \widetilde\nabla\widetilde{\mathcal{T}}^{(\mathrm{ev})})_{\alpha\beta\mu\nu}
\;.
\end{eqnarray}

\subsection{Uniqueness for the evolution equation for 
$\widetilde{\mathcal{T}}^{(\mathrm{ev})}_{\rho\sigma\mu}{}^{\nu}$ }

%\tim{again, some reorganizations concerning the organization of the sections}
%\marc{Additions up to Section 4.8}
Existence of solutions of quasilinear symmetric hyperbolic Fuchsian systems with
prescribed asymptotics at $\scri$
has been analyzed in the literature mainly in the analytic case. For the merely smooth
case, there exist results by Claudel \& Newman \cite{Claudel},
Rendall \cite{Rendall}, and more recently Ames {\it et.al.} \cite{Ames1,Ames}.
The results in these papers involve a number of algebraic requirements,
as well as global conditions in space. It is an interesting problem
to see whether any of these results can be adapted to our setting here,
in particular in order to prove a localized existence result in which
an appropriate singular behavior of
$\widetilde{\mathcal{T}}^{(\mathrm{ev})}_{\rho\sigma\mu}{}^{\nu}$ 
is prescribed on some domain $B$
of $\scri^{-}$ and existence and uniqueness of a corresponding solution is shown in the domain of dependence of $B$. 
This would also require studying the impact of the constraint equations and their preservation under evolution.

For the purposes of this section, where we aim to show that 
the necessary conditions listed in items 
(i) and (ii) of Theorem \ref{thm_nec_cond} for the vanishing of the rescaled
MST tensor 
$\widetilde{\mathcal{T}}^{(\mathrm{ev})}_{\rho\sigma\mu}{}^{\nu}$ 
in a neighborhood of $\scri$  are also sufficient, we merely need
a localized
uniqueness theorem for the evolution system with trivial initial data.

We state and prove such a result in a more general context  by adapting some 
of the ideas in \cite{Ames1}.  Then we show that this result applies to 
the evolution system satisfied by 
$\widetilde{\mathcal{T}}^{(\mathrm{ev})}_{\rho\sigma\mu}{}^{\nu}$.

\subsubsection{A localized uniqueness theorem for symmetric hyperbolic Fuchsian systems}
\label{sec_uniqueness}

%\jose{I have rewritten everything for the case of complex unknowns. I hope I have not introduced any mistake or problem}
Let $(\mcM,\tilde{g})$ be an $(n+1)$-dimensional spacetime and $\scri$ a spacelike
hypersurface. Choose coordinates in a neighborhood of $\scri$
so that $\scri = \{ t=0 \}$,
and the metric is such that $\tilde{g}_{tt}=-1$ and 
$\tilde{g}_{ti}=0$, $i=1,\cdots, n$. Let us consider the
first order, homogenous
 linear symmetric hyperbolic system of  PDEs 
\begin{equation}
A^t \partial_t u + A^i \partial_i u + \frac{1}{t} N u = 0
\label{PDE}
\end{equation}
where $u : \mcM \mapsto \mathbb{C}^m$ is the unknown, $A^t$, $A^i$ and $N$ 
are $m \times m$ matrices which depend smoothly on the spacetime
coordinates $(t,x^i)$. 
We  assume that $\mathbb{C}^m$ is, at each spacetime point $(t,x)$,
endowed with a positive definite sesquilinear
product $\la u, v \ran$ such that the endomorphisms $A^{\mu}(t,x)$, for
$\mu = t, i$, are Hermitian with respect to this product
%\tim{notation concerning the overbar should be explained because it has different meanings in this paper}
\begin{equation*}
\la u, A^{\mu} v \ran = \la \overline A^{\mu} u, v \ran.
\end{equation*}
Define $N_0 (x):= N(t=0,x)$. Our
main assumption is that 
$A^t + N_0$ is strictly positive definite with respect
to $\la \, ,  \,\ran$ at every point $p \in \scri$. 
Since both the inner product and $N$
depend smoothly on $t$ and $x$,
the same holds for a sufficiently small spacetime neighborhood of any
point $p \in \scri$. The domain of dependence of (\ref{PDE}) is defined in the
usual way (namely, the standard definition in terms of 
causal curves, with causality at $T_p \mcM$ being defined
as $k \in T_p \mcM$ being
future timelike (causal) iff $k^t > 0$ and $A^\mu|_{p} k_{\mu}$ is negative 
definite (semidefinite), 
where $k_{\mu}$ is obtained from $k$ by lowering indices with
respect to $\widetilde{g}$).  We also make the assumption that $k = \partial_t$
is future timelike in this sense.

We want to prove that the PDE (\ref{PDE}) with
trivial initial data on a domain $B \subset \scri$ vanishes identically
in the domain of dependence of $B$, denoted by $D(B)$. More precisely

\begin{lemma}
\label{UniquenessLemma}
Let $B \subset \scri$ be a domain with compact closure. Let $u$ be a $C^1$
map $u :\mcM \mapsto \mathbb{C}^m$ which vanishes at $B$ and solves
(\ref{PDE}). Assume that $A^t$ and $A^t + N_0$ are positive definite, then
$u$ vanishes at $D(B)$.
\end{lemma}

\begin{proof}
The proof is adapted from the basic energy estimate Lemma 2.7 in \cite{Ames1}. Since
our setup is simpler, we can use a domain of dependence-type argument,
instead of a global-in-space argument as in \cite{Ames1,Ames}.
First note that since $u$ vanishes at $t=0$, and it is $C^1$,
\begin{equation*}
\partial_t u |_{t=0,x} = \lim_{t\rightarrow 0} \frac{u(t,x)}{t} := u_1
\end{equation*}
with $u_1: \scri \mapsto \mathbb{C}^m$ continuous. 
Taking the limit of (\ref{PDE}) as $t\rightarrow 0$
with $(0,x) \in B$ and using $u(0,x)=0$ it follows
\begin{equation*}
  (A^t  + N_0 ) u_1 =0 \quad \quad  \Longrightarrow \quad \quad u_1 =0
\;,
\end{equation*}
because $A^t + N_0$ is positive definite, and hence has trivial kernel.
Let us consider the real quantity
\begin{equation*}
\J^{\mu} := e^{-kt} \la \frac{u}{t}, A^{\mu} \frac{u}{t} \ran \;,
\quad \quad k \in \mathbb{R}
\;,
\end{equation*}
and consider its (coordinate) divergence. Since the product $\la \, ,  \, \ran$ depends on the spacetime point, we denote by
$\la\,, \, \ran_{_{\mu}}$ the bilinear form (at each spacetime point)
defined by
\begin{equation*}
\partial_{\mu} \la u,v \ran = \la \partial_{\mu} u, v \ran
+ \la u, \partial_{\mu} v \ran + \la u,v \ran_{\mu}, \quad
\quad \forall u,v \in \mathbb{R}^m.
\end{equation*}
It follows
\begin{align}
\partial_{\mu} \J^{\mu}  = &  - \frac{e^{- k t}}{t^2}  \Big ( k + \frac{2}{t} \Big )
\la u, A^t u \ran 
+ \frac{e^{-kt}}{t^2} \Big ( 2 \la  u, A^{\mu} \partial_{\mu} u \ran
+ \la u, (\partial_{\mu} A^{\mu} ) u \ran 
+ \la u, A^{\mu} u \ran_{\mu} \Big ) \nonumber \\
= &  
- \frac{2}{t^3} e^{-kt} \la u, (A^t + N ) u \ran
\nonumber \\
& + e^{-kt} \Big (
 \la \frac{u}{t}, \left ( - k A^t + \partial_{\mu} A^{\mu}
\right ) \frac{u}{t} \ran
+ \la \frac{u}{t}, A^{\mu} \frac{u}{t} \ran_{\mu} \Big )
\;,
\label{est}
\end{align}
where in the first equality
we have used that $A^{\mu}$ is Hermitian w.r.t.\ $\la \, , \, \ran $
and in the second equality we have used the fact that $u$ satisfies \eq{PDE}.  Consider now a domain $V \subset \M$ bounded by
three smooth hypersurfaces-with-boundary
$B \subset \scri$, $B_T \subset \{t=T\}$
and $\Sigma$, whose union is a compact topological hypersurface. 
Note that $B$ and $B_T$ are spacelike (i.e.\ their normal is timelike).
We choose $V$ so that $\Sigma$ is   achronal and that its outward normal
(defined as the normal one-form which contracted with any
outward directed vector is positive) is past causal. Consider the
domain $V_{\epsilon} = V \cap \{ t \geq \ep \}$ for
$\epsilon >0$ small enough. The boundary splits
as $\partial V_{\epsilon} = B_{\epsilon} \cup \Sigma_{\epsilon} \cup
B_T$, with obvious notation.

We integrate $\partial_{\mu} \J^{\mu}$ on $V_{\epsilon}$ with respect to the spacetime
volume form $\boldmath{\eta} = F dt dx$ and use the Gauss identity. Denote
by  $n$ an outward normal to $\partial V$, then
\begin{align*}
\int_V \left ( \partial_{\mu} \J^{\mu} \right ) F dt dx & =
\int_V \partial_{\mu} (\J^{\mu} F)  dt dx -
\int_V (\partial_{\mu} F) \J^{\mu} dt dx  \\
& = \int_{\partial V} \J^{\mu} n_{\mu} dS
- \int_V (\partial_{\mu} F ) \J^{\mu} dt dx
\;,
\end{align*}
where $dS$ is the induced volume form on $\partial V$ corresponding to the
choice of normal $n$. Note in particular that, as a vector,
$n^{\mu}$ points
{\it inwards} both on $B_{\epsilon}$ and on $B_T$, so they can be taken simply to be
$n = \partial_t$ on $B_{\epsilon}$ and $n = - \partial_t$ on $B_T$, i.e.
$\boldmath{n} = -dt$ on $B_{\epsilon}$ and $\boldmath{n} = dt$ on $B_T$.
Inserting (\ref{est}) and splitting the integral at the boundary in three
pieces yields
\begin{align*}
\int_{B_{T}} &e^{-k T} \la \frac{u}{T}, A^t \frac{u}{T} \ran
dS
-
\int_{B_{\epsilon}} e^{-k \epsilon} \la \frac{u}{\epsilon}, A^t \frac{u}{\epsilon}
\ran dS
 = \\
& - \int_{\Sigma_{\epsilon}} e^{-k t} \la \frac{u}{t}, (A^{\mu} n_{\mu}) \frac{u}{t} \ran
dS 
\\
&  - 
\int_V \frac{2}{t^3} e^{-kt} \la u, (A^t + N ) u \ran
\boldmath{\eta}  \\
& + \int_V
e^{-kt} \left [ \la \frac{u}{t}, \left ( - k A^t + (\partial_{\mu} F) A^{\mu}
+ \partial_{\mu} A^{\mu}
\right ) \frac{u}{t} \ran
+ \la \frac{u}{t}, A^{\mu} \frac{u}{t} \ran_{\mu} \right] 
\boldmath{\eta}  := I_{V_{\epsilon}}
\;.
\end{align*}
The matrix $A^{\mu} n_{\mu}$ on $\Sigma_{\epsilon}$ is positive semidefinite because
$n$ is past causal. We now choose $T$ small enough so that $A^t + N$
is positive definite on $V$ and $k$ large
enough so that the last term in $I_{V_{\epsilon}}$ is negative (recall
that $V$ has compact closure). Thus we have $I_{V_{\epsilon}} \leq 0$ and  in fact
strictly negative unless $u=0$. Thus
\begin{align*}
\int_{B_{T}} e^{-k T} \la \frac{u}{T}, A^t \frac{u}{T} \ran
dS
-
\int_{B_{\epsilon}} e^{-k \epsilon} \la \frac{u}{\epsilon}, A^t \frac{u}{\epsilon}
\ran dS \leq 0
\;.
\end{align*}
We now take the limit $\epsilon \rightarrow 0$ and use the fact
that $\frac{u}{\epsilon} \rightarrow u_1 =0 $ to conclude
\begin{align*}
\int_{B_{T}} e^{-k T} \la \frac{u}{T}, A^t \frac{u}{T} \ran
dS \leq 0
\;.
\end{align*}
Since the product $\la\, , \,  \ran$ is positive definite, it follows $u=0$
on $B_T$. As a~consequence $I_V=0$ which implies $u=0$ on $V$.
It is clear that  the domain of~dependence $D(B)$
can be exhausted by such $V$'s, so $u=0$ on $D(B)$ as claimed.
\qed
\end{proof}

%\subsection{Uniqueness for the evolution equation for  $\widetilde{\mathcal{T}}^{(\mathrm{ev})}_{\rho\sigma\mu}{}^{\nu}$ }
\subsubsection{Application to the Fuchsian system satisfied by $\widetilde{\mathcal{T}}^{(\mathrm{ev})}_{\rho\sigma\mu}{}^{\nu}$ }

%\jose{Many changes to match smoothly with other parts. Some rewriting too}
In this section we show that the symmetric hyperbolic evolution system (\ref{evol3}) for 
${\cal S}_{\rho\sigma\mu}{}^{\nu}=\widetilde{\mathcal{T}}^{(\mathrm{ev})}_{\rho\sigma\mu}{}^{\nu}$ satisfies all the
conditions of Lemma~\ref{UniquenessLemma} in the unphysical space-time and conclude that the unique solution 
with vanishing data at $\scri^{-}$ is trivial and hence, 
since the system is linear and homogeneous, we also get uniqueness of all solutions given regular initial data at $\scri^{-}$.

We choose coordinates $\{t,x^{i}\}$ on a neighborhood of $\scri^{-}$ satisfying 
$\tilde{g}_{tt}=-1$, $\tilde{g}_{ti}=0$, and $\scri^{-} = \{ t = 0\}$. 
%The algebraic properties of 
%$\widetilde{\mathcal{T}}^{(\mathrm{ev})}_{\rho\sigma\mu}{}^{\nu}$ 
%imply that the symmetric tensor
%$\widetilde{\mathcal{T}}^{(\mathrm{ev})}_{titj} := {\mathcal E}_{ij}$
%contains the same information as
%$\widetilde{\mathcal{T}}^{(\mathrm{ev})}_{\rho\sigma\mu}{}^{\nu}$. 
%${\mathcal E}_{ij}$ is symmetric
%and trace-free with respect to 
The induced metrics on the hypersurfaces $\Sigma_{t}$ of constant $x^{0}=t$ are denoted by $h^t$, with corresponding volume forms $\widehat\volform^{t}$.
%Consider the component $\alpha = t$, $\beta = i$, $\mu=t$ of equation
%(\ref{unphys_ev}), which reads
%\begin{equation}
%-\partial_t \widetilde{\mathcal{T}}^{(\mathrm{ev})}_{tijt}
%+ \partial_l \widetilde{\mathcal{T}}^{(\mathrm{ev})}_{tij}{}^{l}
%= \,  
%  \mathcal{J}(\widetilde{\mathcal{T}}^{(\mathrm{ev})})_{tit} 
%+ \mathcal{\hat{J}}(\widetilde{\mathcal{T}}^{(\mathrm{ev})})_{tit} 
%\;,
%\label{coord}
%\end{equation}
%where all the Christoffel symbol terms have been absorbed into 
%$\mathcal{\hat{J}}(\widetilde{\mathcal{T}}^{(\mathrm{ev})})_{tit}$.
%Self-duality 
%of $\widetilde{\mathcal{T}}^{(\mathrm{ev})}_{\rho\sigma\mu}{}^{\nu}$ implies
%$\widetilde{\mathcal{T}}^{(\mathrm{ev})}_{tij}{}^{l} = - i 
%\boldmath{\hat{\eta}^{t}}{}_{kj}{}^l
%{\mathcal E}_{i}{}^{k}$
%where $\boldmath{\hat{\eta}^t}$ is the volume form of $h^t$. Inserting into
%(\ref{coord}) and taking the symmetrization in $i,j$ yields
Rewriting \eq{ev_eqns2} by moving all the Christoffel symbol terms to the right-hand side
and using \eq{exp_theta}
 we have 
\begin{align}
\partial_{t} {\cal E}_{ij}+i \widehat\volform_{(j}{}^{lk}\partial_{l}{\cal E}_{i)k} 
%\partial_t {\mathcal E}_{ij} 
%+ i \boldmath{\hat{\eta}^t}{}_{k (j}{}^{l} \partial_{l} {\mathcal E}_{i)}{}^{k}
%=  \mathcal{J}(\widetilde{\mathcal{T}}^{(\mathrm{ev})})_{tit} 
%+ \mathcal{\hat{J}}(\widetilde{\mathcal{T}}^{(\mathrm{ev})})_{tit}  \nonumber \\
 = & -\frac{1}{t} 
\frac{1}{|Y|^4} \Big ( 3 Y_i Y_j Y^k 
- \frac{1}{2} h_{ij} |Y|^2 Y^k - \frac{3}{2} |Y|^2 Y_{(i} \delta^k_{j)}
\Big ) Y^l {\mathcal E}_{lk} 
\nonumber \\
&  + (O(1) {\mathcal E})_{ij} \label{complex}
\;.
\end{align}
%
%where we have used  the fact that the conformal factor $\Theta$ satisfies $\Theta = \sqrt{\frac{\Lambda}{3}} x^{0} + O(1)$. 
The unkown is the complex symmetric and trace-free tensor ${\mathcal E}_{ij}$ introduced in \eq{e-m}. 
%which we can write as ${\mathcal E}_{ij} = E_{ij} + i B_{ij}$. 
The system (\ref{complex}) is of the form (\ref{PDE}) with 
$u= \{ {\cal E}_{ij} \} \in \mathbb{C}^{5}$,  
$A^t = \mbox{Id}_{5}$ and 
\begin{align*}
A^{l}{}_{ij}^{nk}=i\volform^{tl(k}{}_{(j}\delta^{n)}_{i)}.
%A^l ({\cal E})_{ij} = i \widehat\volform^{t}{}_{lk(j}\widetilde\partial^{l}{\cal E}_{i)}^{k} 
%\Big ( - \boldmath{\hat{\eta}^{t}}{}_{k(j}{}^l B_{i)}{}^k,
%\boldmath{\hat{\eta}^{t}}{}_{k(j}{}^l E_{i)}{}^k \Big)\;.
 \end{align*}
Take the sesquilinear product $\la \, , \, \ran$ defined by
$\langle {\cal E}, \widehat{\cal E} \rangle ={\cal E}_{ij} \widehat{\overline{\cal E}}{}^{ij}$
(indices lowered and raised with $h^t$), which is obviously positive definite --its norm leading to the density \eq{s-e2}. It is straightforward 
to check that $A^{\mu}$ is Hermitian with respect to this product and $A^t$
is obviously positive definite.
%\marc{I would have expected that the notion spacetime causality agrees with the PDE symmetric hyperbolic causality in this case, but I do not get that. Maybe I am making some mistake but I cannot find it...
%\\ -- \\
%jose: I am still puzzled about this...
%\\ -- \\
%tim: comment added at the end of the section
%}

It remains to check that $A^t + N_0 = \mbox{Id}_{5}+N_{0}$
is positive definite. The endomorphism $N_0$ is, from (\ref{complex}),
%$N_0 (E,B) = (N_1 (E), N_1(B))$ with
\begin{align*}
N_0({\cal E})_{ij} := \frac{1}{|Y|^4} \Big ( 3 Y_i Y_j Y^k 
- \frac{1}{2} h_{ij} |Y|^2 Y^k - \frac{3}{2} |Y|^2 Y_{(i} \delta^k_{j)}
\Big ) Y^l {\cal E}_{lk}
\end{align*}
so that
$$
\la {\cal E}, (\mbox{Id}_{5}+N_{0}){\cal E} \ran = {\cal E}_{ij} \overline{\cal E}^{ij}+\frac{3}{|Y|^4}\overline{\cal E}^{ij}Y_{i}Y_{j}{\cal E}_{kl}Y^{k}Y^{l}-
\frac{3}{2}\frac{1}{|Y|^{2}}Y_{i}\overline{\cal E}^{ik}Y^l {\cal E}_{lk}.
$$
To see if this has a sign we introduce the following objects orthogonal to $Y^{i}$
\begin{align*}
c_{ij} &:= {\cal E}_{ij}-\frac{2}{|Y|^{2}} Y_{(i}c_{j)}-\frac{1}{|Y|^{4}}Y_{i}Y_{j}{\cal E}_{kl}Y^{k}Y^{l}, \\
c_{i}&:= Y^{k}{\cal E}_{ki}-\frac{1}{|Y|^{2}} Y_{i} {\cal E}_{kl}Y^{k}Y^{l}
\end{align*}
and the previous expression can be rewritten as
$$
\la {\cal E}, (\mbox{Id}_{5}+N_{0}){\cal E} \ran = c_{ij}\overline{c}^{ij} +\frac{1}{2}\frac{1}{|Y|^{2}}c_{i}\overline{c}^{i} +\frac{5}{2}\frac{1}{|Y|^4}{\cal E}_{kl}Y^{k}Y^{l}
\overline{\cal E}_{ij}Y^{i}Y^{j}
$$
which is manifestly positive definite.
%We need to check that $\mbox{Id} + N_1$ is positive definite. Consider
%an orthonormal basis $\{ e_i \} = \{ e_1, e_A\}$ at $p \in \scri$
%with $e_1 = \frac{1}{|Y|} Y$ and decompose $E_{ij}$ in this basis
%as
%\begin{align*}
%E = c \, e_1 \otimes e_1 + c^A ( e_1 \otimes e_A + e_A \otimes e_1 )
%+ c^{AB} e_A \otimes e_B
%\end{align*}
%with $c^{AB}$ symmetric and satisfying the trace-free condition 
%$c + \delta_{AB} c^{AB}$.
%\tim{I guess this should be zero}
% A simple calculation gives (capital
%Latin indices are raised with $\delta^{AB}$) \marc{Please check \\ -- \\ tim: I agree}
%\begin{align*}
%\langle E, E \rangle + \langle E, N_1 (E) \rangle =
%\frac{5}{2} c^2 + \frac{1}{2} c_A c^A + c_{AB} c^{AB},
%\end{align*}
%which indeed is positive definite. 
We have thus proven
\begin{Lemma}
\label{UniquenessT}
Let $(\mcM,\widetilde{g})$ be a spacetime admitting 
a smooth conformal compactification. If 
$(\mcM,\widetilde{g})$ admits a Killing vector field
for which the rescaled MST
$\widetilde{\mathcal{T}}^{(\mathrm{ev})}_{\mu\nu\sigma}{}^{\rho}$ vanishes
at $\scri^{-}$, then 
$\widetilde{\mathcal{T}}^{(\mathrm{ev})}_{\mu\nu\sigma}{}^{\rho}$ vanishes
in a neighborhood of $\scri^{-}$.
\end{Lemma}

%\begin{lemma}
%\label{lemma_evolution2}
% The rescaled MST $\widetilde{\mathcal{T}}^{(\mathrm{ev})}_{\mu\nu\sigma}{}^{\rho}$, with $Q=Q_{\mathrm{ev}}$, satisfies a linear homogeneous system  of wave  equations which is regular near $\scri$.
%\end{lemma}
%
%\begin{remark}
%{\rm
%It follows from the trace of the symmetric hyperbolic system on $\scri$ that
%$\widetilde\nabla_0\widetilde{\mathcal{T}}^{(\mathrm{ev})}_{\mu\nu\sigma}{}^{\r%ho}|_{\scri}$ vanishes
%supposing that $\widetilde{\mathcal{T}}^{(\mathrm{ev})}_{\mu\nu\sigma}{}^{\rho}%|_{\scri}$ does.
%More general, the data for $\widetilde\nabla_0\widetilde{\mathcal{T}}^{(\mathrm%{ev})}_{\mu\nu\sigma}{}^{\rho}|_{\scri}$ 
%follow from those for $\widetilde{\mathcal{T}}^{(\mathrm{ev})}_{\mu\nu\sigma}{}^{\rho}%|_{\scri}$, provided the latter are finite. \marc{``finiteness'' added here}
%}
%\end{remark}

The characteristics of the symmetric hyperbolic system  \eq{complex} coincide with those of the
 propagational part of the Bianchi equation, and are computed and discussed in \cite[Section 4]{F1}.
It is shown there that they form null \emph{and timelike} hypersurfaces.

\subsection{Main result}

We end up with the following main result:
%\tim{formulate also as a Cauchy problem?}
%
\begin{theorem}
\label{first_main_thm}
Consider a spacetime $(\mcM,g)$, or rather its conformally rescaled counterpart,  in wave map gauge \eq{gauge_conditions_compact},%
\footnote{
\label{footnote_gauge}
In fact, it suffices if $\widetilde R$ and $\widetilde W^{\sigma}$, including certain transverse derivatives thereof, vanish on $\scri^-$. Moreover, a corresponding result
must also hold for non-vanishing gauge source functions. We leave it to the reader to work
this out. 
}
solution to Einstein's vacuum field equations with  $\Lambda>0$, which admits a smooth conformal extension through $\scri^-$ and
which contains  a KVF $X$. 
Denote  by $h$ the Riemannian metric induced by $\widetilde g=\Theta^2 g$  on $\scri^-$, and  by $Y$ the CKVF induced by $X$ on $\scri^-$. 
%Assume that  $|Y|^2 > 0$.
%\tim{...$Y$ is non-trivial}

Then, there exists a function $Q$, namely $Q=Q_0$ ($=Q_{\mathrm{ev}}$ for an appropriate choice of the $\sigma$-constant), for which the MST $\mathcal{S}_{\mu\nu\sigma}{}^{\rho}$
 corresponding to $X$ vanishes in
the domain of dependence of $\scri^-$
%some neighborhood  of  $\scri^-$
 if and only if
the following relations hold: 
\begin{enumerate}
\item[(i)] $ \widehat C_{ij} = \CconstL|Y|^{-5}(Y_iY_j -\frac{1}{3}|Y|^2 h_{ij})$ for some constant $\Cconst$, where $\widehat C_{ij}$ is the Cotton-York tensor of the Riemannian  3-manifold $(\scri^-, h)$, and
\item[(ii)]    $D_{ij} = \widetilde d_{titj}|_{\scri^-} =\Dconst  |Y|^{-5} (Y_iY_j -\frac{1}{3}|Y|^2 h_{ij})$ for some constant $\Dconst$, where $\widetilde d_{\mu\nu\sigma}{}^{\rho}$ is the rescaled
Weyl tensor of the unphysical spacetime  $(\widetilde{\mcM\enspace}\hspace{-0.5em} ,\widetilde g)$.
\end{enumerate}
\end{theorem}
\begin{proof}
 Theorem~\ref{thm_nec_cond} shows that (i) and (ii) are necessary conditions. Conversely, if (i) and (ii) hold, it follows from Theorem~\ref{prop_Qs} and Theorem~\ref{thm_nec_cond} that there exists a choice of the %complex 
$\sigma$-constant $a$
 in $Q_{\mathrm{ev}}$ for which the rescaled MST $\widetilde{\mathcal{T}}^{(\mathrm{ev})}_{\alpha\beta\mu\nu}$ vanishes on $\scri^-$. It then follows from Lemma~\ref{UniquenessT}
 that it vanishes in 
%spacetime wherever the equations are regular.
%This will be the case at least in 
some neighborhood of $\scri^-$.
%\tim{added... I  hope one can argue this way}
However, once we know that the MST vanishes in such neighborhood, it follows from the results
 in \cite{mars_senovilla} that the metric has to take one of the local forms given there. But for all these metrics the MST vanishes
globally. So we can conclude that it vanishes in fact in the whole domain of dependence of $\scri^-$.
\qed
\end{proof}

%\begin{remark}
%{\rm
%\tim{added... I  hope one can argue this way}
%However, once we know that the MST vanishes in some neighborhood of $\scri$, it follows from the results
%results in \cite{mars_senovilla} that the metric has to take one of the local forms given there. But for all these metrics the MST %vanishes
%globally. So we can conclude that it vanishes in fact in the whole domain of dependence of $\scri$.
%}
%\end{remark}

Recall the Definition~\ref{KdS_like} of Kerr-de Sitter-like spacetimes.  Lemma~\ref{UniquenessT} and Theorem~\ref{first_main_thm}
  lead to the following
characterization of KdS-like space-times:
\begin{corollary}
\label{cor_charact_KdS}
Let  $(\mcM,g)$
 be a $\Lambda>0$-vacuum spacetime 
admitting smooth  conformal compactification and corresponding null
infinity $\scri$ as well as a KVF $X$. Then the following statements are equivalent:
\begin{enumerate}
\item[(i)]  $(\mcM,g)$ is Kerr-de Sitter-like at a connected component $\scri^{-}$.
\item[(ii)] There exists a function $Q$ such that the  MST associated to $X$ vanishes in the domain of dependence of $\scri^{-}$.
\item[(iii)]  The CKVF $Y$ induced by $X$ on $\scri^-$ satisfies the conditions (i) and (ii) of Theorem~\ref{first_main_thm}.
\end{enumerate}
\end{corollary}

Reformulated in terms of an asymptotic Cauchy problem Theorem~\ref{first_main_thm} 
becomes Theorem~\ref{first_main_thm2}  given in the Introduction.

\section{A second conformal Killing vector field}
\label{sec_CKVFs}

\subsection{Existence of a second  conformal Killing vector field}
\label{second_KVF}
%\tim{several changes and additions from now on}
It follows from  \cite[Theorem 4]{mars_senovilla} that
(here overbars mean ``complex conjugate of'')
\begin{eqnarray}
\varsigma^{\mu} \,=\, \frac{4}{|Q\mathcal{F}^2-4\Lambda|^2}X^{\sigma}\ol{\mathcal{F}}_{\sigma}{}^{\rho}\mathcal{F}_{\rho}{}^{\mu} +\mathrm{Re}\Big(\frac{\mathcal{F}^2}{(Q\mathcal{F}^2-4\Lambda)^2}\Big) X^{\mu}
\label{varsigma}
\end{eqnarray}
is another KVF which commutes with $X$,
supposing that $(\mcM,g)$ is a $\Lambda$-vacuum spacetime for which $Q$, $\mathcal{F}^2$ and $Q\mathcal{F}^2 -4\Lambda$ are not identically zero and
 whose MST vanishes w.r.t.\ the KVF $X$ (cf.\ \cite{mars, perjes} for the $\Lambda=0$-case). 
Note that $\varsigma$ may be trivial or merely a multiple of $X$.

However, expression (\ref{varsigma}) can be taken as definition of a vector field $\varsigma$
in any $\Lambda$-vacuum spacetime admitting a KVF $X$ and such that 
${\mathcal F}^2 \neq 0$ and $Q \mathcal{F}^2 - 4 \Lambda \neq 0$.
We take $Q=Q_0$ and assume  that $(\mcM,g)$ admits a smooth $\scri$,
but we do \emph{not} assume that the MST vanishes.

%\hspace{5mm}
%
%{\bf Questions: 
%
%\hspace{5mm}
%
%1. Expressions (5.2)-(5.8) are fully general? Do they use that
%the MST  vanishes?
%
%\hspace{5mm}
%
%2. Should the Q in all these expressions be $Q_0$?
%
%\hspace{5mm}
%
%3. If the computations are fully general, one should not say after (5.8) 
%``As the results in [8] predict''. This is a much more general statement,
%
%\hspace{5mm}
%
%4. I have been not able to show that the expression in Definition 5.1 is CKV
%even assuming that (5.9) holds. I have not tried very hard, but
%spent a few hours and failed...
%
%}
%
%\hspace{5mm}

A somewhat lengthy computation reveals that
 (recall that $f$ and $N$ denote divergence and curl of $Y$, respectively)
%\tim{(5.5) corrected}
%
\begin{eqnarray}
\varsigma^{t}
&=& -\frac{2}{|Q_0\mathcal{F}^2-4\Lambda|^2}\Theta^4h^{ij}\ol\sigma_{j}\mathcal{F}_{i t} +\mathrm{Re}\Big(\frac{\mathcal{F}^2}{(Q_0\mathcal{F}^2-4\Lambda)^2}\Big) \widetilde X^{t}
\nonumber
\\
&=& -\frac{1}{8\Lambda^2}\Theta^4h^{ij}\ol\sigma_{j}\mathcal{F}_{i t} +\frac{1}{16\Lambda^2}\mathrm{Re}(\mathcal{F}^2) \widetilde X^{t}
 +O(\Theta)
\nonumber
\\
&=&O(\Theta)
\;,
\\
\varsigma^{i}
&=&\frac{2}{|Q_0\mathcal{F}^2-4\Lambda|^2} \Theta^4h^{ij}\Big( h^{ kl}\ol \sigma_l\mathcal{F}_{k j} 
+\ol \sigma_t\mathcal{F}_{ jt} \Big)
+\mathrm{Re}\Big(\frac{\mathcal{F}^2}{(Q_0\mathcal{F}^2-4\Lambda)^2}\Big)\widetilde X^i
\nonumber
\\
&=&\frac{1}{8\Lambda^2} \Theta^4h^{ij}\Big( h^{ kl}\ol \sigma_l\mathcal{F}_{k j} 
+\ol \sigma_t\mathcal{F}_{jt} \Big)
+\frac{1}{16\Lambda^2}\mathrm{Re}(\mathcal{F}^2)\widetilde X^i +O(\Theta)
%\\
%&=& -\frac{1}{4\Lambda^2}|Y|^2 \nabla_t\nabla_t\widetilde X^i  -  \frac{1}{8\Lambda^2} \widehat R |Y|^2Y^i
%\\
%&&
%  +\frac{1}{16\Lambda^2} Y_k N^k  N^i
% +\frac{1}{16\Lambda^2}\widehat  \nabla^k |Y|^2 \widehat\volform_{kn}{}^iN^n
%\\
%&&
% -\frac{1}{8\Lambda^2}  \widehat\nabla^k(Y_jN^j) \widehat \volform_{kl}{}^i Y^l
%\\
%&&
%  -\frac{1}{4\Lambda^2}Y^i Y^j\ol{\widetilde\nabla_t\widetilde\nabla_t\widetilde X_j}
% +\frac{1}{36\Lambda^2}Y^i    f^2 
%-\frac{1}{12\Lambda}c Y^i  
% +O(\Theta)
%\nonumber
%\\
%&=&  
%\frac{1}{6\Lambda^2}(Y^i  Y^j\widehat  \nabla_{j}f -|Y|^2 \widehat \nabla^i f)
%+ \frac{1}{12\Lambda^2}f\widehat\volform^{ijk}Y_jN_k
%\nonumber
%\\
%&&
%  +\frac{1}{8\Lambda^2} Y_k N^k  N^i
% +O(\Theta)
%\label{second_CKVF}
%\;,
\\
&=&
\frac{1}{8\Lambda^2}Y_k N^kN^i    +\frac{1}{8\Lambda^2} Y^i \Big(  
-\frac{1}{2}|N|^2
 +\frac{2}{9}   f^2  -2 \widehat c\Big)
\nonumber
\\
&&
 -\frac{1}{2\Lambda^2} |Y|^2 \Big(\widehat L^i{}_k Y^k  + \frac{1}{3}\widehat \nabla^if \Big)
+  \frac{1}{12\Lambda^2} f\widehat\volform^{ikl}Y_kN_l
 +O(\Theta)
\\
&=&
\frac{1}{8\Lambda^2}Y_k N^k  N^i   
+\frac{1}{2\Lambda^2} Y^i \Big(  Y^kY^l\widehat L_{kl}+ 
   \frac{1}{3}Y^k\widehat\nabla_k f \Big)
\nonumber
\\
&&
 -\frac{1}{2\Lambda^2} |Y|^2 \Big(\widehat L^i{}_k Y^k  + \frac{1}{3}\widehat \nabla^if \Big)
+ 
 \frac{1}{12\Lambda^2} f\widehat\volform^{ikl}Y_kN_l
 +O(\Theta)
\label{second_CKVF}
\;,
\end{eqnarray}
where we used \eq{expansion_QF_Lambda}, \eq{expansion_sigma_t}, \eq{deriv_YN}, \eq{expansion_sigma_i} as well as the following relations:
\begin{eqnarray}
\mathrm{Re}(\mathcal{F}^2) &=& -\frac{4}{3}\Lambda|Y|^2  \Theta^{-2}
 -4Y^i\ol{\widetilde\nabla_t\widetilde\nabla_t\widetilde X_i}
 +\frac{4}{9}   f^2 
 -\frac{4}{3}\Lambda c
+ O(\Theta)
\;,
\\
\mathcal{F}_{it}
% &=& \Theta^{-3}(\widetilde H_{it} + i \widetilde H_{it}^*) + \Theta^{-2} (\breve F_{it} + iF^*_{it})
%\\
 &=& \sqrt{\frac{\Lambda}{3}}Y_i\Theta^{-3}  -  \frac{i}{2}N_i  \Theta^{-2}
-\frac{1}{2}  \sqrt{\frac{3}{\Lambda}} \Theta^{-1}\Big(\ol{\widetilde \nabla_t\widetilde \nabla_t\widetilde X_i}+ \frac{\widehat R}{2} Y_i\Big)
+O(1)
\;,
\phantom{xx}
\\
\mathcal{F}_{ij} &=& i\sqrt{\frac{\Lambda}{3}}\widehat \volform_{ijk} Y^k \Theta^{-3}  +\frac{1}{2}\widehat\volform_{ijk}N^k\Theta^{-2} +O(\Theta^{-1})
%\;,
%\\
%\sigma_t &=&
%  2\sqrt{\frac{\Lambda}{3}} |Y|^2\Theta^{-3}   - i  Y_i N^i \Theta^{-2}
%-  \frac{1}{2}\sqrt{\frac{3}{\Lambda}} \widehat R |Y|^2 \Theta^{-1}
% + O(1)
%\;,
%\\
%\sigma_i &=& 
%  -  \widehat  \nabla_{i}|Y|^2 \Theta^{-2} +  i\sqrt{\frac{3}{\Lambda}} \widehat\nabla_i(Y_kN^k) \Theta^{-1}
%+O(1)
\;,
\\
\widetilde X^t &=& \frac{1}{3} \sqrt{\frac{3}{\Lambda}} f\Theta + O(\Theta^3)
\;,
\\
\widetilde X^i &=& Y^i + \frac{1}{2} \frac{3}{\Lambda} \ol{\widetilde\nabla_t\widetilde\nabla_t \widetilde X^i} \Theta^2 + O(\Theta^3)
\;.
\end{eqnarray}
We conclude that the vector field $\varsigma$ is always tangential to $\scri$, and not just in the setting where the MST
vanishes.

Proceeding in the same manner as we did for the functions $\widehat c$ and $\widehat k$, we regard \eq{second_CKVF} as the definition of a vector field 
on some Riemannian 3-manifold:

\begin{definition}
\label{Defivarsigma}
%\jose{Equality with $18\Lambda^2\varsigma^i|_{\scri} $ removed from this def. Mentioned later}
Let $(\Sigma,h)$ be a Riemannian 3-dimensional manifold which admits a CKVF $Y$. Then we  define the vector field
\begin{eqnarray}
 \widehat\varsigma^i(Y)
%\,:=\, 18\Lambda^2\varsigma^i|_{\scri} 
%&:=&
%3(Y^i  Y^j\widehat  \nabla_{j}f -|Y|^2 \widehat \nabla^i f)
%+ \frac{3}{2}f\widehat \nabla^i|Y|^2
%  +\frac{9}{4} Y_k N^k  N^i
% - f^2Y^i
%\\
&:=&
\frac{9}{4}Y_k N^k  N^i   
+9 Y^i \Big(  Y^kY^l\widehat L_{kl}+ 
   \frac{1}{3}Y^k\widehat\nabla_k f \Big)
\nonumber
\\
&&
 - 9 |Y|^2 \Big(\widehat L^i{}_k Y^k  + \frac{1}{3}\widehat \nabla^if \Big)
+ 
 \frac{3}{2} f\widehat\volform^{ikl}Y_kN_l
\;.\phantom{xx}
\label{dfn_varsigma}
\end{eqnarray}
\end{definition}

As in the case of $\widehat c$ and $\widehat k$, one can use a spacetime argument  where $\Sigma$ is embedded as ``null infinity'' 
into a $\Lambda>0$-vacuum spacetime with vanishing MST to prove that $\widehat \varsigma(Y)$,
%\jose{added} 
which in that case reads $ \widehat\varsigma^i(Y)=18\Lambda^2\varsigma^i|_{\scri} $, is a (possibly trivial) CKVF which commutes with $Y$, supposing that
$|Y|^2>0$ and
%\jose{repetition avoided} 
(\ref{condition_on_C}) hold.
%\tim{formulate as a lemma?... no if one can do it more general}
%
%\begin{equation}
%\widehat C_{ij} \,=\, B|Y|^{-5}(Y_iY_j)_{\mathrm{tf}}
%\;.
%\label{condition_C-Y2}
%\end{equation}
Irrespective of that, one can raise the question under which conditions $\widehat\varsigma(Y)$ is a CKVF and under which conditions
it commutes with $Y$.

For this purpose let us compute the covariant derivative of $\widehat\varsigma$. A  lengthy calculation which uses
\eq{deriv_YN}, \eq{HessPsi}, \eq{HessY}, the conformal Killing equation for $Y$ as well as the relation
\begin{equation}
\widehat \nabla_i N_j \,=\,\frac{2}{3}\widehat\volform_{ij}{}^k\widehat\nabla_k f - \widehat\volform_j{}^{kl}\widehat R_{klim}Y^m
%\\
%\frac{1}{3}(\widehat\nabla_{i}\widehat\nabla_j f)_{\mathrm{tf}} &=& (Y^l\widehat\nabla^k\widehat R_{k(ij)l})_{\mathrm{tf}} - \widehat R_{(i}{}^l\volform_{j)l}{}^kN_k
%-\frac{2}{3}(\widehat R_{ij})_{\mathrm{tf}}f- (Y^k\widehat\nabla_{(i}\widehat R_{j)k})_{\mathrm{tf}}
%\\
%&=& -Y^l \widehat\nabla_l(\widehat  L_{ij})_{\mathrm{tf}}- \widehat L_{(i}{}^l\widehat \volform_{j)l}{}^k\widehat  N_k
%-\frac{2}{3}(\widehat L_{ij})_{\mathrm{tf}}f
\end{equation}
gives
\begin{eqnarray}
\widehat\nabla_i\widehat\varsigma_j
&=&
3\widehat\nabla_iY_j Y^k\widehat  \nabla_{k}f 
+ 3Y_j \widehat\nabla_iY^k\widehat  \nabla_{k}f 
+ 3Y^kY_j \widehat\nabla_i\widehat  \nabla_{k}f 
- 6 Y^k \widehat\nabla_i Y_k \widehat \nabla_j f
\nonumber
\\
&&
+ 3 Y^k \widehat \nabla_jY_k\widehat\nabla_if 
- 3|Y|^2 \widehat\nabla_i\widehat \nabla_j f
+ 3f\widehat\nabla_i  Y^k \widehat \nabla_jY_k
+ 3f  Y^k \widehat\nabla_i\widehat \nabla_jY_k
\nonumber
\\
&&
  +\frac{9}{4} N_j\widehat\nabla_i( Y_k N^k  )
  +\frac{9}{4} Y_k N^k  \widehat\nabla_i N_j
 - 2f  Y_j \widehat\nabla_i f
 - f^2\widehat\nabla_i Y_j
\nonumber
\\
&&
 + 9\widehat\nabla_iY_jY^kY^l\widehat L_{kl}
 + 18Y_jY^k\widehat\nabla_i Y^l\widehat L_{kl} 
 + 9Y_jY^kY^l\widehat\nabla_i\widehat L_{kl}
\nonumber
\\
&&
- 18\widehat\nabla_iY_l Y^kY^l\widehat L_{jk}
- 9|Y|^2\widehat\nabla_i  Y^k\widehat L_{jk}
- 9|Y|^2 Y^k\widehat\nabla_i\widehat L_{jk}
\\
&=&
\frac{3}{4}f|N|^2   h_{ij} 
 -\frac{9}{2}\underbrace{Y_{[(i} N_k\widehat  \nabla_{l]}f \widehat  \volform_{j)}{}^{kl}}_{=\frac{1}{3} h_{ij}\widehat \volform^{klm}Y_kN_l\widehat\nabla_m f}
- \frac{27}{2} \underbrace{ Y^m Y_{[j} N_k \widehat L_{l]m} \widehat \volform_{i}{}^{kl}  }_{=\frac{1}{3} h_{ij}Y^mY_pN_k\widehat L_{lm}\widehat \volform^{pkl}}
\nonumber
\\
&&+ 
\frac{3}{2}\widehat \volform_{ijl}\Big( N^l Y^k\widehat  \nabla_{k}f   + Y_k N^k \widehat\nabla^l f  +3 Y_m N^m  Y^k \widehat L_{k}{}^l 
 +3 N^l Y^kY^m\widehat L_{km} 
\nonumber
\\
&&
 - \frac{1}{3}f^2N^l \Big) 
+  \frac{3}{2} \widehat \volform_{[i}{}^{kl}\Big(
 N_{j]} Y_k\widehat \nabla_lf  
-  Y_{j]} N_k\widehat  \nabla_{l}f  
-3  Y_k  N_l\widehat\nabla_{j]}f 
\Big)
\nonumber
\\
&&
-2 fY_{[i} \widehat  \nabla_{j]} f 
 - 6fY^kY_{[i} \widehat L_{j]k}
+ 9 Y^m N_{l} Y_k\widehat L_{m[i}\widehat \volform_{j]}{}^{kl}  
\nonumber
\\
&&
 %+ 18Y_jY^kY^l\widehat\nabla_{[i}\widehat L_{l]k} - 18|Y|^2 Y^k\widehat\nabla_{[i}\widehat L_{k]j}
 + 9Y_j\widehat \volform_{ik}{}^{l}\widehat C_{lp}Y^kY^p - 9|Y|^2\widehat \volform_{ik}{}^l\widehat C_{lj} Y^k
\\
&&
+ \frac{9}{2}\underbrace{\Big( 3Y^m N_{[k} Y_l\widehat L_{m]i}\widehat \volform_{j}{}^{kl}  
- Y_jN_k Y_l\widehat L_{mi} \widehat \volform^{mkl}
\Big)}_{=0}
\;.
\label{deriv_varsigma}
\end{eqnarray}
Employing again the fact that $Y$ is a CKVF one shows that $Y$ and $\widehat\varsigma$ commute,
\begin{eqnarray}
[Y,\widehat\varsigma]_i &=& Y^j\widehat\nabla_j\widehat\varsigma_i - \widehat\varsigma^j\widehat\nabla_jY_i
\\
&=& Y^j\widehat\nabla_j\widehat\varsigma_i + \frac{1}{2}\widehat \volform_{ijk}\widehat\varsigma^jN^k-\frac{1}{3} f\widehat\varsigma_i
\\
&=&
\frac{3}{4}f|N|^2  Y_i
 -\frac{9}{4} Y_i\widehat \volform^{klm}Y_kN_l\widehat\nabla_m f
- \frac{9}{2} Y_iY^mY_pN_k\widehat L_{lm}\widehat \volform^{pkl}
\nonumber
\\
&&- 
\frac{1}{2}Y^j\widehat\volform_{ijl}\Big( -\frac{3}{2} N^l Y^k\widehat  \nabla_{k}f -f^2N^l   +\frac{9}{2} Y_k N^k \widehat\nabla^l f  +9 Y_m N^m  Y^k \widehat L_{k}{}^l 
  \Big) 
\nonumber
\\
&&
+  \frac{3}{4}\widehat \volform_{i}{}^{kl} |Y|^2 N_k\widehat  \nabla_{l}f  
 + 3fY^jY^kY_{i} \widehat L_{jk}
 - 3f|Y|^2Y^k \widehat L_{ik}
\nonumber
\\
&&
+ fY^jY_{i} \widehat  \nabla_{j} f 
- f |Y|^2 \widehat  \nabla_{i} f 
+ \frac{1}{2}\widehat \volform_{ijk}\widehat\varsigma^jN^k
-\frac{1}{3} f\widehat\varsigma_i
\\
&=&
\frac{9}{4}\Big( 3\widehat \volform_{i}{}^{jk}Y^l Y_{[j} N_k \widehat  \nabla_{l]}f   - Y_i\widehat \volform^{klm}Y_kN_l\widehat\nabla_m f\Big)
\nonumber
\\
&&
 + \frac{9}{2}\Big(3\widehat \volform_{i}{}^{jk}Y^m  N_{[j}  Y_m \widehat L_{k]l} Y^l
- Y_iY^mY_jN_k\widehat L_{lm}\widehat \volform^{jkl}\Big)
\\
&=&
0
\end{eqnarray}

\begin{lemma}
Let $(\Sigma,h)$ be a Riemannian 3-manifold which admits a CKVF~ $Y$. Let the vector field $\widehat\varsigma(Y)$ be given by \eq{dfn_varsigma}.
Then $$[Y,\widehat\varsigma]=0\;.$$
\end{lemma}

We further deduce from \eq{deriv_varsigma} that
\begin{equation}
(\widehat\nabla_{(i}\widehat\varsigma_{j)})_{\mathrm{tf}} \,=\,
  9Y_{(i}\widehat \volform_{j)k}{}^{l}\widehat C_{lp}Y^kY^p - 9|Y|^2\widehat \volform_{(i|k|}{}^l\widehat C_{j)l} Y^k
\;,
\end{equation}
i.e.\ $\widehat\varsigma(Y)$ will be a (possibly trivial) CKVF if and only if
\begin{eqnarray}
Y_{(i}\widehat \volform_{j)k}{}^{l}\widehat C_{lp}Y^kY^p  \,=\, |Y|^2\widehat \volform_{(i|k|}{}^l\widehat C_{j)l} Y^k
\\
\Longleftrightarrow \quad (\widehat C_{lp} Y_{(i} -Y_p\widehat C_{l(i} ) \widehat \volform_{j)k}{}^{l}Y^kY^p
 \,=\, 0
\label{CKVF_condition}
\;.
\end{eqnarray}

%Contraction with $\widehat C^{jm}Y_m$ yields
%%
%\begin{equation}
%Z^jY_{j}\widehat \volform_{ik}{}^{l} Z_lY^k  \,=\, 
%|Y|^2Z^j\widehat \volform_{ik}{}^l\widehat C_{jl} Y^k
%+ |Y|^2Z^j\widehat \volform_{jk}{}^l\widehat C_{il} Y^k
%\;.
%\end{equation}
%%
%\begin{equation}
%0  \,=\, 
%Z^iZ^j\widehat \volform_{ik}{}^l\widehat C_{jl} Y^k
%\;.
%\end{equation}

\begin{lemma}
Let $(\Sigma,h)$ be a Riemannian 3-manifold which admits a CKVF $Y$. Then the vector field $\widehat\varsigma(Y)$,  defined in \eq{dfn_varsigma}, 
is a (possibly trivial) CKVF if and only if \eq{CKVF_condition} holds.
\end{lemma}

\begin{remark}
{\rm
In particular  \eq{CKVF_condition}  is fulfilled  supposing that  \eq{condition_on_C} holds as one should expect from the results in
\cite{mars_senovilla}.
}
\end{remark}

\begin{remark}
%\jose{Added. It can be removed if useless \\ -- \\ marc: I'd keep it}
{\rm
Observe that, from (\ref{cotton-york}), condition \eq{CKVF_condition} can be re-expressed as
$$
Y_p Y^k \big( Y^p \widehat C_{(ij)k} - Y_{(i} \widehat C^p{}_{j)k}\big)=0
\;.
$$
}
\end{remark}

%As the results in \cite{mars_senovilla} predict, $\varsigma$ is tangential to $\scri$. 
%Moreover, $\varsigma^i|_{\scri}$ is a (possibly not independent)  CKVF on $(\scri, h_{ij})$ which commutes with $Y$.

\subsection{Properties of the KID equations}
\label{app_KID_properties}

In this section we study the case where the  KID equations on a spacelike $\scri^-$ of a $\Lambda>0$-vacuum spacetime  admit at least two solutions, as it is the case for  e.g.\
Kerr-de Sitter, or, more generally, for spacetimes with vanishing MST \cite{mars_senovilla}.

%\marc{name $\varsigma$ changed to avoid confusion with  the previous definition}

Recall the  KID equations \cite{ttp2}
\begin{eqnarray}
   \mcL_Y D_{ij} + \frac{1}{3}D_{ij}\widehat \nabla_k Y^k &=&0
\label{reduced_KID2}
%\\
%\,[X,DZ ]-D [X,Z] +  \mathrm{div}XDZ &=& 0  \enspace \forall\,Z
\;.
\end{eqnarray}
%
%for a conformal Killing vector field $X^i$.
Consider  two  CKVFs $Y$ and $\zeta$ on the Riemannian 3-manifold  $(\scri^-,h)$
  which  both assumed to  solve the
 KID equations. Then also their commutator, which is obviously
a CKVF, provides another (possibly trivial) solution
of the  KID equations,
\begin{equation}
\mcL_{[Y,\zeta]} D_{ij} + \frac{1}{3}D_{ij}\widehat\nabla_k[Y,\zeta]^k  \,=\, 0
\;.
\end{equation}
This reflects the well-known fact that KVFs together with the commutator  form a Lie algebra.
% of two KVFs is again a KVF.

%KID equation for $\varsigma$, contracted with $Y^j$, set $P_j=Y^kD_{jk}$
%\begin{eqnarray*}
%0&=&  \varsigma^k\widehat\nabla_k P_i -   \varsigma^kD_{i}{}^j\widehat\nabla_k  Y_j  + Y^jD_{i}{}^k\widehat\nabla_{j}\varsigma_k 
%+ P^k\widehat\nabla_{i}\varsigma_k 
%+ \frac{1}{3}P_i\widehat \nabla_k \varsigma^k
%\end{eqnarray*}
%
%Contraction with $Y^i$ and using that $Y_{[i}P_{j]}=0$, conformal Killing equation
%\begin{eqnarray*}
%0&=&  \varsigma^k\widehat\nabla_k(Y^i P_i )
%+ \frac{1}{3}Y^iP_i\widehat \nabla_k \varsigma^k
%\end{eqnarray*}

Let us continue assuming that $Y$ and $\zeta$ are two CKVFs  on $(\scri^-,h)$ which solve the KID equations,
and let us further assume that $D_{ij}$ satisfies
%\jose{repetition avoided}
 condition (\ref{condition_on_D})
%\tim{under weaker assumptions possible?... I don't think so}
%
%\begin{equation}
%D_{ij}\,=\, C |Y|^{-5}(Y_iY_j)_{\mathrm{tf}}
%\end{equation}
%%
(in particular, we assume $|Y|^2>0$).
Then the KID equations \eq{reduced_KID2} for $\zeta$ can be written as 
%\marc{$Z$ changed to $V$}
 (set $V:=[Y,\zeta]$ and assume $\Dconst \ne 0 $)
\begin{equation}
2Y_{(i} V_{j)} + 
 (h_{ij}-5|Y|^{-2}Y_iY_j)Y_kV^k
\,=\, 0
\;.
\label{eqn_Y_Z}
\end{equation}
Contraction with $Y^j$ yields
\begin{equation}
  0\,=\,  |Y|^2 V_{i} -3Y_i Y_jV^j
\;.
\label{contraction}
\end{equation}
Another contraction   with $Y^i$ gives
\begin{equation}
   Y^kV_{k}  \,=\,0
\;.
\end{equation}
Inserting this  into \eq{contraction} we find that \eq{eqn_Y_Z} is equivalent to
\begin{equation}
V\,=\, [Y,\zeta] \,=\, 0
\;.
\end{equation}

We have proven the following: 
%\marc{Added the condition $C\neq 0$ \\ -- \\ tim: what happens for $C=0$?}
%\tim{there was a mistake in an earlier version}
%\tim{second KVF?}
%
\begin{lemma}
\label{lemma_second_KVF}
Let $(\scri^-,h_{ij})$ be a Riemannian 3-manifold 
%with constant non-zero scalar curvature $\widehat R=\mathrm{const.}\ne 0$
which admits a CKVF $Y$ with $|Y|^2> 0$.
Denote by $(\mcM,g_{\mu\nu})$ the $\Lambda>0$-vacuum spacetime constructed from the initial data $h_{ij}$ and $D_{ij}= \Dconst|Y|^{-5}(Y_iY_j -\frac{1}{3}|Y|^2 h_{ij})$,
$ \Dconst\neq 0$.
Then any other %(conformal Killing) 
vector field on $(\scri^-,h_{ij})$ extends to a KVF of $(\mcM,g_{\mu\nu})$ if and only if
it is a CKVF of $(\scri^-,h_{ij})$ which  commutes with $Y$.
\end{lemma}

\begin{remark}
{\rm
Note that the unphysical Killing equations imply that a KVF in the physical spacetime can be non-trivial if and only if
the induced CKVF on $\scri$ is non-trivial (compare \cite{ttp2}).
}
\end{remark}

\begin{remark}
{\rm
Assume that $\scri$ is conformally flat. 
%\marc{Is this obvious? Do we need a reference? \\ -- \\ tim:  follows from the results in paper II. Apart from this: The Lie algebra of the CKFVs of Euclidean 3-space is $so(1,4)$  and I hope it is known that, given an element, there always exists another one which commutes with it...? \\ -- \\
%marc: Shall we write ``Then it can be shown that there exists...''? 
%\\ -- \\ tim: done}
Then  it can be shown that there exists at least one independent CKVF $\zeta$ which commutes 
with $Y$. 
It then follows from Lemma~\ref{lemma_second_KVF} that  the emerging spacetime admits at least two KVFs.
This provides a simple proof that $\Lambda>0$-vacuum spacetimes with vanishing MST and conformally flat $\scri$ have at least two KVFs (cf.\ \cite[Theorem 4]{mars_senovilla}).
}
\end{remark}

Moreover, we have the following
\begin{proposition}
\label{prop_2CKVF}
Let $(\Sigma, h_{ij})$ be a Riemannian 3-manifold which admits a CKVF $Y$ with $|Y|^2>0$.
Assume further that its Cotton-York tensor satisfies $\widehat C_{ij}=C|Y|^{-5}(Y_iY_j -\frac{1}{3}|Y|^2 h_{ij})$, $C=\mathrm{const}$.
Then $(\Sigma, h_{ij})$ admits a second, independent CKVF $\zeta$ which commutes with $Y$.
\end{proposition}
\begin{proof}
One more time we use a spacetime argument: 
There exists a $\Lambda>0$-vacuum spacetime $(\mcM, g_{\mu\nu})$ with a KVF $X$ such that the associated MST vanishes,
such that $(\Sigma, h_{ij})$ can be identified with past null infinity, and such that $X^i|_{\Sigma}=Y^i$.
It follows from the classification results in \cite{mars_senovilla} that $(\mcM, g_{\mu\nu})$ admits a second independent KVF 
which commutes with $X$, and which induces a CKVF on $\scri^-$ with the asserted properties.
\qed
\end{proof}

\begin{remark}
{\rm
The second CKVF  $\zeta$ may or may not be
$\varsigma$ as given in Definition \ref{Defivarsigma}. 
The statement of the Proposition is that,
even when $\varsigma$ happens to be linearly dependent to $Y$, there
is still another independent CKVF on 
$(\Sigma,h_{ij})$.
}
\end{remark}

Proposition~\ref{prop_2CKVF} might be useful to classify Riemannian 3-manifolds which admit a CKVF which is related
to the Cotton-York tensor via \eq{condition_C}.

%It follows that the  conformal Killing vector  fields  $Y$ on the round 3-sphere to generate KdS
%(which, as shown above, need to be proper conformal Killing vector fields)
% are those for which there exists precisely one independent  Killing vector field $X$ which commutes with $Y$.
%\tim{does commuting initially imply commuting in spacetime?}

\vspace{1.2em}
\noindent {\textbf {Acknowledgements}}
MM acknowledges financial support under the projects  FIS2012-30926,
FIS2015-65140-P (Spanish MINECO-fondos FEDER)
and P09-FQM-4496 (J. Andaluc\'{\i}a---FEDER).
TTP acknowledges financial support  by the Austrian Science Fund (FWF): P 24170-N16.
JMMS is supported under grant FIS2014-57956-P (Spanish MINECO-fondos FEDER), GIU12/15 (Gobierno Vasco), UFI 11/55 (UPV/EHU) and by project 
P09-FQM-4496 (J. Andaluc\'{\i}a---FEDER).
The research of WS was funded by the Austrian Science Fund (FWF): P 23337-N16.

\end{document}